\documentclass{scrartcl}
 
 \pdfoutput=1
 \KOMAoptions{paper=letter}
 
\usepackage{geometry}
\usepackage{fix-cm}
	\usepackage[T1]{fontenc}
\usepackage{xspace}
\usepackage{amsmath}
\usepackage{amsthm}
\usepackage{amsfonts}
\usepackage{amssymb}
\usepackage[hidelinks]{hyperref} 

\usepackage{nicefrac}
\usepackage{tikz}

\usetikzlibrary{positioning,arrows,shapes,decorations,calc,fit,arrows.meta} 
\usepackage{colortbl}
\usepackage{booktabs}
\usepackage{cleveref}
\usepackage{array}
\usepackage{graphicx}
\usepackage{xcolor}
\usepackage{enumitem}
\usepackage{tabularx}
\usepackage{thm-restate}
\usepackage{ragged2e}
\usepackage{microtype}
\usepackage{cleveref}
\usepackage{mathrsfs} 
\usepackage{caption}
 \usepackage{array}
 \usepackage{graphicx}
 \usepackage{xcolor}
 \usepackage{enumitem}
 \usepackage{tabularx}
 \usepackage{thm-restate}
 \usepackage{ragged2e}
 \usepackage{todonotes} 
 \presetkeys{todonotes}{inline}{}
\usepackage[round]{natbib}
\usepackage{pdfpages}
 
 \usepackage{multirow}
 \newcounter{remark}

 \theoremstyle{definition}

 \theoremstyle{plain}
 
 \newtheorem{lemma}{Lemma}

 \usepackage{pgfplots}
 \usepackage{pgfplotstable}
 \usepgfplotslibrary{fillbetween}
 \pgfplotsset{compat=1.7}
 \usepackage{caption}
 \usepackage{subcaption}

 \newcolumntype{L}[1]{>{\raggedright\let\newline\\\arraybackslash\hspace{0pt}}m{#1}}
 \newcolumntype{C}[1]{>{\centering\let\newline\\\arraybackslash\hspace{0pt}}m{#1}}
 \newcolumntype{R}[1]{>{\raggedleft\let\newline\\\arraybackslash\hspace{0pt}}m{#1}}

 \newcommand{\plot}[3]{\pgfplotstableread[col sep=comma,]{#1}\datatable
 	\begin{tikzpicture}
 		\begin{axis}[
 			width=5.8cm,
 			height=3.8cm,
 			x tick label style={font=\small, /pgf/number format/.cd,
 				precision=1,
 				/tikz/.cd},
 			xtick=data,
 			xticklabels from table={\datatable}{n},
 			legend style={at={(1.6,1)},anchor=north east},
 			y tick label style={font=\small,
 			}
 			]    
 			\addplot [mark=x, orange, opacity=0.8 ] table [x expr=\coordindex, y={RandomDictatorship}]{\datatable};
 			\addlegendentry{$RD$};
 			
 			\addplot [name path = C1ML, mark=x, blue!80, opacity=0.8] table [x expr=\coordindex, y={C1ML}]{\datatable};
 			\addlegendentry{$C1ML$};
 			
 			\addplot [name path = C2ML, mark=x, red!95!black, opacity=0.8] table [x expr=\coordindex, y={C2ML}]{\datatable};
 			\addlegendentry{$C2ML$};

 			\addplot [name path = CRWW, mark=x, green!50!black, opacity=0.8] table [x expr=\coordindex, y={CRWW}]{\datatable};
 			\addlegendentry{$CRWW$};

 			\addplot [name path = PluralityVeto, mark=x, purple!60!black, opacity=0.8 ] table [x expr=\coordindex, y={PluralityVeto}]{\datatable};
 			\addlegendentry{$RPV$};
 			#3
 		\end{axis}
 		\node at(2.125,-0.6) {\small #2};
 	\end{tikzpicture}
 }

 \title{\LARGE The Metric Distortion of Randomized Social Choice Functions:
 	\LARGE C1 Maximal Lottery Rules and Simulations}
 
 \date{}
 
 \author{
 	\large Fabian Frank, TU Munich, fabian\_w.frank@tum.de\\
 	\large Patrick Lederer, UNSW Sydney, p.lederer@unsw.edu.au
 }

 \newcommand{\shin}{{\in}}

\begin{document}
 	
 	\maketitle
 	
 	\begin{abstract}
 		The metric distortion of a randomized social choice function (RSCF) quantifies its worst-case approximation ratio to the optimal social cost when the voters' costs for alternatives are given by distances in a metric space. This notion has recently attracted significant attention as numerous RSCFs that aim to minimize the metric distortion have been suggested. Since such tailored voting rules have, however, little normative appeal other than their low metric distortion, we will study the metric distortion of well-established RSCFs. Specifically, we first show that C1 maximal lottery rules, a well-known class of RSCFs, have a metric distortion of $4$, which is optimal within the class of majoritarian RSCFs. Secondly, we conduct extensive computer experiments on the metric distortion of RSCFs to obtain insights into their average-case performance. These computer experiments are based on a new linear program for computing the metric distortion of a lottery and reveal that the average-case metric distortion of some classical RSCFs is often only slightly worse than that of RSCFs tailored to minimize the metric distortion. 
 		Finally, we also analytically study the expected metric distortion of RSCFs for the impartial culture distribution.  Specifically, we show that, under this distribution, every reasonable RSCF has an expected metric distortion close to $2$ when the number of voters is large. 
 	\end{abstract}

\section{Introduction}
An important challenge in multi-agent systems is collective decision-making: given the possibly conflicting preferences of a group of agents over some alternatives, a joint decision has to be made. To address this problem, researchers in the field of social choice theory try to identify desirable mechanisms to aggregate the agents' preferences. In more detail, social choice theory is mainly concerned with \emph{social choice functions (SCFs)} and \emph{randomized social choice functions (RSCFs)}, which formalize deterministic and randomized voting rules: an SCF maps the voters' preferences (expressed as linear rankings of the alternatives) to a single winner, and an RSCF returns a probability distribution over the alternatives from which the final winner will eventually be chosen \citep{ASS11a,BCE+15a}.

In an attempt to quantitatively measure the quality of SCFs and RSCFs, \citet{PrRo06a} introduced the distortion of voting rules. The idea of this notion is that voters have latent cardinal utilities over the alternatives and that voting rules should try to select alternatives with high social welfare. However, (R)SCFs do not have access to the voters' utilities, and the \emph{distortion} of a voting rule thus quantifies the worst-case ratio between the (expected) social welfare of the selected alternative and that of the optimal alternative. A prominent variant of this problem has been suggested by \citet{ABP15a}: in the metric distortion setting, voters and alternatives are located in a metric space and the distance between a voter and an alternative specifies the cost incurred to a voter when an alternative is elected. Voting rules should then try to select an alternative with low social cost but, since voters only report ordinal preferences, they can only approximate the optimal social cost. The \emph{metric distortion} of an SCF (resp. RSCF) is hence the worst-case ratio between the (expected) social cost of the selected alternative and of the optimal alternative, where the worst-case is taken over all preference profiles and all metric spaces that are consistent with the given profile. 

The metric distortion of SCFs and RSCFs has recently attained significant attention \citep[see, e.g.,][]{AFSV21b}. In particular, after \citet{ABP15a} and \citet{AnPo17a} have shown that no SCF (resp. RSCF) has a metric distortion of less than $3$ (resp. $2$), numerous authors tried to find voting rules with minimal metric distortion \citep[e.g.,][]{ABE+18a,Kemp20a,KiKe22a,KiKe23a,CRWW23a}. However, many of the suggested voting rules are specifically tailored to minimize the metric distortion and have otherwise little normative appeal. 
For example, the recently proposed Plurality-Veto rule \citep{KiKe22a} is not even anonymous and its latest variant called Simultaneous-Veto \citep{KiKe23a} fails Pareto-optimality. We thus find it noteworthy that some well-established RSCFs also have a low metric distortion. For instance, the uniform random dictatorship and C2 maximal lottery (C2ML) rules, two of the most prominent RSCFs in the literature, both have a metric distortion of $3$ \citep{FFG16a,AnPo17a,CRWW23a}. Since such established RSCFs satisfy numerous desirable properties, we will study their metric distortion in more detail, even though voting rules with lower metric distortion are known.

\paragraph{Our Contribution.} The goal of this paper is to enhance the understanding of the metric distortion of established RSCFs. We will contribute to this end in three ways.
Firstly, we investigate the metric distortion of C1 maximal lottery (C1ML) rules, a class of RSCFs that is well-known for being robust to small changes in the voters' preferences 
\citep{LLL93b,Hoan17a,BBS20a}. C1ML rules intuitively choose randomized Condorcet winners: these rules return a lottery $p$ such that, for every lottery $q$, it is at least as likely that a majority of the voters prefers an outcome drawn from $p$ to an outcome drawn from $q$ than vice versa. As our first result, we show that every C1ML rule has a metric distortion of at most $4$ and give a lower bound on the metric distortion of all majoritarian RSCFs (which only depend on the majority relation) that converges to $4$ as the number of alternatives increases. Since C1ML rules are majoritarian, this proves that they minimize the metric distortion within the class of majoritarian RSCFs when the number of alternatives is unbounded. 

Secondly, we conduct extensive computer experiments on the metric distortion of five RSCFs: the uniform random dictatorship, 
C1 and C2 maximal lottery rules, a randomized variant of the Plurality-Veto rule \citep{KiKe22a}, and the CRWW rules suggested by \citet{CRWW23a}, which have the 
best currently known metric distortion. In more detail, we sample preference profiles from numerous distributions, 
compute the lotteries chosen by our RSCFs, and then compute the worst-case metric distortion for the given lotteries and profiles. Moreover, we conduct an analogous experiment also with real-world data taken from PrefLib \citep{MaWa13a}. Our simulations show that the average metric distortion of all RSCFs is rather similar and significantly better than their worst-case guarantees. In particular, for many ``structured'' distributions 
C1ML and C2ML rules are only slightly worse than CRWW rules, which typically have the best metric distortion in our experiments. In light of their normative appeal, this gives a strong argument for using a C1ML or C2ML rule instead of an RSCF designed to minimize the metric distortion. 

Our computer experiments rely on a new linear program for computing the metric distortion of a lottery for a given profile, which we believe to be of independent interest. Specifically, our LP has only $\mathcal O(nm^2)$ constraints, where $n$ is the number of voters and $m$ the number of alternatives, and thus allows us to compute the metric distortion of a lottery even for large profiles. 

Finally, we complement our simulations with an analytical study of the expected metric distortion of RSCFs when preference profiles are sampled from the impartial culture distribution. For this setting, we show that the expected metric distortion of every reasonable RSCF converges to a value between $2$ and $2+\frac{1}{m-1}$ (where $m$ is the number of alternatives) when the number of voters goes to infinity. This result aligns with our simulations for the impartial culture model and shows that, at least under the simplistic impartial culture distribution, the choice of the voting rule has surprisingly little effect on the expected metric distortion.

\paragraph{Related Work.}

\begin{table}
        \centering
    \caption{Overview of the best known upper and lower bounds on the metric distortion in various classes of voting rules. Each row together with the labels ``RSCF'' and ``SCF'' determines a class of voting rules. The columns labeled ``LB'' and ``UB'' show the best known lower and upper bounds for the metric distortion of rules within the given class when there is an unbounded number of alternatives. The bold numbers are proven in this paper. 
    }
    \begin{tabular}{|c | c c | c c|}
        \hline
             & \multicolumn{2}{c |}{RSCF} & \multicolumn{2}{c|}{SCF}  \\
             & LB & UB & LB & UB \\
             \hline
        All & $2.112$ & $2.753$ & $\quad3\quad$ & $\quad3\quad$\\
        Tops-only & $\quad3\quad$ & $\quad3\quad$ & $\infty$ & $\infty$\\
        Pairwise & $3$  & $3$ & $3$ & $2+\sqrt{5}$\\
        Majoritarian & $\mathbf{4}$ & $\mathbf{4}$ & $5$ & $5$\\\hline
        \end{tabular}
    \label{tab:distortion}
\end{table}

We will next review the most relevant results on the metric distortion of voting rules and refer to the survey by \citet{AFSV21b} for more details. An overview of the upper and lower bounds for the met\-ric distortion of various classes of voting rules is given in \Cref{tab:distortion}.
The study of the metric distortion of deterministic SCFs was initiated by \citet{ABP15a} who have, e.g., shown that no SCF has a metric distortion of less than~$3$. Inspired by this work, numerous researchers tried to find voting rules with a metric distortion of $3$ \citep{GKM17a,SkEl17a,ABE+18a,AFP22b}, but it was only in a recent line of work that such SCFs have been designed \citep{MuWa19a,Kemp20a,GHS20a,KiKe22a,KiKe23a}. In particular, these works culminated in the Plurality-Veto rule, a simple SCF with a metric distortion of $3$ \citep{KiKe22a}.
Interestingly, \citet{KiKe23a} recently also aimed to design a normatively more appealing SCF with optimal metric distortion.

As an alternative approach to minimize the metric distortion, researchers also studied RSCFs. In particular, \citet{AnPo17a} have shown that no RSCF has a metric distortion of less than $2$ and that the uniform random dictatorship has a metric distortion of~$3$. Moreover, \citet{GAX17a} have proven that all tops-only RSCFs (i.e., RSCFs that can only access the voters' favorite alternatives) have a metric distortion of at least $3-\frac{2}{m}$ when there are $m$ alternatives. Similarly, \citet{CRWW23a} have shown that C2 maximal lottery rules have a metric distortion of $3$ and it is known that all pairwise RSCFs (i.e., RSCFs that only depend on the numbers of voters that prefer $x$ to $y$ for all alternatives $x,y$) have a metric distortion of at least $3-\frac{2}{m}$ \citep{GKM17a}. Thus, when the number of alternatives is unbounded, the uniform random dictatorship minimizes the metric distortion within the class of tops-only RSCFs and C2 maximal lottery rules within the class of pairwise RSCFs.
Moreover, several RSCFs have been designed with the goal to minimize the metric distortion \citep[][]{GAX17a,FGMP19a,GHS20a}, but none of them guarantees a metric distortion of less than $3$. It was hence only recently that both the upper and lower bound of the metric distortion of RSCFs has been improved: \citet{ChRa22a} have shown that every RSCF has a metric distortion of at least $2.112$ and \citet{CRWW23a} designed the CRWW rules with a metric distortion of at most~$2.753$.

Finally, our work is related to several papers \citep{CDK17a,CDK18a,GKP+23a,CaFe24a} that analyze the expected distortion of voting rules when the voters' utilities are drawn from a distribution. By contrast, we study the expected metric distortion of voting rules for the worst-case metrics of randomly drawn profiles, i.e., we consider realistic profiles without imposing any additional structure on the voters' utilities. Furthermore, we note that \citet{EFLS24a} suggested an improved linear program for computing the non-metric distortion of voting rules, which can be seen as a mathematically unrelated analog of our new linear program for computing the metric distortion of RSCFs.

\section{Model}\label{sec:model}

Let $V_n=\{v_1,\dots, v_n\}$ denote a finite set of $n\geq 1$ voters and $X_m=\{x_1,\dots, x_m\}$ a finite set of $m\geq 1$ alternatives. We suppose that every voter $v\in V_n$ reports a \emph{preference relation $\succ_v$}, which is formally a complete, transitive, and anti-symmetric binary relation over $X_m$. The set of all preference relations over $X_m$ is denoted by $\mathcal{R}(X_m)$. A \emph{preference profile $R$} is the collection of the preference relations of all voters in $V_n$
and the set of all preference profiles over an electorate $V_n$ and a set of alternatives $X_m$ is given by $\mathcal{R}(X_m)^{V_n}$. In this paper, we will allow for both varying sets of voters and alternatives. The set of all preference profiles is hence given by $\mathcal{R}^*=\bigcup_{n,m\in\mathbb{N}} R(X_m)^{V_n}$. Moreover, $\mathcal{R}^*_m$ is the set of all profiles on $m$ alternatives, i.e., $\mathcal{R}^*_m=\bigcup_{n\in \mathbb{N}} R(X_m)^{V_n}$. Given a profile $R$, we will denote by $V_R$ and $X_R$ the sets of voters and alternatives that are present in the profile $R$, and by $n_R$ and $m_R$ the sizes of these sets. 

Next, we introduce additional notation for preference profiles. In particular, we define $t_R(x)=|\{v\in N_R\colon \forall y\in X_R\setminus \{x\}\colon x\succ_v y\}|$ as the number of voters that top-rank alternative $x$ in the profile~$R$. Furthermore, we let the \emph{support} $n_{xy}(R)=|\{v\in V_R\colon x\succ_v y\}|$ for $x$ against $y$ denote the number of voters who prefer $x$ to $y$ in $R$. Finally, the \emph{majority relation $\succsim_R$} of a profile $R$ is defined by $x\succsim_R y$ if and only if $n_{xy}(R)\geq n_{yx}(R)$. That is, $x\succsim_R y$ if at least as many voters prefer $x$ to $y$ than vice versa. Following the literature, $\succ_R$ denotes the strict part of $\succsim_R$ (i.e., $x\succ_R y$ iff $x\succsim_R y$ and not $y\succsim_R x$) and $\sim_R$ the indifference part (i.e., $x\sim_R y$ iff $x\succsim_R y$ and $y\succsim_R x$).

\subsection{Randomized Social Choice Functions}\label{subsec:RSCF}

The study objects of this paper are randomized social choice functions which are voting rules that may use chance to determine the winner of the election. To formalize this, we define \emph{lotteries} as probability distributions over the set of alternatives $X_R$: a lottery is a function $p:X_R\rightarrow [0,1]$ such that $\sum_{x\in X_R} p(x)=1$. We furthermore denote by $\Delta(X_R)$ the set of all lotteries over $X_R$. A \emph{randomized social choice function (RSCF)} $f$ is then a function that maps every preference profile $R\in\mathcal{R}^*$ to a lottery $p\in \Delta(X_R)$. We denote by $f(R,x)$ the probability that $f$ assigns to alternative $x$ in the profile~$R$ and next introduce five (classes of) RSCFs:

\paragraph{Uniform random dictatorship.} The {uniform random dictatorship $f_\mathit{RD}$} picks a voter $v\in V_R$ uniformly at random and implements his favorite alternative as the winner of the election. More formally, the probability that an alternative $x$ is selected in a profile $R$ by the uniform random dictatorship is 
$f_\mathit{RD}(R,x)=\frac{t_R(x)}{n_R}$.

\paragraph{Randomized Plurality-Veto.} \citet{KiKe23a} suggested the Plurality-Veto rule as a deterministic SCF with the optimal metric distortion of $3$. For this rule, we first fix a sequence of the voters $(v_1,\dots, v_n)$ and assign a score $s(x)$ to each alternative that is initially equal to $t_R(x)$. Then, we iterate through the voters according to the given sequence, ask each voter for his worst alternative with positive score, and reduce the score of this alternative by $1$. Finally, the winner of this rule is the last alternative with positive score. Since the winner of Plurality-Veto rule depends on the order over the voters, we denote by $PV(R)$ the set of alternatives that can be chosen for some order. Furthermore, we define the \emph{randomized Plurality-Veto rule} $f_\mathit{RPV}$ as the RSCF that picks an alternative from $PV(R)$ uniformly at random. The set $PV(R)$ and hence $f_\mathit{RPV}$ can be efficiently computed by solving $m$ matching problems \cite{KiKe22a,KiKe23a}.

\paragraph{C2ML rules.} {C2 maximal lottery (C2ML) rules}, which have been suggested by \citet{Fish84a} and recently promoted by, e.g., \citet{Bran13a}, compute a randomized Condorcet winner: these rules select a lottery $p$ such that, for all lotteries~$q$, the expected number of voters that prefer the outcome chosen from $p$ to the outcome chosen from $q$ is at least as large as the expected number of voters that prefer the outcome chosen from $q$ to the outcome chosen from $p$. To formalize this, we extend the support $n_{xy}(R)$ to lotteries $p$, $q$ by defining $n_{pq}(R)=\sum_{x,y\in A}p(x)q(y)n_{xy}(R)$. Then, the set of C2 maximal lotteries is given by 
$\mathit{C2ML}(R)=\{{p\in \Delta(X_R)\colon}{\forall q\in \Delta(X_R)\colon} n_{pq}(R)\geq n_{qp}(R)\}$.
The set of C2 maximal lotteries is always non-empty by the minimax theorem and almost always a singleton \citep{LLL97a,LeBr05a}. Finally, an RSCF is a C2ML rule if $f(R)\in\mathit{C2ML}(R)$ for every profile $R\in\mathcal{R}^*$.

\paragraph{C1ML rules.} C1 maximal lottery (C1ML) rules, which go back to \citet{Fish84a}, also choose a randomized Condorcet winner but in a different sense: C1ML rules select a lottery $p$ such that, for all lotteries $q$, it is at least as likely that a majority of the voters prefers the outcome chosen from $p$ to the outcome chosen from $q$ than vice versa. To formalize this, we extend the majority relation to lotteries $p$, $q$ by defining that $p\succsim_R q$ if and only if $\sum_{x,y\in A\colon x\succ_R y} p(x)q(y)\geq \sum_{x,y\in A\colon x\succ_R y} p(y)q(x)$. The set of C1 maximal lotteries is then 
$\mathit{C1ML}(R)=\{p\in\Delta(X_R)\colon \forall q\in \Delta(X_R)\colon p\succsim_R q\}$. Just as for C2 maximal lotteries, this set is always non-empty and almost always a singleton. In particular, if the number of voters is odd, there are unique C1 and C2 maximal lotteries. An RSCF is a {C1ML rule} if $f(R)\in \mathit{C1ML}(R)$ for all profiles $R\in\mathcal{R}^*$. 

\paragraph{CRWW rules.} Finally, we introduce the RSCFs suggested by \citet{CRWW23a}, which we refer to as CRWW rules. As a subroutine, these rules rely on another RSCF called $f_{\beta-\mathit{radius}}$. To define this RSCF, we say $x$ $\beta$-covers $y$ in a profile $R$ for some $\beta\in [0,1]$ if $n_{xy}(R)\geq \beta n_R$ and $n_{zx}(R)\geq \beta n_R$ implies $n_{zy}(R)\geq \beta n_R$ for all $z\in X_R$. Moreover, we define $U_\beta(R)$ as the set of alternatives that are not $\beta$-covered in $R$ and $R|_{U_\beta(R)}$ as the profile that arises from $R$ by removing all alternatives not in $U_\beta(R)$. Then, $f_{\beta-\mathit{radius}}$ computes the uniform random dictatorship on $R|_{U_\beta(R)}$, i.e., $f_{\beta-\mathit{radius}}(R)=f_\mathit{RD}(R|_{U_\beta(R)})$. Based on this subroutine, constants $B=0.876353$, $p=\frac{1}{1+\int_{0.5}^{B}\frac{1}{1-x^2} dx}\approx 0.552327$, and the distribution $\rho(\beta)=\frac{p}{(1-p)(1-\beta^2)}$ on the interval $(\frac{1}{2},B)$, CRWW rules are defined as follows: with probability $p$, we execute a C2ML rule and with probability ${1-p}$, we sample a value $\beta\in (0.5, B)$ from the distribution $\rho(\beta)$ and return $f_{\beta-\mathit{radius}}(R)$. Hence, an RSCF $f$ is a CRWW rule if there is a C2ML rule $f'$ such that
\[f(R)=pf'(R)+{(1-p)}\int_{0.5}^B \rho(\beta) f_{\beta-\mathit{radius}}(R)d\beta.\]\smallskip

The uniform random dictatorship $f_\mathit{RD}$, C2ML rules, and C1ML rules are well-established in the literature. For example, $f_\mathit{RD}$ is known to be strategyproof \citep{Gibb77a}, whereas both C2ML rules and C1ML rules satisfy, e.g., Condorcet-consistency and composition-consistency \citep{Bran13a}. By contrast, the randomized Plurality-Veto rule and the CRWW rules are designed to minimize the metric distortion and only known to satisfy basic further axioms. Moreover, we note that the uniform random dictatorship $f_\mathit{RD}$, C2ML rules, and C1ML rules belong to important classes of RSCFs: $f_\mathit{RD}$ is a \emph{tops-only} RSCF as it only accesses the voters' favorite alternatives, C2ML rules are \emph{pairwise} as they only access the supports $n_{xy}(R)$ for all $x,y\in X_R$, and C1ML rules are \emph{majoritarian} as they only depend on the majority relation $\succsim_R$. In more detail, an RSCF $f$ is majoritarian if $f(R)=f(R')$ for all profiles $R,R'\in\mathcal{R}^*$ with ${\succsim_R}={\succsim_{R'}}$. 
    
\subsection{Metric Distortion}\label{subsec:metdist}

In order to assess the quality of RSCFs, we analyze their metric distortion. The idea of this approach is that voters and alternatives are embedded in a metric space and that the distance between a voter $v$ and an alternative $x$ specifies the cost that $v$ experiences when $x$ is selected. Following the utilitarian approach, the optimal alternative is then the one that minimizes the total distance to all voters. However, since voters only report ordinal preferences over the alternatives instead of their cardinal costs, we cannot determine the best alternative. The goal of metric distortion is hence to select a lottery that approximates the optimal social cost well for every metric space that is consistent with the voters' preferences.

To formalize this, we call a function $d:(V_R\cup X_R)^2\rightarrow \mathbb{R}_{\geq 0}$ a \emph{metric} if it satisfies for all $x,y,z\in V_R\cup X_R$ that \emph{i)} ${d(x,x)=0}$, \emph{ii)} $d(x,y)=d(y,x)$, and \emph{iii)} $d(x,z)\leq d(x,y)+d(y,z)$. We note that some definitions of metrics also require that $d(x,y)>0$ if $x\neq y$, but the literature on metric distortion typically omits this condition since it does not affect the results. The distance $d(v,x)$ states the cost incurred to voter $v$ when alternative $x$ is selected. The \emph{social cost} of an alternative $x$ is thus $sc(x,d)=\sum_{v\in V_R} d(v,x)$ and the social cost of lottery $p$ is $sc(p,d)=\sum_{x\in X_R} p(x) sc(x,d)$. Finally, a metric~$d$ is \emph{consistent} with a profile $R$ if $x\succ_v y$ implies $d(v,x)\leq d(v,y)$ for all voters $v\in V_R$ and alternatives $x,y\in X_R$. We denote by $D(R)$ the set of metrics that are consistent with~$R$. 

Given a profile $R$, the goal of metric distortion is to find a lottery whose social cost is close to the optimal social cost for all metrics that are consistent with~$R$. We thus define the metric distortion of a lottery $p$ in a profile $R$ as $\mathit{dist}(p,R)=\sup_{d\in D(R)} \frac{sc(p,d)}{\min_{x\in X_R} sc(x,d)}$. Note that $\min_{x\in X_R} sc(x,d)$ might be $0$; we hence define $\frac{0}{0}=1$ and $\frac{z}{0}=\infty$ for $z>0$. For the ease of presentation, we will use in our results that $\infty>x$ for all $x\in \mathbb{R}$ and $y+z\cdot \infty=\infty$ for all $y\in\mathbb{R}$, $z\in \mathbb{R}_{>0}$.
Next, the \emph{metric distortion} $\mathit{dist}(f)$ of an RSCF~$f$ is its worst-case metric distortion over all possible profiles, i.e., $\mathit{dist}(f)=\sup_{R\in\mathcal{R}^*} \mathit{dist}(f(R), R)$. To allow for a more fine-grained analysis, we further define $\mathit{dist}_m(f)=\sup_{R\in\mathcal{R}_m^*} \mathit{dist}(f(R), R)$ as the metric distortion of $f$ when only profiles on $m$ alternatives are considered. We note that $\mathit{dist}(f)=\infty$ and $\mathit{dist}_m(f)=\infty$ if the respective suprema are unbounded.

We recall here that the uniform random dictatorship $f_\mathit{RD}$, the randomized Plurality-Veto rule $f_\mathit{RPV}$, C2ML rules $f_\mathit{C2ML}$, and CRWW rules $f_\mathit{CRWW}$ have a metric distortion of $\mathit{dist}(f_\mathit{RD})=3$, $\mathit{dist}(f_\mathit{RPV})=3$, $\mathit{dist}(f_\mathit{C2ML})=3$, and $\mathit{dist}(f_\mathit{CRWW})\leq2.753$, respectively. By contrast, the metric distortion of C1ML rules is unknown.

\section{Analysis of C1ML Rules}\label{subsec:C1ML}

As our first contribution, we will show that C1ML rules have a metric distortion of $4$ and that no other majoritarian RSCF has a lower metric distortion when the number of alternatives is unbounded. Thus, our results show that C1ML rules minimize the metric distortion among majoritarian RSCFs. Our analysis of the C1ML rule is further motivated by the fact that maximal lottery rules have been repeatedly recommended for practical usage \citep{Bran13a,BBS20a}. 
All missing proofs can be found in the appendix. 

To prove our results, we first show a strong relation between the metric distortion of majoritarian RSCFs and distances in the majority relation. To this end, we define the \emph{majority distance $\mathit{md}(x,y,{\succsim_R})$} as the length of the shortest path from $x$ to $y$ in the majority relation $\succsim_R$. In particular, $\mathit{md}(x,x,{\succsim_R})=0$, $\mathit{md}(x,y,{\succsim_R})=1$ if $x\succsim_R y$, and $\mathit{md}(x,y,{\succsim_R})=\infty$ if there is no path from $x$ to $y$ in $\succsim_R$. We extend this notion also to lotteries $p$ by defining $\mathit{md}(p,y,{\succsim_R})= \sum_{x\in X_R} p(x) \mathit{md}(x,y,{\succsim_R})$ and note that $\mathit{md}(p,y,{\succsim_R})=\infty$ if and only if there is an alternative $x$ with $p(x)>0$ and $\mathit{md}(x,y,{\succsim_R})=\infty$. 

\begin{restatable}{proposition}{pathlength}\label{prop:pathlength}
It holds for all majoritarian RSCFs $f$ and preference profiles $R$ on $m\geq 3$ alternatives that
    \begin{enumerate}[label=(\arabic*),leftmargin=*]
        \item $\mathit{dist}(f(R), R)\leq 1+2\max_{x\in X_R} \mathit{md}(f(R), x, {\succsim_R})$. 
        \item $\mathit{dist}_m(f)\geq 1+2\max_{x\in X_R} \mathit{md}(f(R), x, {\succsim_R})$. 
    \end{enumerate}
\end{restatable}
\begin{proof}[Proof sketch]
    For Claim (1), we first note that there is nothing to show if $\max_{x\in X_R} \mathit{md}(f(R), x, {\succsim_R})=\infty$ and we hence suppose that $\mathit{md}(f(R), x, {\succsim_R})<\infty$ for all $x\in X_R$. We then prove that $sc(x,d)\leq (1+2\mathit{md}(x,y,{\succsim_R}))sc(y,d)$ for all $x,y\in X_R$ and $d\in D(R)$ by an induction on the majority distance between $x$ and $y$. This insight implies Claim~(1) as $\mathit{dist}(f(R), R)=\sup_{d\in D(R)}\frac{\sum_{x\in X_R} f(R,x) sc(x,d)}{\min_{y\in X_R} sc(y,d)}$. For Claim (2), we show that there is for every $\epsilon>0$ a preference profile $R^\epsilon$ and a metric space $d\in D(R^\epsilon)$ such that ${\succsim_{R^\epsilon}}={\succsim_R}$ and $\frac{sc(f(R), d)}{\min_{y\in X_R} sc(y,d)}\geq 1+2\max_{x\in X_R} \mathit{md}(f(R), x, \succsim_R) -\epsilon$. Since $f(R^\epsilon)=f(R)$ as $f$ is majoritarian, we then infer Claim~(2) by letting $\epsilon$ go to~$0$. 
\end{proof}

Claims related to \Cref{prop:pathlength} have been shown by \citet[][Lemma 6]{ABE+18a} and \citet[][Corollary 5.1]{Kemp20a}, but these results lack the lower bound given in (2). Based on our proposition, we will next compute the metric distortion of C1ML~rules. In particular, our subsequent theorem shows that C1ML rules have a metric distortion of at most $4$ and that no majoritarian RSCF has a lower metric distortion if the number of alternatives $m$ is unbounded. 

\begin{restatable}{theorem}{maj}\label{thm:maj}
The following claims are true:
\begin{enumerate}[leftmargin=*, label=(\arabic*)]
    \item It holds for all C1ML rules $f$ and $m\geq 3$ that $\mathit{dist}_m(f)\leq 4$ and $\mathit{dist}_m(f)\geq 4-(\frac{1}{3})^{\lfloor\frac{m-3}{2}\rfloor}$. 
    \item It holds for all majoritarian RSCFs $f$ that $\mathit{dist}_m(f)\geq 4-\frac{3}{m}$ if $m\geq 3$ is odd and $\mathit{dist}_m(f)\geq 4-\frac{3}{m-1}$ if $m\geq 3$ is even.
\end{enumerate}
\end{restatable}
\begin{proof}
We will only prove Claim (1) here and give a proof sketch for Claim (2). The full proof of Claim (2) is deferred to the appendix.\smallskip

\textbf{Claim (1), upper bound:} Let $f$ denote a C1ML rule, let $R$ denote a profile, and define $p=f(R)$. It follows from a result by \citet{DuLa99a} that $p(x)>0$ implies $\mathit{md}(x,y,{\succsim_R})\leq 2$ for all $x,y\in X_R$. Based on this insight, we will show that $\mathit{md}(p,z,{\succsim_R})\leq \frac{3}{2}$ for all $z\in X_R$ as Claim (1) of \Cref{prop:pathlength} then proves that $\mathit{dist}(p,R)\leq 4$. We thus fix an alternative $z\in X_R $ and let $q$ denote the lottery with $q(z)=1$. Further, we define $X^+=\{x\in X_R\colon x\succ_R z\}$ and $X^-=\{x\in X_R\colon z\succ_R x\}$. By the definition of C1ML rules, it holds that $p\succsim_R q$, which implies that $\sum_{x\in X^+} p(x)\geq \sum_{x\in X^-} p(x)$. This means that $\sum_{x\in X^-} p(x)\leq \frac{1}{2}$. Next, it holds for all $x\in X_R$ with $p(x)>0$ that $\mathit{md}(x,z,{\succsim_R})=1$ if $x\succsim_R z$ and $\mathit{md}(x,z,{\succsim_R})=2$ if $z\succ_R x$ due to our previous observation. Hence, we infer that $\mathit{md}(p,z,{\succsim_R})\leq \sum_{x\in X_R\colon x\succsim_R z} p(x) + 2\sum_{x\in X_R\colon z\succ_R x} p(x)=1+\sum_{x\in X^-} p(x)\leq \frac{3}{2}$. Finally, Claim (1) of \Cref{prop:pathlength} shows that $\mathit{dist}(p,R)\leq 4$.\medskip

\textbf{Claim (1), lower bound:} For proving our lower bound, we recall that C1ML rules are majoritarian and that $|\mathit{C1ML}(R)|=1$ if the majority relation of $R$ is strict \citep{LLL97a}. Moreover, by McGarvey's construction \citep{McGa53a}, there is for every complete binary relation $\succsim$ on $X_m$ a profile $R$ with ${\succsim_R}={\succsim}$. Due to Claim (2) of \Cref{prop:pathlength}, we can hence show the lower bound by constructing a complete and anti-symmetric binary relation $\succsim^*$ for every $X_m$ with $m\geq 3$ such that $\max_{x\in X_R} \mathit{md}(p, x, {\succsim^*})=\frac{3}{2}-\frac{1}{2}\cdot (\frac{1}{3})^{\lfloor \frac{m-3}{2}\rfloor}$, where $p$ is the unique C1 maximal lottery of a profile $R$ with ${\succsim_R}={\succsim^*}$. We first suppose that $m\geq 3$ is odd and consider the following relation $\succsim^*$ on $X_m$: for all odd $k<m$ and all $j$ with $k+2\leq j\leq m$, it holds that $x_{k+1}\succ^* x_{k}$, $x_{k}\succ^* x_{j}$, and $x_j\succ^* x_{k+1}$. It can be checked that the unique C1 maximal lottery $p$ for this relation is defined by $p(x_k)=p(x_{k+1})=(\frac{1}{3})^{\frac{k+1}{2}}$ for all odd $k<m$ and $p(x_m)=(\frac{1}{3})^{\frac{m-1}{2}}$. This means that $\sum_{x\in X^o} p(x)=\sum_{x\in X^e} p(x_k)=\frac{1}{2}-\frac{1}{2}p(x_m)$ for the sets $X^o=\{x_1,x_3,\dots, x_{m-2}\}$ and $X^e=\{x_2, x_4, \dots, x_{m-1}\}$. Next, by definition of $\succsim^*$, it holds for all odd $k<m$ that $\mathit{md}(x_k, x_m, {\succsim^*})=1$ and $\mathit{md}(x_{k+1}, x_m, {\succsim^*})=2$. Hence, $\mathit{md}(p,x_m,{\succsim^*})= \sum_{x\in X^o}p(x) + 2\sum_{x\in X^e} p(x)=3\left(\frac{1}{2}-\frac{1}{2}p(x_m)\right)=\frac{3}{2}-\frac{1}{2}\cdot (\frac{1}{3})^{\frac{m-3}{2}}$. \Cref{prop:pathlength} then shows that $\mathit{dist}_m(f)\geq 4-(\frac{1}{3})^{\frac{m-3}{2}}$. Finally, to extend this result to even $m$, we add a new alternative to $\succsim^*$ that loses all majority comparisons. Every C1ML rule will assign probability $0$ to this alternative and it does hence not affect our analysis.\medskip

\textbf{Claim (2):} In this proof sketch, we assume that $m\geq 3$ is odd. To prove the theorem in this case, we will again use Claim (2) of \Cref{prop:pathlength} and hence construct a profile $R$ such that $\max_{x\in X_R}\mathit{md}(p,x,{\succsim_R})\geq \frac{3}{2}-\frac{3}{2m}$ for every lottery $p$. Next, McGarvey's theorem \citep{McGa53a} allows us to focus on complete binary relations on $X_m$. The theorem then follows by proving that $\max_{x\in X_R}\mathit{md}(p,x,{\succsim})\geq \frac{3}{2}-\frac{3}{2m}$ for all lotteries $p$ and the ``cyclic'' relation $\succsim$ given by $x_i\succ x_{i+_m k}$ for all $i\,\shin\, \{1,\dots,m\}$, $k\,\shin\, \{1,\dots, \frac{m-1}{2}\}$ (where $i+_m k=i+k$ if $i+k\leq m$ and $i+_m k=i+k-m$ else).
\end{proof}

\paragraph{Remark 1.} The upper bound in Claim (1) of \Cref{thm:maj} is tight as there are C1ML rules $f$ with $\mathit{dist}(f)=4$. To see this, consider the lottery $p$ given by $p(a)=p(c)=\frac{1}{2}$ and a profile $R$ with $X_R=\{a,b,c\}$, $a\succ_R b$, $b\succ_R c$, and $c\sim_R a$. Since $p$ is C1 maximal in $R$ and $\mathit{md}(p,b,{\succsim_R})=\frac{3}{2}$, \Cref{prop:pathlength} shows that $\mathit{dist}(f)=4$ for all C1ML rules $f$ with $f(R)=p$. By contrast, the lower bound for C1ML rules is not tight. It can be shown that every C1ML rule has a metric distortion of at least $4-3\gamma_m$, where $\gamma_m$ denotes the minimal non-zero probability that a C1ML rule assigns to an alternative in a profile with $m$ alternatives and an odd number of voters. However, the probabilities~$\gamma_m$ are not well-understood \citep{FiRy95b}, so we cannot use them to improve our lower bound for C1ML rules. 

\paragraph{Remark 2.} \Cref{prop:pathlength} allows us to identify the majoritarian RSCF that minimizes $\mathit{dist}_m(f)$ for a fixed number of alternatives $m$: this RSCF $f^*$ chooses for each profile $R$ a lottery $p$ that minimizes $\max_{x\in X_R} \mathit{md}(p,x, {\succsim_R})$. Based on a computer-aided approach, we have shown that $\mathit{dist}_m(f^*)=4-\frac{3}{m}$ for all odd $m\leq 9$, which proves that Claim (2) of \Cref{thm:maj} is tight in these cases.

\paragraph{Remark 3.} \Cref{prop:pathlength} recovers known bounds on the metric distortion of majoritarian SCFs. For instance, this proposition implies that every alternative in the uncovered set has a metric distortion of $5$ because the uncovered set is the set of alternatives that can reach every other alternative in at most two steps. This result has been first shown by \citet{ABE+18a}.

\section{Simulations}\label{sec:simulations}

As our second contribution, we conduct extensive computer experiments to gain insights into the average-case metric distortion of the RSCFs in \Cref{subsec:RSCF}. To this end, we first derive a linear program that efficiently computes the metric distortion of a lottery for a profile (\Cref{subsec:compute}), and then explain the setup and results of our experiments (\Cref{subsec:simulations,subsec:realworld}). The code for our experiments is publicly available on Zenodo \citep{FrLe25b}.

\subsection{Computing the Metric Distortion}\label{subsec:compute}

The main challenge for our experiments is to compute the metric distortion $\mathit{dist}(p, R)$ for a given lottery~$p$ and profile~$R$. To this end, we note that it suffices to compute the term $\mathit{dist}(p,R,x)=\sup_{d\in D(R)} \frac{sc(p,d)}{sc(x,d)}$ for all alternatives $x\in X_R$ because $\mathit{dist}(p, R)=\max_{x\in X_R} \mathit{dist}(p,R,x)$. Moreover, we can assume that $sc(x,d)=1$ since the term $\frac{sc(p,d)}{sc(x,d)}$ is invariant under scaling $d$. Hence, we only need to find for every alternative $x$ the metric $d_x$ that maximizes $sc(p,d_x)$ subject to $d_x\in D(R)$ and $sc(x,d_x)=1$. While this can be done by linear programs that use the distances $d(x,v)$ as variables and encode that $d\in D(R)$ and $sc(x,d)=1$, this straightforward approach is too slow for our experiments as we need $\mathcal{O}((n+m)^3)$ constraints to formalize the triangle inequalities for metrics. 

To derive a more efficient method to compute $\mathit{dist}(p,R,x)$, we will use that the metric distortion of a lottery $p$ for a profile $R$ can be computed by only considering the biased metrics of \citet{ChRa22a}. To define these metrics, we let $\succeq_v$ denote the relation given by $x\succeq_v y$ if and only if $x\succ_v y$ or $x=y$ for all $x,y\in X_R$. Then, a metric $d$ is \emph{biased} for a profile $R$ if there is an alternative $x^*\in X_R$ and a function $t:X_R\rightarrow \mathbb{R}_{\geq 0}$ such that \emph{(i)} $t(x^*)=0$, \emph{(ii)} $d(x^*,v)=\frac{1}{2}\max_{x,y\in X_R\colon x\succeq_v y} t(x)-t(y)$ for all $v\in V_R$, and \emph{(iii)} $d(x,v)=d(x^*,v)+\min_{y\in X_R\colon x\succeq_v y} t(y)$ for all $v\in V_R$ and all $x\in X_R\setminus \{x^*\}$. Unfortunately, due to the maxima and minima in the definition of these metrics, we cannot directly use them to compute $\mathit{dist}(p, R)$. We thus adapt the idea of biased metrics to construct a linear program that efficiently computes this value. In more detail, we will show that the following LP (called \ref{LP}), which uses variables $d(x,v)$ and $t(x)$ for $x\in X_R$ and $v\in V_R$, computes $\mathit{dist}(p,R,x^*)$ for every lottery $p$, profile~$R$, and alternative~$x^*$.  
\noindent\begin{equation*}
\hspace{-0.13cm}
    \begin{array}{lll}
    \text{max}\!\!\!\!\!\!\!\!\! & \sum\limits_{x\in X_R} p(x)\sum\limits_{v\in V_R} d(x,v)\\
    \text{s.t.}\!\!\!\! & t({x^*})=0 &\\ 
          & t(x)\geq 0 & \forall x\!\in\! X_R\\
          & d(x^*, v)\geq \frac{1}{2}(t(x)-t(y)) & \forall v\!\in\! V_R, \,x, y\!\in\! X_R\colon x\succeq_v y\!\!\!\!\!\\
          &d(x, v)\leq d(x^*, v)+t(y) & \forall v\!\in\! V_R, \,x, y\!\in\! X_R\colon x\succeq_v y\!\!\!\!\!\\
          & d(x, v)+d(x^*, v)\geq t(x) & \forall v\!\in\! V_R, \,x\!\in\! X_R\\
          & \sum_{v\in V_R} d(x^*, v)=1 &
    \end{array}
    \tag{LP 1}\label{LP}
\end{equation*}

\begin{restatable}{proposition}{LP}\label{prop:LPcorrect}
Fix a profile $R$, a lottery $p$, and~an alternative~$x^*$. If the optimal objective value $o_{LP}^*$ of \ref{LP} is bounded, then $\mathit{dist}(p,R,x^*)=o_{LP}^*$ and otherwise $\mathit{dist}(p,R,x^*)=\infty$.
\end{restatable}
\begin{proof}[Proof sketch]
    Let $R$ denote a profile, $p$ a lottery, and $x^*$ an alternative. It can be checked that every biased metric $d$ together with its inducing function $t$ satisfies the conditions of \ref{LP},
    so the optimal objective value $o_{LP}^*$ of our LP is lower bounded by $\mathit{dist}(p,R,x^*)$. Specifically, the constraints in the first four lines follow directly from the definition of $d$ and $t$, the fifth line follows by substituting the definitions of $d(x,v)$ and $d(x^*,v)$, and the constraint that $\sum_{v\in V_R} d(x^*, v)=1$ can be enforced by scaling $d$ and $t$ without affecting $\mathit{dist}(p,R,x^*)$.
    Conversely, to show that $\mathit{dist}(p,R,x^*)\geq o_{LP}^*$, we prove that every feasible solution of \ref{LP} with objective value $o_{LP}$ can be transformed into a metric $d\in D(R)$ such that $\frac{sc(p,d)}{sc(x^*,d)}\geq o_{LP}$. We note that this direction is independent of the work of \citet{CRWW23a} as we need to reason about our constraints to turn a feasible solution of our LP into a metric. 
\end{proof}

Given a profile $R$ on $n$ voters and $m$ alternatives, \ref{LP} has $\mathcal{O}(nm^2)$ constraints and it is thus very fast to solve this LP. For example, based on \ref{LP}, we need in average roughly $20$ seconds on a single core of an Apple M1 Ultra chip to compute the metric distortion of a lottery for a profile with $201$ voters and $15$ alternatives.

\subsection{Simulations with Synthetic Data}\label{subsec:simulations}

As our first computer experiment, we conduct extensive simulations based on synthetically generated preference profiles. In more detail, we generate $1000$ preference profiles on $m$ alternatives and $n$ voters for $14$ distributions over preference profiles and all combinations of $(m,n)\in \{5,10,15\}\times \{11,21,\dots, 201\}$. For each of the sampled preference profiles $R$, we then compute the lottery $f(R)$ chosen by the five RSCFs discussed in \Cref{subsec:RSCF} and their respective metric distortion $\mathit{dist}(f(R),R)$. Finally, for each $m\in \{5,10,15\}$ and each distribution, we plot the average metric distortion of each RSCF as a function depending on the number of voters $n$. Due to space restrictions, we show the results of these simulations only for two exemplary distributions, namely the impartial culture model and the $3$-dimensional Euclidean cube model. The plots for the other distributions (the Mallow's model, the P\'olya-Eggenberg urn model, and the $t$-dimensional Euclidean cube and ball models with various parameterizations) as well as further statics can be found in the appendix. In particular, our experiments cover all major models used in the ``map of elections'' \citep{SFS+20a,BBF+21a,BFL+24a}. We next define the impartial culture and the Euclidean cube models.

\paragraph{Impartial Culture (IC)} In this model, each voter is assigned a preference relation independently and uniformly at random. Hence, for each voter $v\in V_n$ and preference relation ${\succ}\in\mathcal{R}(X_m)$, the probability that $\succ$ is assigned to $v$ is $\frac{1}{m!}$.

\paragraph{$t$-Dimensional Euclidean Cube ($t$EC)} In this model, we assign voters and alternatives independently and uniformly at random to points in the $t$-dimensional cube $[-1,1]^t$. The voters' preference relations are then given by their distances to the alternatives: a voter $v$ prefers alternative $x$ to alternative $y$ if $|p_v-p_x|_2<|p_v-p_y|_2$ where $p_v$, $p_x$, and $p_y$ denote the points of $v$, $x$, and $y$ in the hypercube. In the main body, we use this model with $t=3$ dimensions.\medskip

The results of our simulations for these two models are shown in \Cref{fig:results}. We first note that, in most experiments, the measured variance is rather lower, with typical values lying between $0.05$ to $0.01$ (see the appendix for details). Moreover, in all experiments, the average metric distortion of the considered RSCFs is significantly smaller than their worst-case metric distortion, thus indicating that such worst-case bounds are too pessimistic for more realistic profiles. In particular, the average metric distortion of all RSCFs is usually in the interval $[2,2.5]$, which also shows that the choice of a particular rule has only limited effect. This is especially striking when comparing C1ML and C2ML rules, which are almost indistinguishable in our experiments even though the worst-case metric distortion is $3$ for C2ML rules and $4$ for C1ML rules. 

Beyond these general observations, there are several interesting trends in our experiments that can be observed for most of the distributions. We explain these trends for each RSCF individually. 

\paragraph{CRWW rule} In most of our simulations and especially when $m\in \{10,15\}$, the CRWW rule $f_\mathit{CRWW}$ has the lowest average metric distortion among the tested rules. In particular, for effectively all distributions and all numbers of voters, the average metric distortion of this rule lies between $2$ and $2.15$, thus demonstrating a strong resilience against the used sampling model. These results suggest that, when the metric distortion is the central factor for deciding on the RSCF, we should use the CRWW rule as it has both the best worst-case and average-case metric distortion.

\paragraph{Randomized Plurality-Veto} The randomized Plurality-Veto rule $f_\mathit{RPV}$ has often a very low average metric distortion if there are only $m=5$ alternatives, but it becomes worse as $m$ increases. For instance, in the $3$-dimensional Euclidean cube model, it has for most values of $n$ an average metric distortion of less than $2$, but its average metric distortion increases to over $2.2$ when $m=15$. By contrast, in the impartial culture model, the average metric distortion of $f_\mathit{RPV}$ depends significantly on the number of voters and roughly converges against $2+\frac{1}{m-1}$. We believe the reason for this is that, in our simulations, $f_\mathit{RPV}$ randomizes over larger sets of alternatives when $m$ increases. This is beneficial for the metric distortion if preference profiles are sufficiently close to uniform, but detrimental if there is an alternative that every voter appreciates. Consequently, the randomized Plurality-Veto rule performs worse than, e.g., the C1ML and C2ML rules when using the Euclidean Cube model and $m\in \{10,15\}$ alternatives.

\paragraph{Uniform random dictatorship} The average metric distortion of the uniform random dictatorship $f_\mathit{RD}$ becomes smaller as the number of voters increases when using distributions that are close to uniform. For instance, for the impartial culture model and all $m\in \{5,10,15\}$, the average metric distortion of $f_\mathit{RD}$ converges to a value close to $2$ as the number of voters $n$ increases. By contrast, if the voters' preferences are more structured (e.g., in the Euclidean cube model), the average metric distortion is largely independent of the number of voters and significantly worse than that of the other rules. A possible explanation for this is that $f_\mathit{RD}$ only considers the voters' favorite alternatives and thus fails to identify strong compromise alternatives. In particular, such strong alternatives are likely to exist for structured distributions, but typically do not exist if $n$ is large and the distribution over profiles is close to uniform.

\paragraph{C1ML and C2ML rules.} The average metric distortion of C1ML and C2ML rules is rather high for the impartial culture model with $m=5$ (close to $2.25$ when $n\geq 100$), but it is close to that of the CRWW rule for more structured distributions (e.g., the $3$-dimensional Euclidean cube model). Moreover, for most distributions, the average metric distortion of these rules decreases when the number of alternatives increases. Our explanation for this is that C1ML and C2ML rules use very little randomization as they select randomized Condorcet winners. This behavior results in a low metric distortion if there are alternatives that severely beat all other alternatives in a pairwise comparison, but it is detrimental if all alternatives are roughly equally good or if their is a Condorcet winner that only narrowly beats the other alternatives.
\begin{figure}[t]
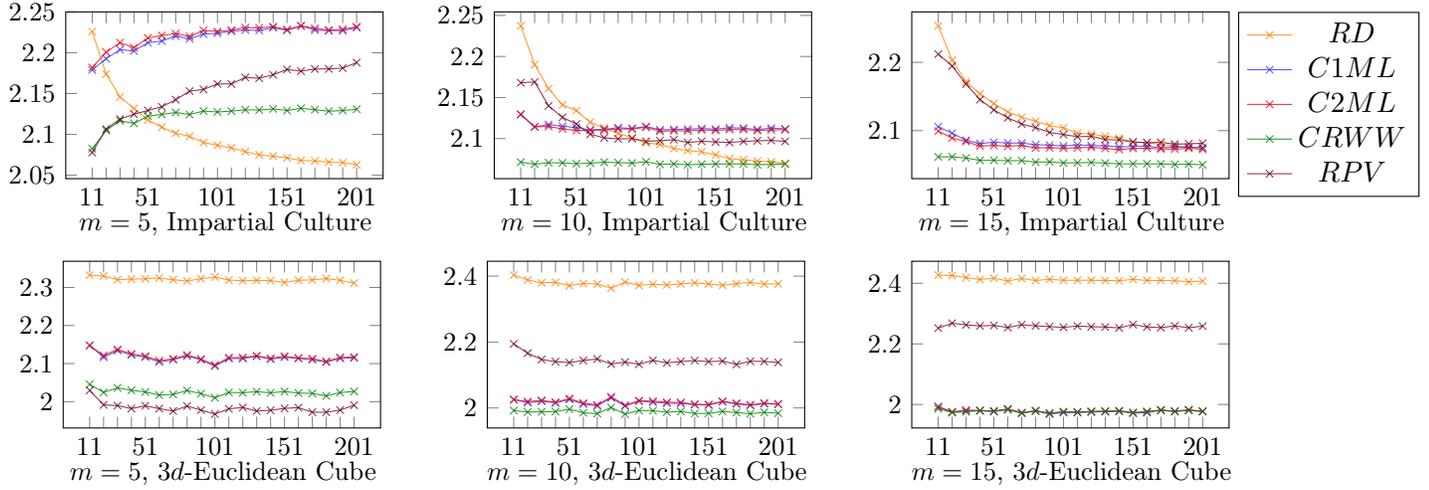

    \centering
\begin{tabular}{rrl}
    \hspace{-0.5cm}\plot{CSV/5candidates/ImpartialCultureAverage.csv}{$m=5$, Impartial Culture}{\legend{};} 
    & 
    \plot{CSV/10candidates/ImpartialCultureAverage.csv}{$m=10$, Impartial Culture}{\legend{};} & 
    \plot{CSV/15candidates/ImpartialCultureAverage.csv}{$m=15$, Impartial Culture}{} \\
\plot{CSV/5candidates/EuclideanCube3Average.csv}{$m=5$, $3d$-Euclidean Cube}{\legend{};} & 
\plot{CSV/10candidates/EuclideanCube3Average.csv}{$m=10$, $3d$-Euclidean Cube}{\legend{};} & 
\plot{CSV/15candidates/EuclideanCube3Average.csv}{$m=15$, $3d$-Euclidean Cube}{\legend{};}
\end{tabular}\vspace{-0.3cm}
\caption{Results of our simulations for the impartial culture and $3$-dimensional Euclidean cube models. For both models and $m\!\in\!\{5,10,15\}$ alternatives, we plot the average metric distortion ($y$-axis) of the uniform random dictatorship, the C2ML and C1ML rules, the randomized Plurality-Veto rule, and the CRWW rule subject to the number of voters $n\!\in\! \{11,21,\dots,201\}$ ($x$-axis).}
\label{fig:results}
\end{figure}
\begin{figure}[t]
    \centering
\pgfplotstableread[col sep=comma,]{CSV/SpotifyDataAverage.csv}\datatable
\begin{tikzpicture}
\begin{axis}[
    width=9cm,
    height=5cm,
    x tick label style={font=\small},
    xtick=data,
    xticklabels from table={\datatable}{n},
    legend style={at={(1.4,1,5)},anchor=north east},
    y tick label style={font=\small,
    }
    ]    
    \addplot [mark=x, orange, opacity=0.8   ] table [x expr=\coordindex, y={RandomDictatorship}]{\datatable};
    \addlegendentry{$RD$};
    
    \addplot [name path = C1ML, mark=x, blue!80, yshift=0.4, opacity=0.8] table [x expr=\coordindex, y={C1ML}]{\datatable};
    \addlegendentry{$C1ML$};

    \addplot [name path = C2ML, mark=x, red!95!black, opacity=0.8] table [x expr=\coordindex, y={C2ML}]{\datatable};
    \addlegendentry{$C2ML$};

    \addplot [name path = CRWW, mark=x, opacity=0.8, green!50!black ] table [x expr=\coordindex, y={CRWW}]{\datatable};
    \addlegendentry{$CRWW$};

    \addplot [name path = PluralityVeto, mark=x, opacity=0.8, purple!60!black ] table [x expr=\coordindex, y={PluralityVeto}]{\datatable};
    \addlegendentry{$RPV$};
\end{axis}
\draw[line width = 0.03cm, color = black!40] (0.45,0) -- (0.45, 0.2); 
\draw[line width = 0.03cm, color = black!40] (1.575,0) -- (1.575, 0.2); 
\draw[line width = 0.03cm, color = black!40] (2.7,0) -- (2.7, 0.2);
\draw[line width = 0.03cm, color = black!40] (3.84,0) -- (3.84, 0.2);
\end{tikzpicture}\vspace{-0.2cm}
\caption{Results of our simulations with the Spotify Daily dataset. Each data point presents the average metric distortion of one of our RSCFs over 14 days (e.g., the first data point averages the metric distortion from January 01 to January 14, the second one from January 15 to January 28, etc.).}
\label{fig:realworld}
\end{figure}
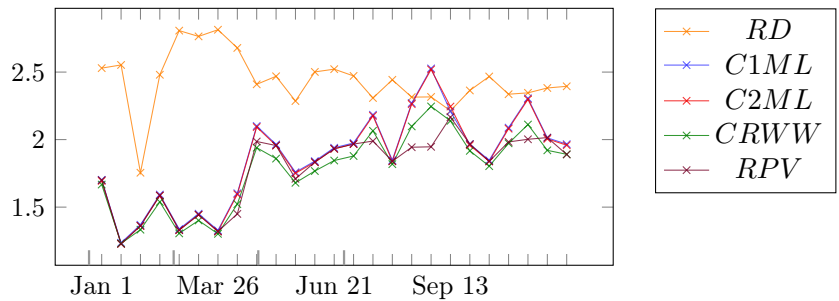

\subsection{Simulations with Real-world Data}\label{subsec:realworld}

We also conduct computer experiments on the average metric distortion of our RSCFs on real-world data from PrefLib \citep{MaWa13a}. In more detail, we use the \emph{Spotify Daily dataset} provided by \citet{BoSa23a} for our experiments. This dataset contains the rankings of the $200$ most listened songs on Spotify for $53$ countries and every day in $2017$, and the rankings of each day form an election. Since not every country ranks the same songs among their top-$200$, \citet{BoSa23a} have identified maximal complete subelections for each day, which typically contain between $40$ and $50$ voters and around $20$ alternatives. These maximal subelections are very well-suited for our computer experiments because they contain a reasonable number of voters and alternatives and all voters rank all alternatives. Due to these characteristics, we can directly compute our five RSCFs and their metric distortion on these maximal subelections, i.e., we conduct our experiments on the real-world data without any modifications other than those made by \citet{BoSa23a}. This also means that it suffices to compute the metric distortion of our RSCFs for each subelection once since no randomization is used in the generation of preference profiles.

The results of our experiments with the Spotify Daily dataset are shown in \Cref{fig:realworld}, where we display the average metric distortion of each RSCF in a biweekly rhythm, i.e., each data point is the average of $14$ days. Additional statistics can again be found in the apendix.
We note that the simulations on the Spotify Daily dataset roughly agree with our computer experiments based on, e.g., the Euclidean models or Mallow's model (see the appendix for more details). In more detail, \Cref{fig:realworld} shows that the metric distortion of the randomized Plurality-Veto and CRWW rules is typically only slightly better than that of the C1ML and C2ML rules, whereas the uniform random dictatorship often performs significantly worse. There are, however, two central differences between our simulations with synthetic data and the Spotify Daily dataset. Firstly, in the first three month, there is often a very dominant alternative in the elections from Spotify, which results in a very low metric distortion for all tested RSCFs but $f_\mathit{RD}$. Such profiles do typically not appear in our synthetic data, which may hint at the fact that our computer experiments are still too pessimistic. However, after the first three month, this effect vanishes and our RSCFs take similar values as in the simulation with synthetic data. Secondly, the randomized Plurality-Veto rule $f_\mathit{RPV}$ performs surprisingly well on the real-world data, in particular in light of the large number of alternatives. The reason for this may be that, while there are many alternatives, only few of them are first-ranked in the preference profiles and only such alternatives have a chance to win under~$f_\mathit{RPV}$.

To summarize our computer experiments with both synthetic and real-world data, we believe that they firstly show that that tailored RSCFs, such as the CRWW rule and the randomized Plurality-Veto rule, typically also have the smallest average-case metric distortion. However, especially when preference profiles are sufficiently structured, C1ML and C2ML rules are only slightly worse, thus providing an argument in favor of these rules. Lastly, our simulations show that the uniform random dictatorship is not suitable to minimize the metric distortion in practice, especially when we expect strong alternatives to exist.

\section{Theoretical Average-case Analysis}\label{sec:IC}

Lastly, we will analytically examine the average-case metric distortion of our RSCFs by calculating their expected metric distortion for a randomly drawn profile. In particular, the results in this section can be seen as a rigorous mathematical counterpart to the simulations in \Cref{sec:simulations}. For mathematical feasibility, we will restrict our attention to the impartial culture model and write $\mathit{IC}(m,n)$ for the respective probability distribution over profiles with $m$ alternatives and $n$ voters. While the impartial culture model is somewhat unrealistic, it is frequently analyzed as it is often seen as a good starting point for average-case analyses \citep[e.g.,][]{GeFi76b,GeFi78a,PaWe78a,PrWi09a,BCH+15a,KaXi21a}. 

As we show next, if the number of voters goes to infinity, the expected metric distortion of every RSCF with bounded metric distortion converges to a value between $2$ and $2+\frac{1}{m-1}$ under the impartial culture model. This means that the choice of the voting rule has only a small effect on the expected metric distortion if there is a large number of voters and alternatives. 

\begin{restatable}{theorem}{ML}\label{thm:ML}
    Let $m\geq 3$. It holds for every RSCF $f$ with $\mathit{dist}_m(f)<\infty$ and $z=\liminf_{n\rightarrow\infty } \mathbb{P}_{R\sim IC(m,n)}[\exists x\in X_R\colon f(R,x)=0]$ that 
    \begin{enumerate}[label=(\arabic*),leftmargin=*]
        \item $\limsup_{n\rightarrow \infty} \mathbb{E}_{R\sim IC(m,n)}[\mathit{dist}(f(R),R)]\leq 2+\frac{1}{m-1}$
        \item  $\liminf_{n\rightarrow \infty} \mathbb{E}_{R\sim IC(m,n)}[\mathit{dist}(f(R),R)]\geq  2+\frac{z}{m-1}$.
    \end{enumerate}
\end{restatable}
\begin{proof}[Proof sketch]
    The basic idea of this proof is that, as $n$ grows larger, a profile drawn from $IC(m,n)$ is with high probability close to the profile $R^*$ where each preference relation is reported by the same amount of voters. We thus start by analyzing the profile $R^*$ and show with the help of \ref{LP} that the metric distortion of a lottery $p$ on $R^*$ is $2+\frac{1}{m-1}-\frac{m}{m-1}\cdot \min_{x\in X_m} p(x)$. Next, we prove that the metric distortion of a lottery $p$ for a profile $R$ that contains $R^*$ as a large subprofile can be bounded based on the metric distortion of $p$ for $R^*$. Based on these two insights, we then fix a number of voters $n$, set $\alpha=\frac{1}{\sqrt[3]{n}}$, and define $T^\alpha$ as the set of profiles for $n$ voters such that each preference relation is reported by more than $(1-\alpha)\frac{n}{m!}$ voters. Using the law of total probability, we derive that 
    \begin{align*}
        \mathbb{E}[\mathit{dist}(f(R), R)]&=\mathbb{P}[R\not\in T^\alpha]\cdot \mathbb{E}[\mathit{dist}(f(R), R)|R\not\in T^\alpha] \\
        &+ \mathbb{P}[R\in T^\alpha]\cdot \mathbb{E}[\mathit{dist}(f(R), R)|R\in T^\alpha]
    \end{align*}
    Finally, for our upper bound, we use standard concentration bounds to show that $\mathbb{P}[R\not\in T^\alpha]$ goes to $0$ as $n$ increases. Since $\mathbb{E}[\mathit{dist}(f(R), R)|R\not\in T^\alpha]\leq \mathit{dist}_m(f)<\infty$, it hence follows that $\mathbb{E}[\mathit{dist}(f(R), R)]$ converges to $\mathbb{E}[\mathit{dist}(f(R), R)|R\in T^\alpha]$ when $n$ goes to infinity, and using our previous insights, we can bound this expectation by $2+\frac{1}{m-1}$. Similarly, for our lower bound, we use that $\mathbb{E}[\mathit{dist}(f(R), R)]\geq \mathbb{P}[R\in T^\alpha]\cdot \mathbb{E}[\mathit{dist}(f(R), R)|R\in T^\alpha]$ and derive a lower bound on $\mathbb{E}[\mathit{dist}(f(R), R)|R\in T^\alpha]$.
\end{proof}

Based on a similar approach as in \Cref{thm:ML}, we will next precisely compute the expected metric distortion of the uniform random dictatorship $f_\mathit{RD}$ and the randomized Plurality-Veto rule~$f_\mathit{RPV}$. In particular, we will show that, in the limit, $f_\mathit{RD}$ has an optimal expected metric distortion of $2$ under the impartial culture model, whereas the expected metric distortion of $f_\mathit{RPV}$ is $2+\frac{1}{m-1}$.

\begin{restatable}{theorem}{expectedRD}\label{prop:expecteddistRD} 
It holds for every $m\geq 3$ that 
\begin{enumerate}[leftmargin=*, label=(\arabic*)]
    \item $\lim_{n\rightarrow\infty}\mathbb{E}_{R\sim IC(m,n)}[\mathit{dist}(f_{RD}(R),R)]=2$.
    \item $\lim_{n\rightarrow\infty}\mathbb{E}_{R\sim IC(m,n)}[\mathit{dist}(f_{RPV}(R),R)]=2+\frac{1}{m-1}$.
\end{enumerate}
\end{restatable}
\begin{proof}[Proof Sketch]
For Claim (1), we show that $f_\mathit{RD}(R)$ returns with high probability a lottery close to the uniform one when $n$ is large and $R$ is drawn from $IC(m,n)$. From this insight, we then derive that $\limsup_{n\rightarrow \infty} \mathbb{E}_{R\sim IC(m,n)}[\mathit{dist}(f(R),R)]\leq 2$.
Combined with the lower bound in \Cref{thm:ML}, this proves Claim (1). For the claim on $f_\mathit{RPV}$, we prove that this rule only randomizes over all alternatives when each alternative is top-ranked and bottom-ranked by the same number of voters. Since this happens with probability $0$ when $n$ goes to infinity, we infer that $\liminf_{n\rightarrow\infty } \mathbb{P}_{R\sim IC(m,n)}[\exists x\in X_R\colon f_\mathit{RPV}(R,x)=0]=1$, and Claim (2) follows from \Cref{thm:ML}.
\end{proof}

\paragraph{Remark 4.} We leave an analogous result to \Cref{prop:expecteddistRD} for C1ML and C2ML rules open because completely reversed preference relations cancel each other out for these RSCFs. Thus, a small part of the profile may determine the outcome, which severely complicates the analysis of these rules. However, computer experiments by \citet{BBS20a} suggest that the probability $\mathbb{P}_{R\sim IC(m,n)}[\exists x\in X_R\colon f(R,x)=0]$ is close to $1$ for C1ML and C2ML rules. Hence, \Cref{thm:ML} implies that the expected metric distortion of these RSCFs under the IC model is close to $2+\frac{1}{m-1}$ when $n$ goes to~$\infty$.

\paragraph{Remark 5.} \Cref{thm:ML} shows that, under the impartial culture distribution, \emph{every} deterministic SCF $f$ with $\mathit{dist}_m(f)<\infty$ has an expected metric distortion of $2+\frac{1}{m-1}$ when $n$ goes to $\infty$. This holds because every SCF is an RSCF that always assigns probability $1$ to a single alternative, so the value $z$ in \Cref{thm:ML} is~$1$.

\section{Conclusion}

In this paper, we study the metric distortion of randomized social choice functions, with a particular focus on well-established RSCFs such as the uniform random dictatorship, C1ML rules, and C2ML rules. Specifically, we first show that every C1ML rule has a metric distortion of at most $4$, and we give a lower bound on the metric distortion of all majoritarian RSCFs that converges to $4$ as $m$ increases. This means that C1ML rules minimize the metric distortion within the class of majoritarian RSCFs when the number of alternatives is unbounded. 
Secondly, we conduct extensive computer experiments on the metric distortion of these three classical RSCFs as well as two RSCFs designed to minimize the metric distortion.
These experiments show that, while RSCFs designed to minimize the metric distortion have also the best average-case metric distortion, C1ML and C2ML rules are often only slightly worse. 
Finally, we also conduct an analytical average-case analysis for the impartial culture model and, surprisingly, derive that the exact choice of voting rule has only a negligible influence on the expected metric distortion if the number of voters is large. 
In summary, we believe that these results demonstrate that established RSCFs, such as C1ML and C2ML rules, are also appealing when studied through the lens of metric distortion as they have a reasonable worst-case metric distortion and their average-case metric distortion is only slightly worse than that of RSCFs that are tailored to minimize the metric distortion. 

\section*{Acknowledgments}

We thank three anonymous reviewers for their valuable feedback. This work was supported by the Deutsche Forschungsgemeinschaft under grants BR 2312/11-2 and BR 2312/12-1, and by the NSF-CSIRO grant on "Fair Sequential Collective Decision-Making" (RG230833).



\clearpage
\appendix



# A. Additional Simulations

In this section, we will show additional simulation results and statistics for our simulations. In particular, we conducted computer experiments for the following 5 standard models, for which we use the verified implementation provided by Boehmer et al. (2024).

**Impartial Culture (IC)** In this model, each voter is assigned a preference relation independently and uniformly at random. Hence, the probability to assign a preference relation $\succ \in \mathcal{R}(X_m)$ to a voter $v$ is $\frac{1}{m!}$.

**$t$0.5Dimensional Euclidean Cube ($t$EC)** In this model, we assign voters and alternatives independently and uniformly at random to points in the $t$-dimensional cube $[-1, 1]^t$. The voters' preference relations are then given by their distances to the alternatives: a voter $v$ prefers alternative $x$ to alternative $y$ if $|p_v - p_x|_2 < |p_v - p_y|_2$ where $p_v$, $p_x$, and $p_y$ denote the points of $v$, $x$, and $y$ in the $t$-dimensional cube. In our experiments, we use this model with $t \in \{1, 3, 5, 15\}$. We note that, for $t = 1$, the preference profiles generated from this model are single-peaked.

**$t$-Dimensional Euclidean Ball ($t$EB)** Analogous to the $t$-dimensional Euclidean cube model, we draw a point for each voter and each alternative independently and uniformly at random, but now from the $t$-dimensional ball $\{x \in \mathbb{R}^t \colon |x|_2 \leq 1\}$. Each voter $v$ again prefers alternative $x$ to alternative $y$ if $|p_v - p_x|_2 < |p_v - p_y|_2$ where $p_v$, $p_x$, and $p_y$ denote the points of $v$, $x$, and $y$. We use this model for $t \in \{3, 5, 15\}$.

**Normalized Mallow's Model ($\phi$MM)** Mallow's model Mallows (1957) is classically parameterized by a parameter $\rho \in [0, 1]$ and a preference relation $\succ$, and introduces a bias towards a common preference relation. In more detail, for every voter $v$ and every preference relation $\succ'$, the probability that voter $v$ is assigned $\succ'$ is $\frac{\rho^{|\succ \setminus \succ'|}}{Z}$ (where $Z = \sum_{\hat{\succ} \in \mathcal{R}(X_m)} \rho^{|\succ \setminus \hat{\succ}|}$ and $\succ \setminus \succ' = \{(x, y) \in X_m^2 \colon x \succ y \wedge y \succ' x\}$). However, as observed by Boehmer et al. (2023), the impact of the parameter $\rho$ depends on the number of alternatives since, e.g., a swap distance of 10 means that two rankings are completely inverse if there are only 5 alternatives but the rankings would be considered rather similar if there are 15 alternatives. Following the suggestion of Boehmer et al. (2023), we thus use the normalized Mallow's model, which is parameterized by the expected relative swap distance. In more detail, given a parameter $\rho$ and a preference relation $\succ^*$, we let $\mathbb{E}_{\succ \sim \rho MM(m)}[swap(\succ, \succ^*)]$ denote the expected swap distance between $\succ^*$ and a ranking chosen by Mallow's model with parameter $\rho$ and $\succ^*$ as ground ranking. Then, the normalized Mallow's model for a value $\phi$ determines the value $\rho^*$ such that $\mathbb{E}_{\succ \sim \rho^* MM(m)}[swap(\succ, \succ *)] \cdot \frac{4}{m(m-1)} = \phi$ and generates preference profile according to Mallow's model with parameter $\rho^*$. We use the normalized Mallow's model for the parameters $\phi \in \{0.5, 0.7, 0.9\}$ and the preference relation $\succ \; = x_1 \succ x_2 \succ \cdots \succ x_m$.

**Pólya-Eggenberger Urn Model ($\alpha$PE)** The Pólya-Eggenberger urn model is parameterized by a parameter $\alpha \in [0, 1]$ and works as follows (Berg, 1985; McCabe-Dansted and Slinko, 2006): we start with a urn where every preference relation is represented by a single ball. Then, we draw a ball uniformly at random, assign the corresponding preference relation to a voter, and return the ball together with $\alpha m!$ copies back to the urn. Hence, it gets more likely to draw the same preference relation again. This process is repeated until each voter is assigned a preference relation. We use this model for $\alpha \in \{0.05, 0.15, 0.3\}$.

For our computer experiments with these models, we use the same numbers as in the main body: for each of the 14 distributions (which includes the different parameterizations of the models) and each combination of $(m,n) \in \{5, 10, 15\} \times \{11, 21, \ldots, 201\}$, we sample 1000 preference profiles with $m$ alternatives and $n$ voters, compute the lotteries chosen by our RSCFs and their metric distortion. In the subsequent plots, we show the average metric distortion of the corresponding RSCF over the 1000 sample profiles for each distribution and each number of alternatives $m$ as a function of $n$. Moreover, to give a better understanding of our data, we also show the maximal and minimal metric distortion of an RSCF measured for those 1000 profiles as well as the variance of the metric distortion. While the maximum and minimum are not robust to outliers, they still allow for interesting observations as they show clear trends for many of our simulations.

We will subsequently discuss some additional observations of our computer experiments. Firstly, we note that the uniform random dictatorship $f_{RD}$ has in most of our simulations a minimal metric distortion of at least 2, i.e., this RSCF fails to provide a low metric distortion even if a metric distortion significantly below 2 is possible (as witnessed by the minimal metric distortion of other RSCFs). Moreover, in all models except for the impartial culture, the 1-dimensional Euclidean cube, and the Pólya-Eggenberger urn models, it also has typically the largest (or close to the largest) maximal metric distortion of the measured rules. This shows that, while $f_{RD}$ has a low metric distortion in the worst-case, it typically fails to actually find lotteries with a low metric distortion on more realistic profiles. On the other hand, because of this high minimal metric distortion, the variance of the metric distortion of $f_{RD}$ is typically the lowest.

Next, we will focus on some interesting effects for particular distributions. First, for the impartial culture model, we note that not only the average metric distortion converges to some value in (or close to) $[2, 2 + \frac{1}{m-1}]$, but also the minimum and maximum metric distortion converge to values within this range as $n$ increases. This is well-explained by the proof of Theorem 2 because, as the number of voters increases, it becomes more and more likely that all alternatives are equally good in a preference profile sampled from the impartial culture distribution. However, it is somewhat surprising that these convergence effects already take place for relatively small numbers of voters.

Next, we note that, in the 1-dimensional Euclidean cube model the average and maximum metric distortion of C1ML and C2ML rules is surprisingly high when $m = 5$. To explain this, we note that this model produces single-peaked preference profiles, so a Condorcet winner is guaranteed to exist and C1ML and C2ML rules will always pick this alternative. However, it is likely that there is another alternative that is only narrowly defeated by the Condorcet winner (i.e., the left or right neighbor of the Condorcet winner with respect to the single-peaked axis), which is sufficient for the Condorcet winner to have a metric distortion close to 3. We believe that this very well explains, e.g., the very smooth plot for the maximum metric distortion of C1ML and C2ML in the 1-dimensional Euclidean cube model when $m = 5$. However, as for all experiments with the Euclidean models, the metric distortion of C1ML and C2ML rules significantly decreases as the number of alternatives increases. Finally, we note that, in the 1-dimensional Euclidean cube model, the randomized Plurality-Veto rule performs best. Our explanation for this is that the randomized Plurality-Veto rule randomizes over alternatives close to the Condorcet winner, which exactly avoids the problem of C1ML and C2ML rules by randomizing over more alternatives. Lastly, the uniform random dictatorship is performing bad in this model as it puts too much probability on undesirable alternatives.

By contrast, for all other Euclidean models, we note that both the numbers and the observations made in the main body are still valid. That is, the number of dimensions and the choice between cube or ball has little influence. Additionally, it should be mentioned that all rules but the uniform random dictatorship obtain a very small minimum metric distortion in our samples, thus indicating that these rules manage

to find good compromise alternatives in some of the preference profiles. Even more, both the maximum and minimum metric distortion are largely independent of $n$ for these models. Lastly, in all of experiments with the Euclidean models, CRWW rules have a rather low maximum metric distortion (ranging between 2.2 and 2.4), thus demonstrating that this rule performs in general very well for Euclidean models. By contrast, the maximum metric distortion of C1ML and C2ML rules behaves analogously to their average as it decreases as the number of alternatives increases.

We next turn to our experiments for the (normalized) Mallow's model, where the results heavily depend on the chosen dispersion parameter $\phi$. First, if $\phi$ is large (i.e., for $\phi = 0.9$), the average metric distortion of all our RSCFs but the randomized Plurality-Veto rule perform is similar to the one for the impartial culture model. This is reasonable because, the larger the value $\phi$, the closer the Mallow's model becomes to the impartial culture distribution. The surprising exception here is the randomized Plurality-Veto rule which performs very well with $\phi = 0.9$. In particular, this rule actually has a lower average metric distortion when using the Mallow's model with $\phi = 0.9$ than with $\phi = 0.7$. We explain this as follows: for both $\Phi = 0.7$ and $\Phi = 0.9$, randomized Plurality-Veto randomizes over relatively large sets of alternatives (in average between $\frac{1}{3}m$ and $\frac{1}{2}m$). This is beneficial if $\Phi = 0.9$ as these profile lack structure. By contrast, if $\Phi = 0.7$, there is typically some strong candidates in the profile, and it is then detrimental to not put sufficient probability on this candidate. Next, we note that, while the CRWW rule performs (close to) best in all experiments with Mallow's model. As last point, we observe that, unsurprisingly, if we choose the dispersion parameter rather small (i.e., $\phi = 0.5$), all rules except for the uniform random dictatorship typically choose the top-ranked alternative of the base ranking, which results in a small average metric distortion. Based on our experiments, we believe that, when the voters' preferences can be modelled by Mallow's model, C1ML and C2ML rules are a rather desirable choice as their average metric distortion is very small.

Lastly, our observations for the Pólya-Eggenberger urn model also depend significantly on the chosen parameter $\alpha$. In general, the trend for this model is that, as $\alpha$ grows, the average metric distortion of all rules but the uniform random dictatorship decreases. By contrast, the average metric distortion of the uniform random dictatorship grows as $\alpha$ increases. The reason for this should be that, as $\alpha$ increases, it becomes more likely that there is a strong compromise alternative, which all rules but the uniform random dictatorship can identify. Moreover, we have also for this model (except for $\alpha = 0.3$) that the average metric distortion of C1ML and C2ML rules decreases as $m$ increases, whereas the average metric distortion of the randomized Plurality-veto rule increases in $m$. However, we actually find here the maximum metric distortion more interesting than the average. In particular, for all experiments with the Pólya-Eggenberge urn model (and especially for $\alpha = 0.3$), we find for each RSCF some profile where its metric distortion is close to its worst-case bound. We believe that this is due to the behavior of this model: when $\alpha$ is large, we effectively pick a few preference relations that will be reported by large groups of voters. For instance, if these preference relations all contain an alternative in the second-place but disagree on the first-ranked alternatives, the corresponding instance is similar to the worst-case instance for the uniform random dictatorship. Hence, if the voters' preferences can be described by the Pólya-Eggenberger urn model, worst-case instances may arise in practice. Conversely, we sample in our simulations with this model also profiles where all voters top-rank the same alternative and the metric distortion of all RSCFs is then 1. Consequently, the variance for our experiments with the urn model is rather high (up to 0.15 for C1ML and C2ML rules).

## A.1. Impartial Culture

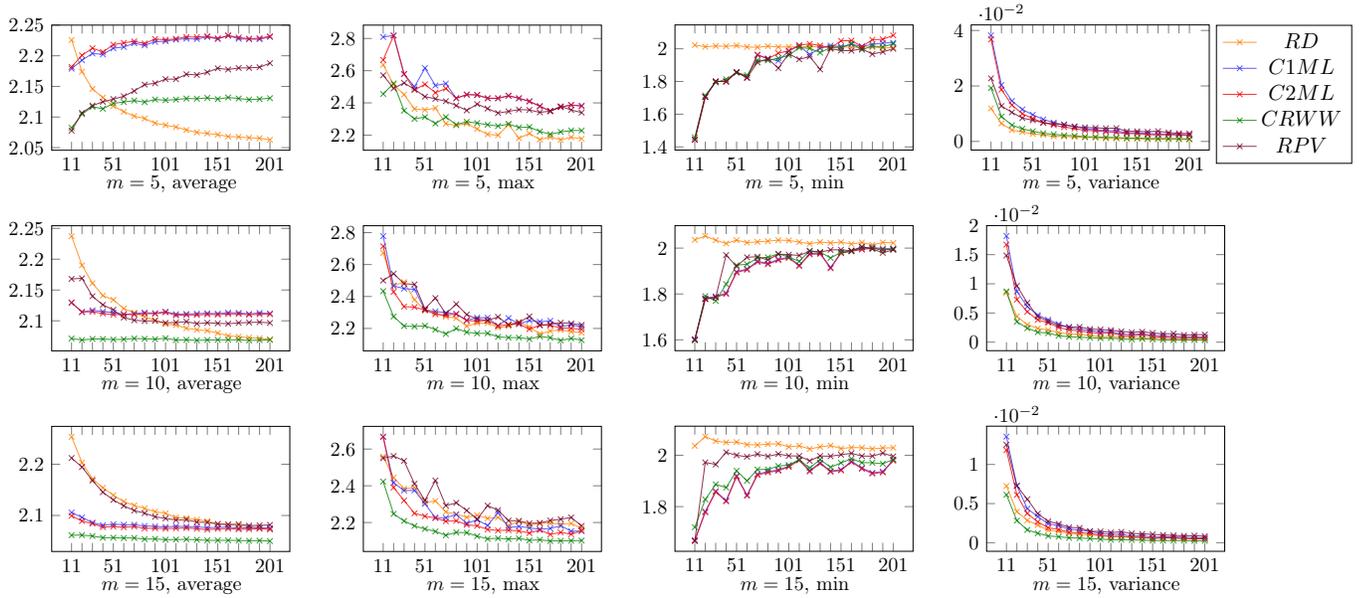

Figure 3: Results of our computer experiments for the impartial culture model. For each number of alternatives $m \in \{5, 10, 15\}$, we present a plot showing the average, maximum, minimum, and variance of the measured metric distortion of our five RSCFs in dependence of the number of voters.

## A.2. Euclidean Cube Model with $t = 1$ Dimension

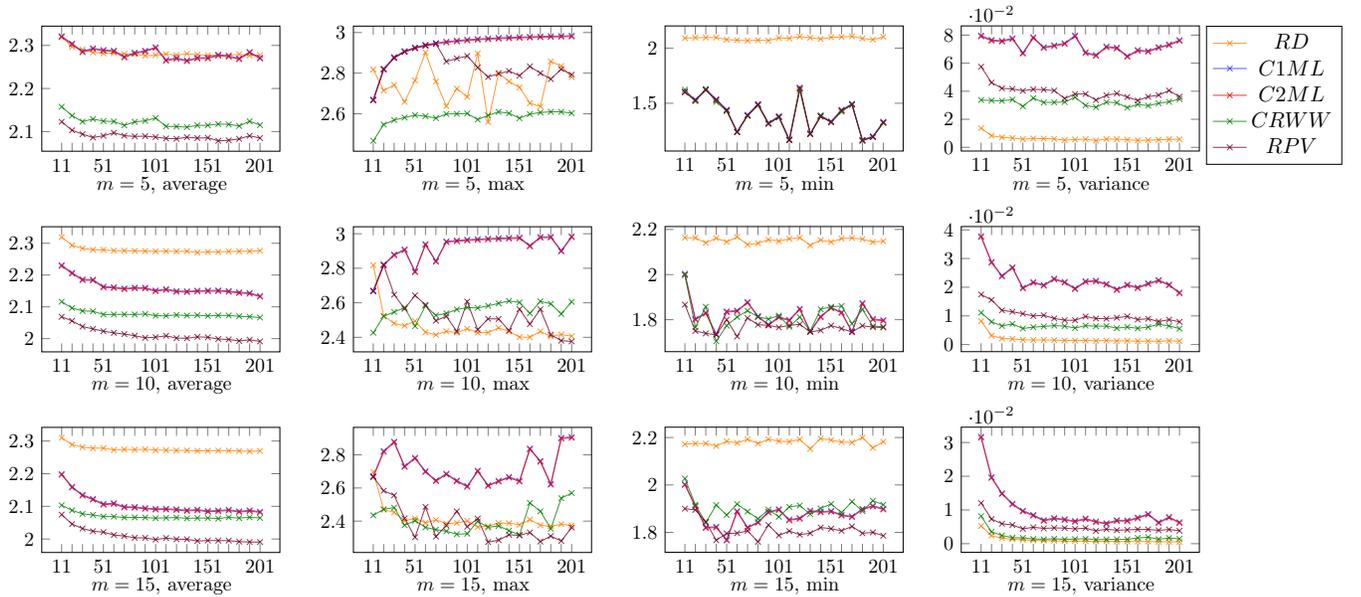

Figure 4: Results of our computer experiments for the Euclidean cube model with $t = 1$ dimension. For each number of alternatives $m \in \{5, 10, 15\}$, we present a plot showing the average, maximum, minimum, and variance of the measured metric distortion of our five RSCFs in dependence of the number of voters.

## A.3. Euclidean Cube Model with $t = 3$ Dimension

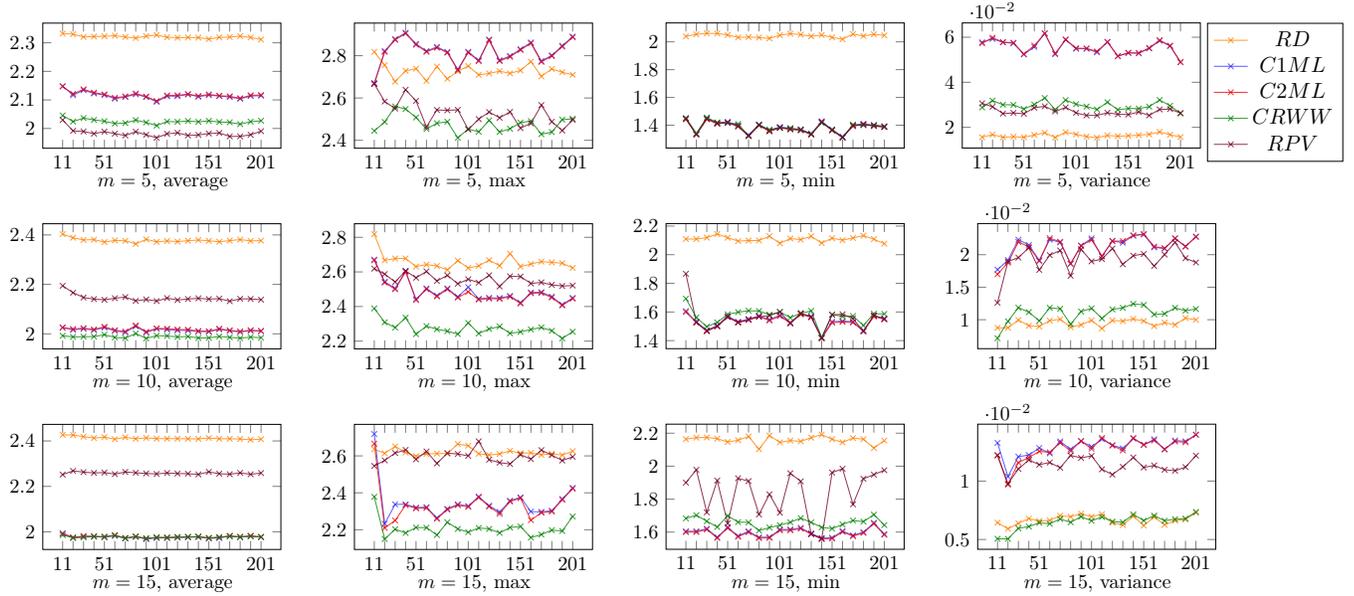

Figure 5: Results of our computer experiments for the Euclidean cube model with $t = 3$ dimensions. For each number of alternatives $m \in \{5, 10, 15\}$, we present a plot showing the average, maximum, minimum, and variance of the measured metric distortion of our five RSCFs in dependence of the number of voters.

## A.4. Euclidean Cube Model with $t = 5$ Dimensions

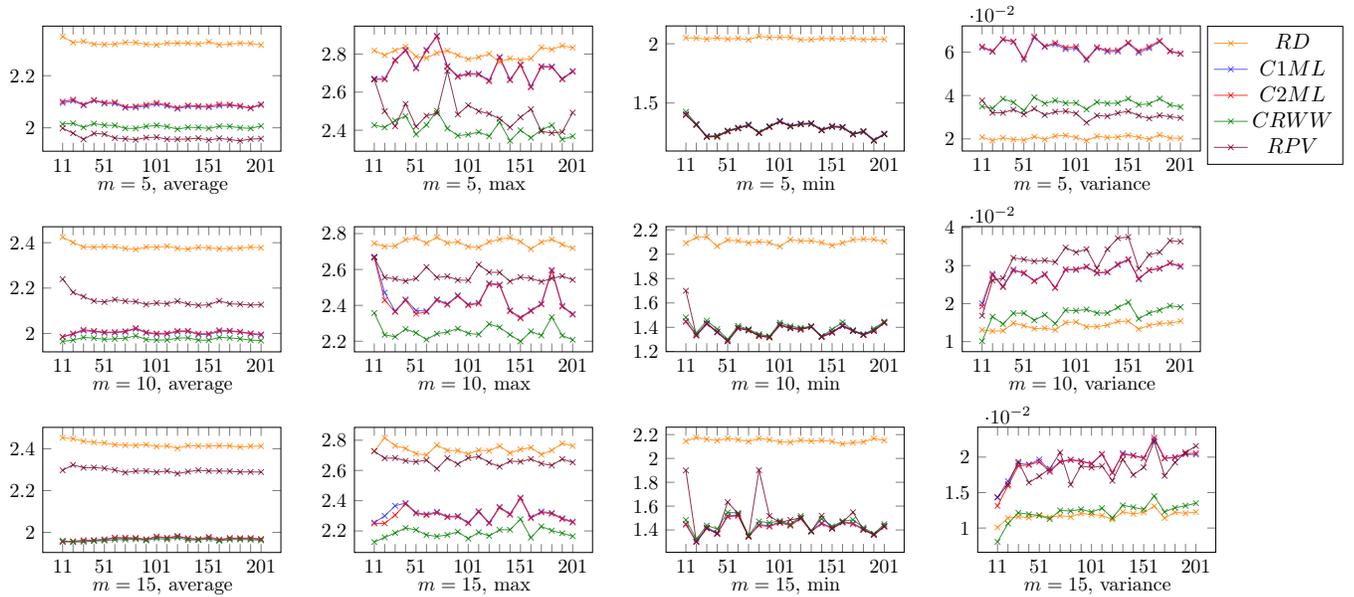

Figure 6: Results of our computer experiments for the Euclidean cube model with $t = 5$ dimensions. For each number of alternatives $m \in \{5, 10, 15\}$, we present a plot showing the average, maximum, minimum, and variance of the measured metric distortion of our five RSCFs in dependence of the number of voters.

## A.5. Euclidean Cube Model with $t = 15$ Dimensions

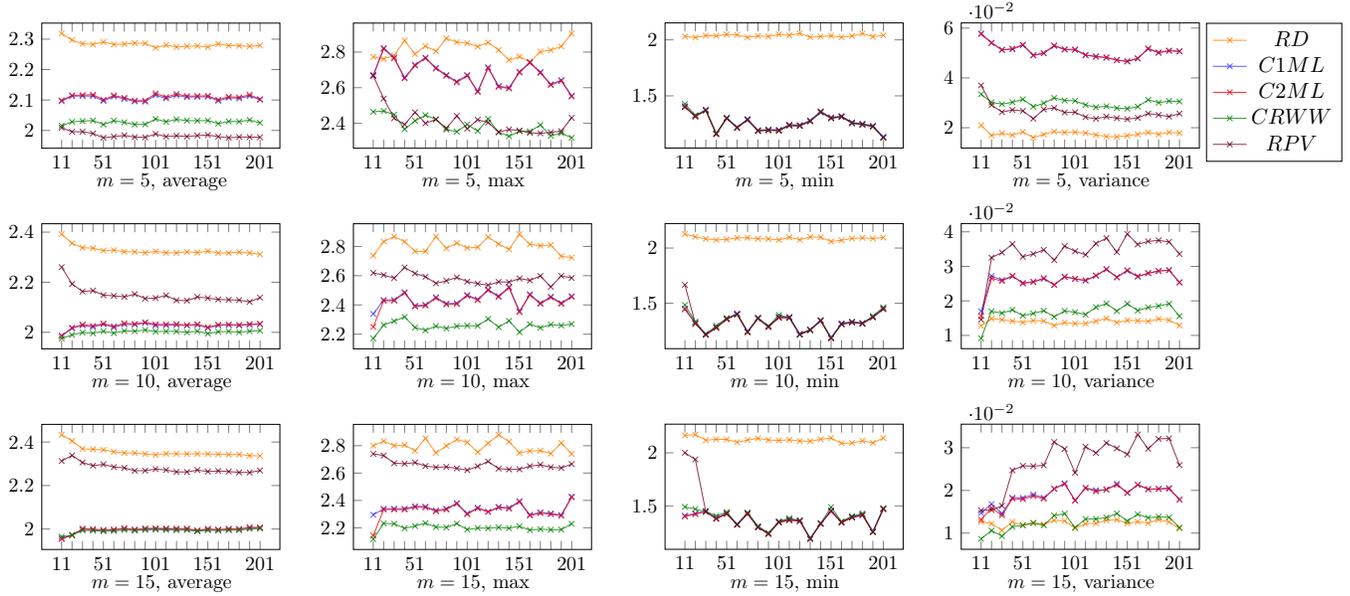

Figure 7: Results of our computer experiments for the Euclidean cube model with $t = 15$ dimensions. For each number of alternatives $m \in \{5, 10, 15\}$, we present a plot showing the average, maximum, minimum, and variance of the measured metric distortion of our five RSCFs in dependence of the number of voters.

## A.6. Euclidean Ball Model with $t = 3$ Dimensions

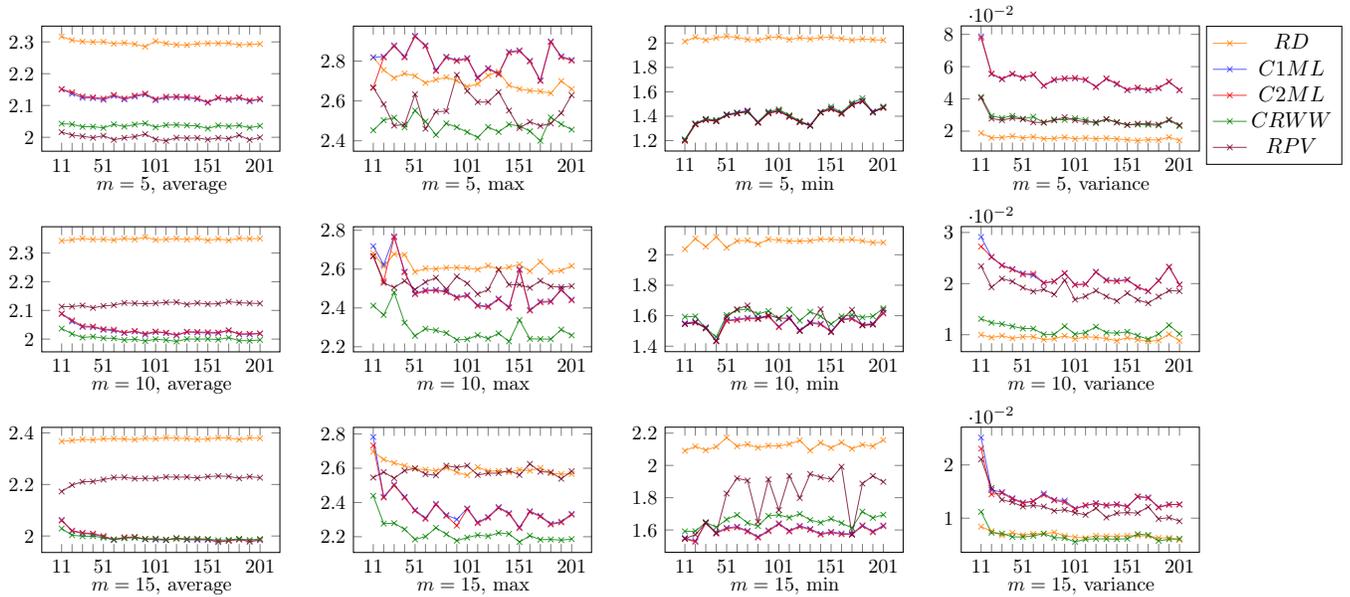

Figure 8: Results of our computer experiments for the Euclidean ball model with $t = 3$ dimensions. For each number of alternatives $m \in \{5, 10, 15\}$, we present a plot showing the average, maximum, minimum, and variance of the measured metric distortion of our five RSCFs in dependence of the number of voters.

## A.7. Euclidean Ball Model with $t = 5$ Dimensions

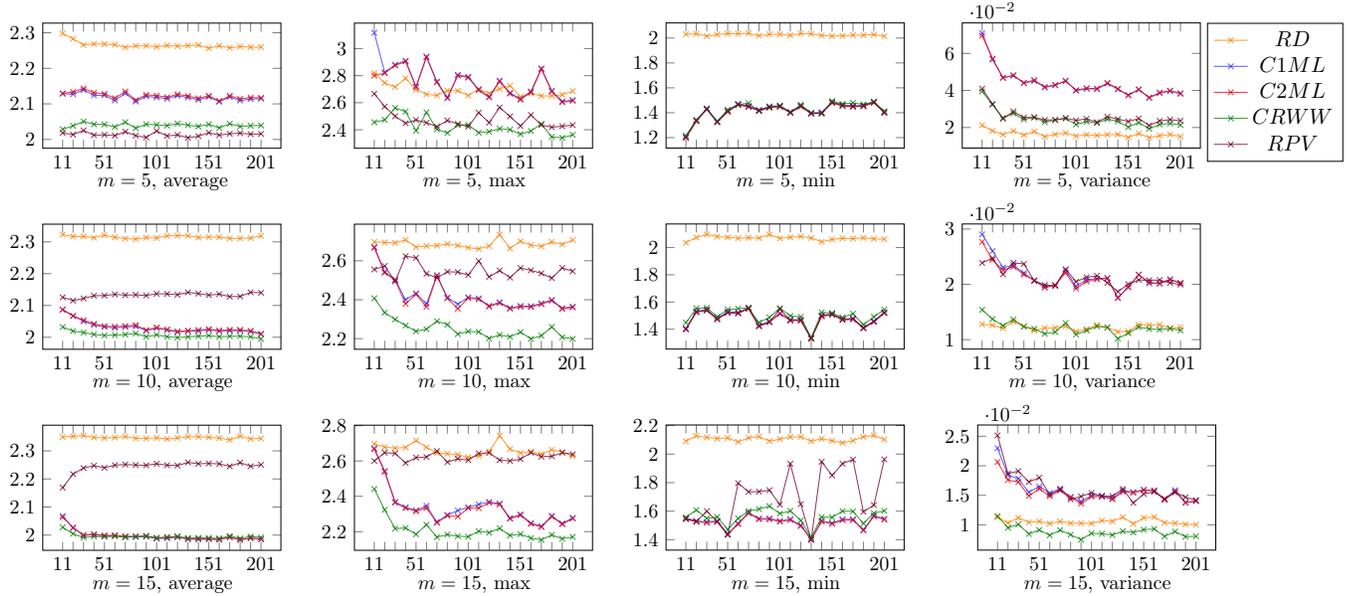

Figure 9: Results of our computer experiments for the Euclidean ball model with $t = 5$ dimensions. For each number of alternatives $m \in \{5, 10, 15\}$, we present a plot showing the average, maximum, minimum, and variance of the measured metric distortion of our five RSCFs in dependence of the number of voters.

## A.8. Euclidean Ball Model with $t = 15$ Dimensions

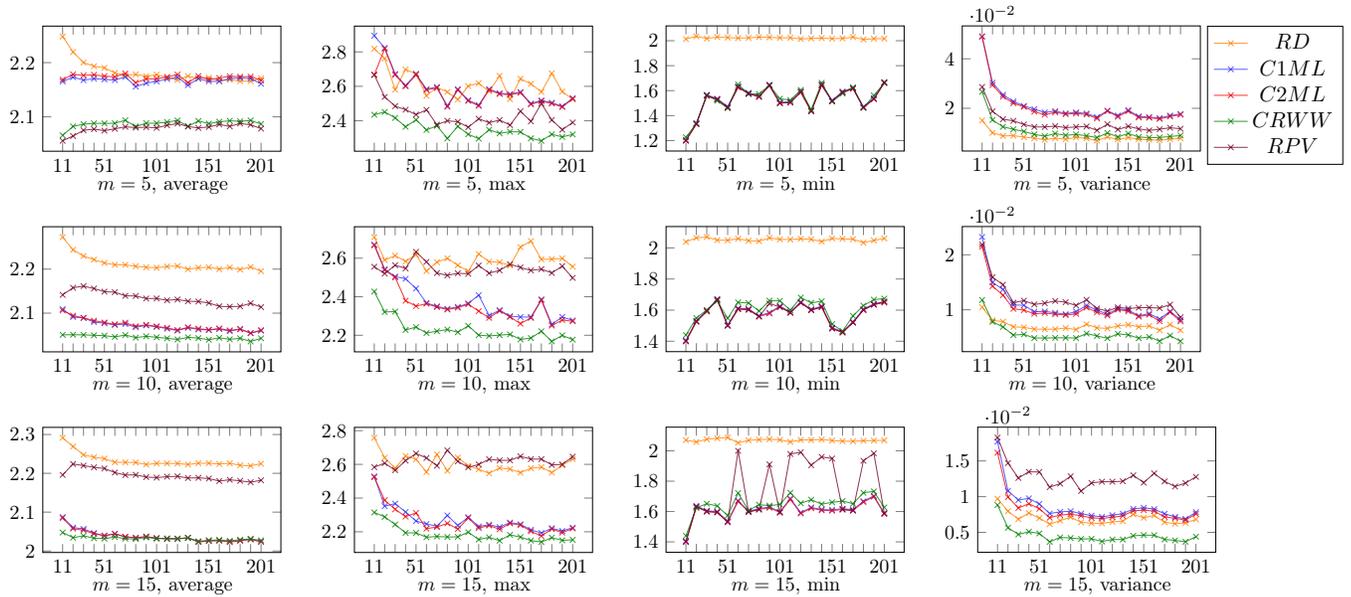

Figure 10: Results of our computer experiments for the Euclidean ball model with $t = 15$ dimensions. For each number of alternatives $m \in \{5, 10, 15\}$, we present a plot showing the average, maximum, minimum, and variance of the measured metric distortion of our five RSCFs in dependence of the number of voters.

### A.9. Mallow's Model with $\phi = 0.5$

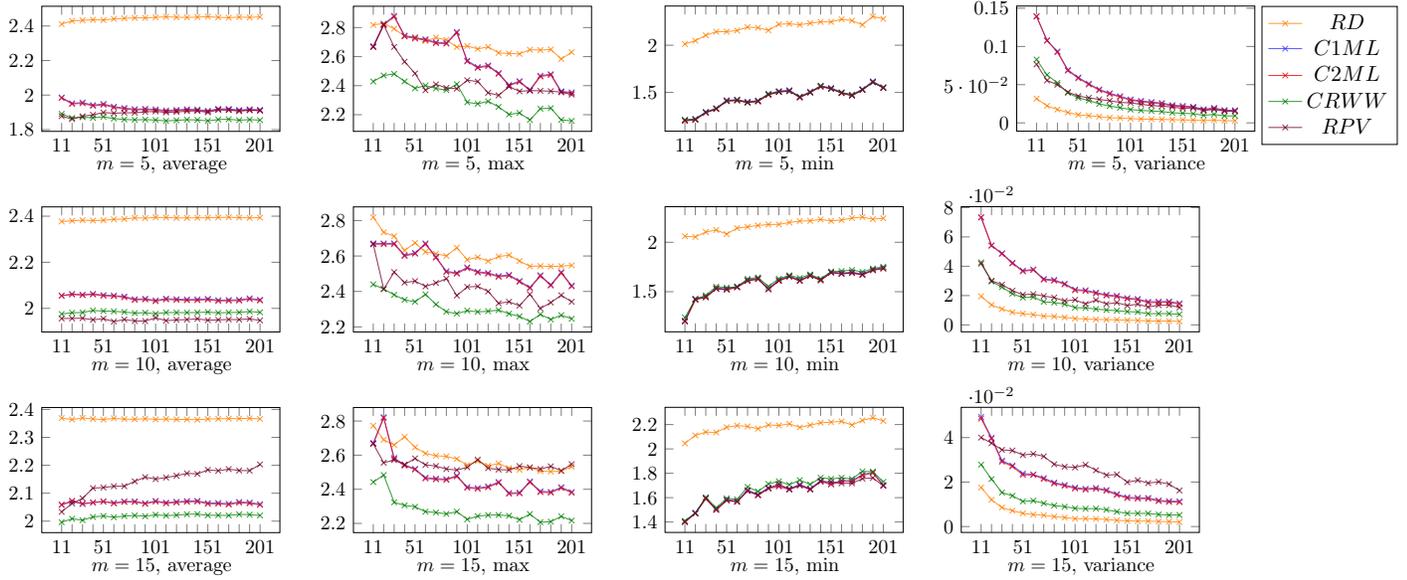

Figure 11: Results of our computer experiments for the (normalized) Mallow's model with $\phi = 0.5$. For each number of alternatives $m \in \{5, 10, 15\}$, we present a plot showing the average, maximum, minimum, and variance of the measured metric distortion of our five RSCFs in dependence of the number of voters.

### A.10. Mallow's Model with $\phi = 0.7$

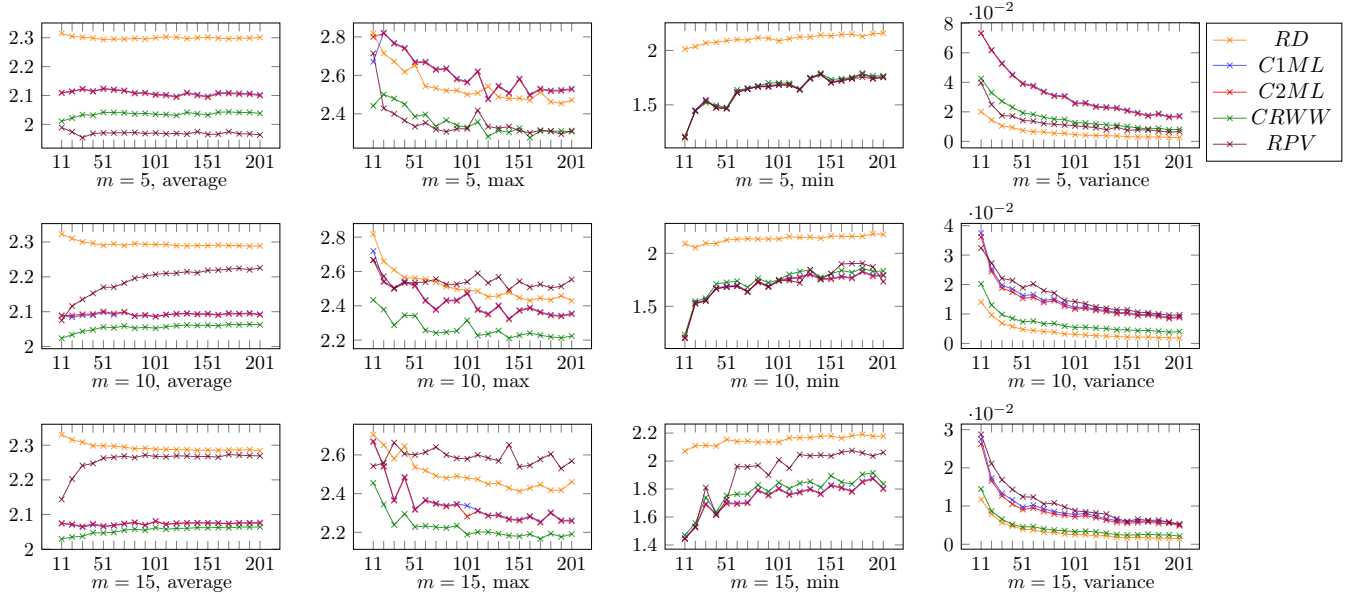

Figure 12: Results of our computer experiments for the (normalized) Mallow's model with $\phi = 0.7$. For each number of alternatives $m \in \{5, 10, 15\}$, we present a plot showing the average, maximum, minimum, and variance of the measured metric distortion of our five RSCFs in dependence of the number of voters.

## A.11. Mallow's Model with $\phi = 0.9$

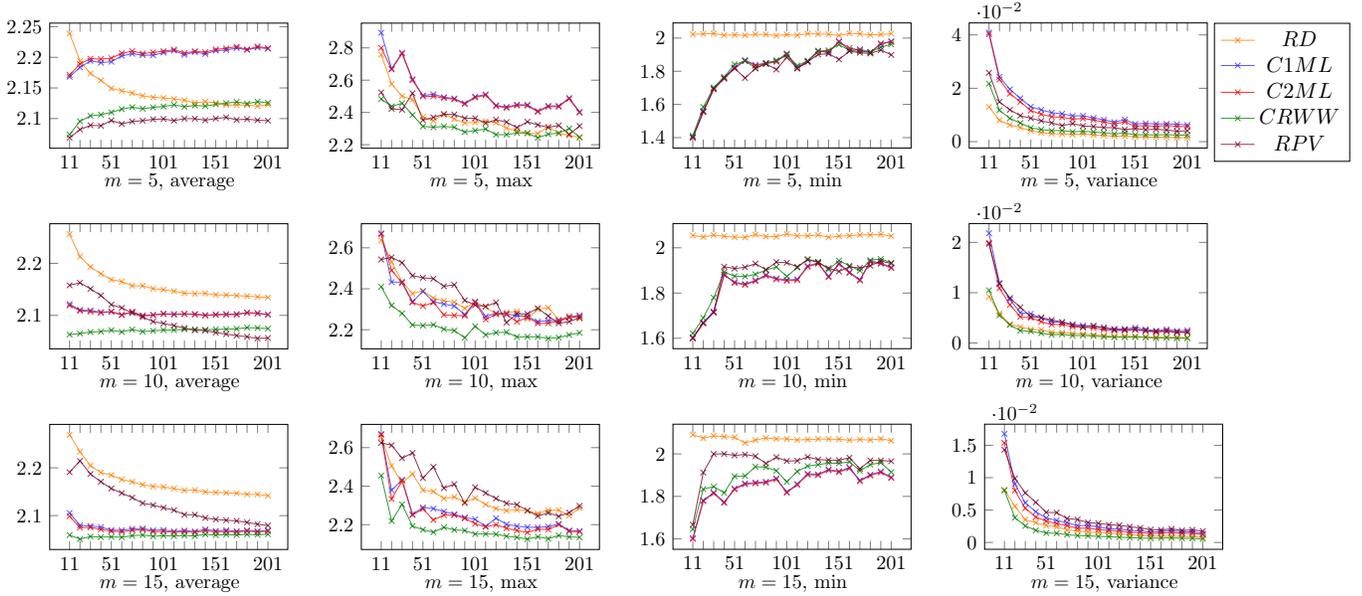

Figure 13: Results of our computer experiments for the (normalized) Mallow's model with $\phi = 0.9$. For each number of alternatives $m \in \{5, 10, 15\}$, we present a plot showing the average, maximum, minimum, and variance of the measured metric distortion of our five RSCFs in dependence of the number of voters.

## A.12. Pólya-Eggenberger Urn Model with $\alpha = 0.05$

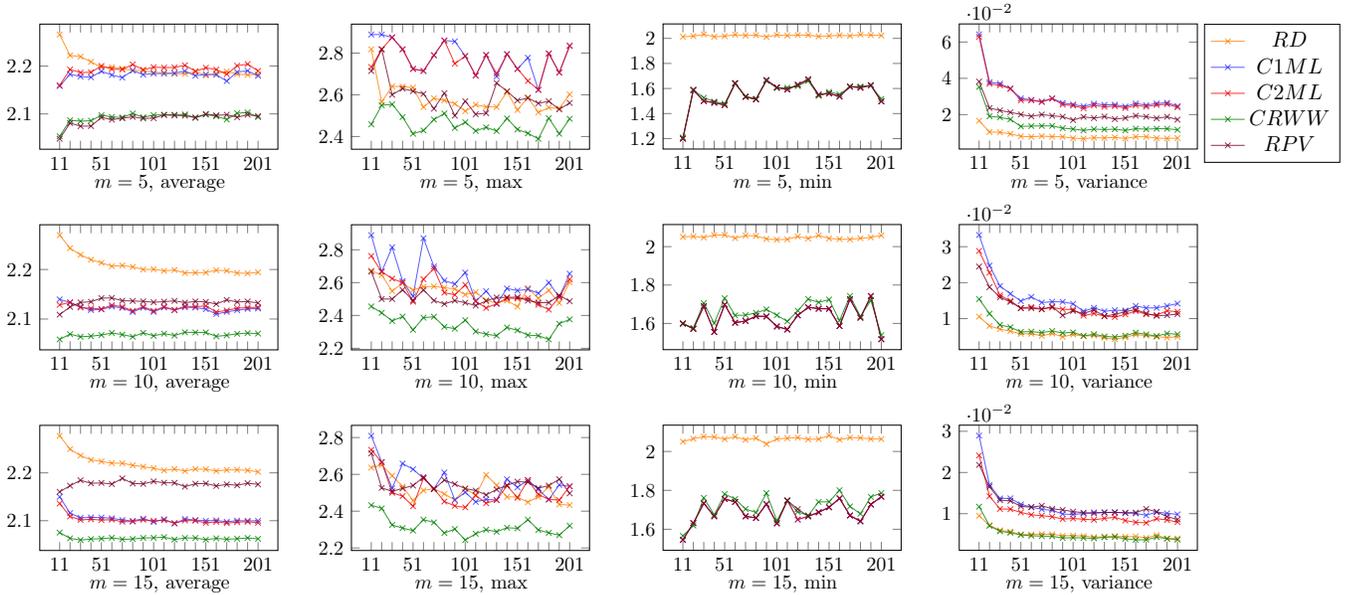

Figure 14: Results of our computer experiments for the Pólya-Eggenberger urn model with $\alpha = 0.05$. For each number of alternatives $m \in \{5, 10, 15\}$, we present a plot showing the average, maximum, minimum, and variance of the measured metric distortion of our five RSCFs in dependence of the number of voters.

## A.13. Pólya-Eggenberger Urn Model with $\alpha = 0.15$

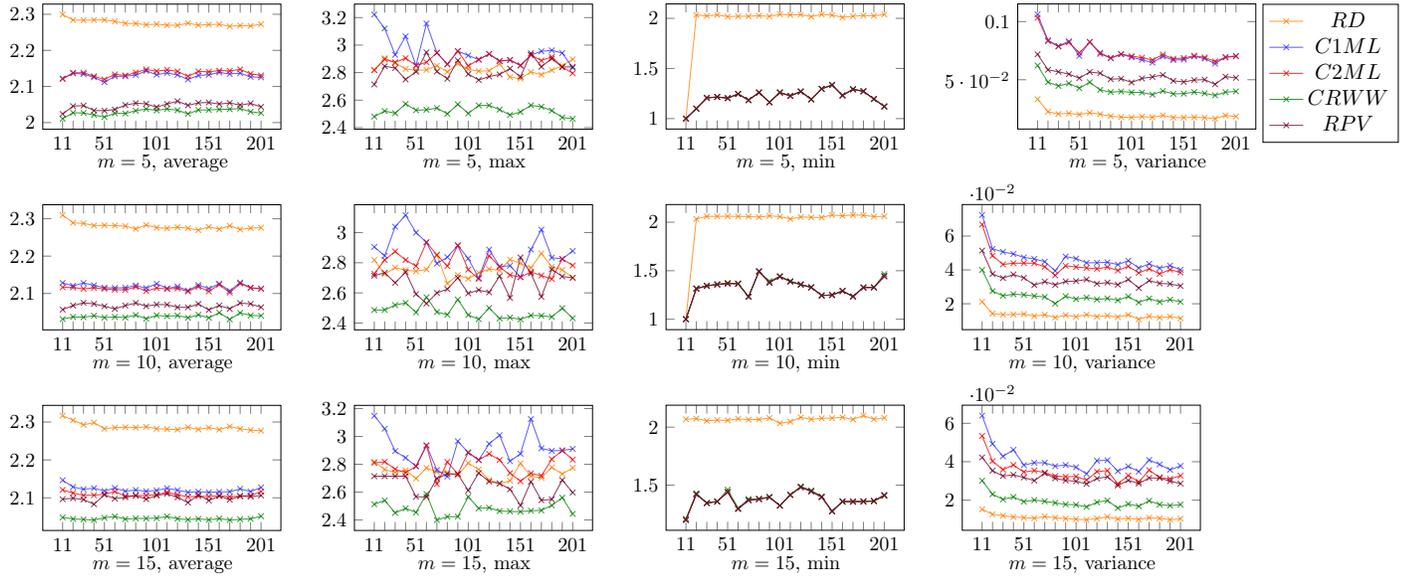

Figure 15: Results of our computer experiments for the Pólya-Eggenberger urn model with $\alpha = 0.15$. For each number of alternatives $m \in \{5, 10, 15\}$, we present a plot showing the average, maximum, minimum, and variance of the measured metric distortion of our five RSCFs in dependence of the number of voters.

## A.14. Pólya-Eggenberger Urn Model with $\alpha = 0.3$

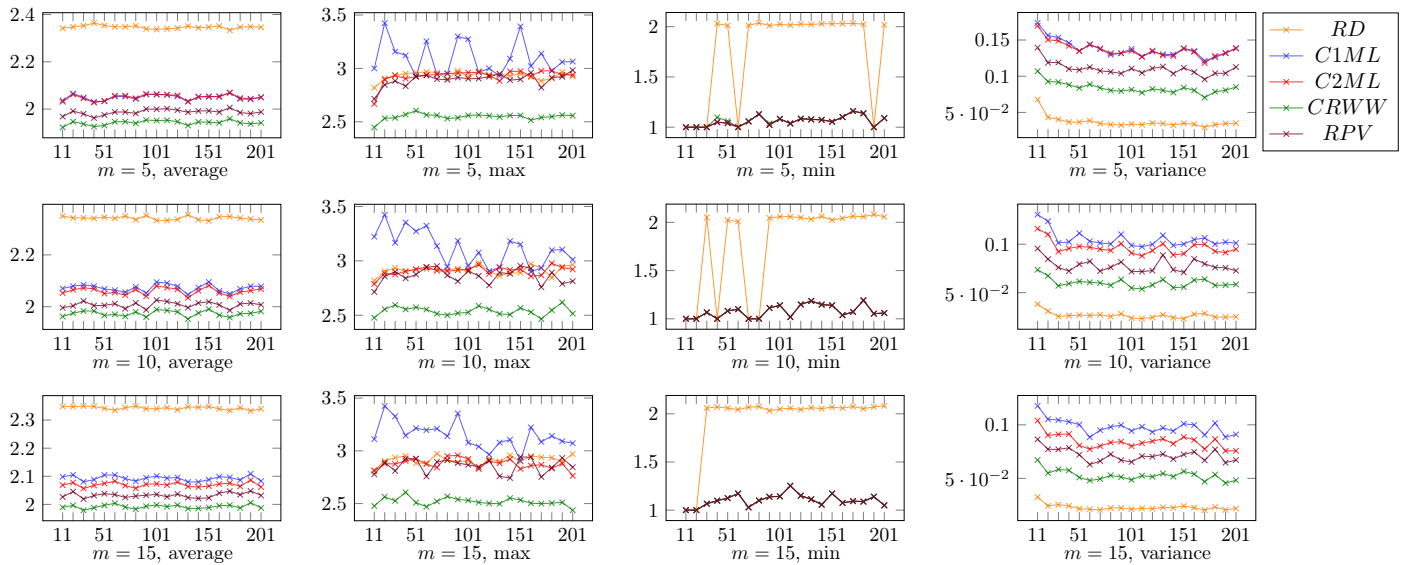

Figure 16: Results of our computer experiments for the Pólya-Eggenberger urn model with $\alpha = 0.3$. For each number of alternatives $m \in \{5, 10, 15\}$, we present a plot showing the average, maximum, minimum, and variance of the measured metric distortion of our five RSCFs in dependence of the number of voters.

## A.15. Spotify Daily

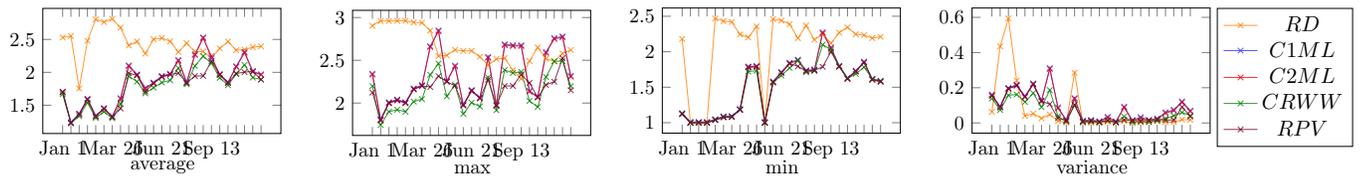

Figure 17: Detailed Results of our computer experiments with the Spotify Daily dataset. Each data point represents combines the outcome of 14 days for the given RSCF.


\setcounter{section}{1}

\section{Proofs from Section 3}

In this section, we present the proofs omitted from \Cref{subsec:C1ML}.
We start by showing \Cref{prop:pathlength}.

\pathlength*
\begin{proof}
	Let $f$ denote a majoritarian RSCF, $R$ an arbitrary profile, and $\succsim_R$ the corresponding majority relation. We will show the two claims of this proposition independently.\medskip
	
	\textbf{Proof of (1):} Our first goal is to show that $\mathit{dist}(f(R), R)\leq 1+2\max_{x\in X_R}\mathit{md}(f(R), x, {\succsim_R})$. To this end, we first note that, if $\max_{x\in X_R}\mathit{md}(f(R),x,{\succsim_R})=\infty$, there is nothing to show as $\mathit{dist}(f(R), R)\leq 1+2\max_{x\in X_R} \mathit{md}(f(R), x, {\succsim_R})=\infty$ holds trivially in this case. We hence assume that $\mathit{md}(f(R),x,{\succsim_R})<\infty$ for all $x\in X_R$, and we will show that $sc(x,d)\leq (1+2\mathit{md}(x,y,{\succsim_R})) sc(y,d)$ for every metric $d\in D(R)$ and all alternatives $x,y\in X_R$ such that $\mathit{md}(x,y,{\succsim_R})\neq\infty$. Since $f(R,x)>0$ implies that $\max_{y\in X_R}\mathit{md}(x,y,{\succsim_R})<\infty$, it then follows that 
	$\frac{\sum_{x\in X_R} f(R,x) sc(x,d)}{sc(y,d)}\leq \frac{\sum_{x\in A} f(R,x) (1+2\mathit{md}(x,y,{\succsim_R}))sc(y,d)}{sc(y,d)}=1+2\mathit{md}(f(R), y, {\succsim_R})$ for all metrics $d\in D(R)$, so $\mathit{dist}(f(R),R)\leq 1+2\max_{x\in X_R} \mathit{md}(f(R), x, {\succsim_R})$.
	
	To prove that $sc(x,d)\leq (1+2\mathit{md}(x,y,{\succsim_R})) sc(y,d)$ for all alternatives $x,y\in X_R$ with $\mathit{md}(x,y,{\succsim_R})\neq\infty$ and all metrics $d\in D(R)$, we use an induction on the majority distance between $x$ and $y$ in $\succsim_R$. First, if $\mathit{md}(x,y,{\succsim_R})=0$, then it clearly holds that $\frac{sc(x,d)}{sc(y,d)}=1$ as $\mathit{md}(x,y,{\succsim_R})=0$ only holds if $x=y$. Next, we assume for the induction hypothesis that there is some $k\in\mathbb{N}$ such that $sc(x',d)\leq (1+2\mathit{md}(x',y',{\succsim_R})) sc(y',d)$ for all metrics $d\in D(R)$ and alternatives $x',y'\in X_R$ with $\mathit{md}(x',y',{\succsim_R})\leq k$. For the induction step, we consider two alternatives $x,y\in X_R$ with $\mathit{md}(x,y,{\succsim_R})=k+1$ and an arbitrary metric $d\in D(R)$. Our goal is to show that $sc(x,d)\leq (1+2(k+1)) sc(y,d)$. To this end, let $z$ denote the successor of $x$ on a shortest path from $x$ to $y$ in $\succsim_R$, which means that $x\succsim_R z$ and ${\mathit{md}(z,y,{\succsim_R})=k}$. By the induction hypothesis, we can thus conclude that $sc(z,d)\leq (1+2k)sc(y,d)$. Next, we partition the voters $v\in N_R$ into the sets $N_{xz}=\{v\in N_R\colon x\succ_v z\}$ and $N_{zx}=\{v\in N_R\colon z\succ_v x\}$. Since $d\in D(R)$, it follows for all voters $v\in N_{xz}$ that $d(v,x)\leq d(v,z)$. Moreover, using the triangle inequality, we can show the following inequality for the voters $v\in N_{zx}$, where $v'$ is a voter in $N_{xz}$. 
	\begin{align*}
		d(v,x)&\leq d(v, y)+d(y,v')+d(v',x)\\
		&\leq d(v, y)+d(y,v')+d(v',z)\\
		&\leq d(v,y) + d(y,v')+d(v',y)+d(y,v)+d(v,z)\\
		&=2d(v,y)+2(v',y)+d(v,z)
	\end{align*}
	
	Finally, we observe that $|N_{xz}|\geq |N_{zx}|$ since $x\succsim_R z$, so there is an injective function $s$ from $N_{zx}$ to $N_{xz}$. Putting everything together, we infer the following inequality.
	\begin{align*}
		\sum_{v\in N_R} d(v,x)&=\sum_{v\in N_{xz}} d(v,x)+\sum_{v\in N_{zx}} d(v,x)\\
		&\leq \sum_{v\in N_{xz}} d(v,z)+\sum_{v\in N_{zx}} 2d(v,y)+2d(s(v),y)+d(v,z)\\
		&\leq \sum_{v\in N_R} d(v,z) + 2d(v,y)\\
		&=sc(z,d)+2sc(y,d)\\
		&\leq (1+2(k+1))sc(y,d)
	\end{align*}
	
	The first inequality follows from our bounds on $d(v,x)$ for $v\in N_{xz}$ and $v\in N_{zx}$, the second one simply reorganizes the terms and uses that $s$ is an injective function, and the last inequality follows by the induction hypothesis. This inequality proves the induction step, so it follows that $sc(x,d)\leq (1+2\mathit{md}(x,y,{\succsim_R})) sc(y,d)$ for all alternatives $x,y\in X_R$ with $\mathit{md}(x,y,{\succsim_R})<\infty$ and metrics $d\in D(R)$. This completes the proof of Claim (1).\medskip
	
	\textbf{Proof of (2):} As the second point, we will show that $\mathit{dist}_m(f)\geq 1+2\max_{x\in X_R} \mathit{md}(f(R),x,{\succsim_R})$. To this end, we use a case distinction with respect to whether $\max_{x\in X_R} \mathit{md}(f(R),x,{\succsim_R})<\infty$ or $\max_{x\in X_R} \mathit{md}(f(R),x,{\succsim_R})=\infty$.\smallskip
	
	\emph{Case 1:} First, we suppose that $\mathit{md}(f(R),x,{\succsim_R})<\infty$ for every alternative $x\in X_R$ and show that $\mathit{dist}_m(f)\geq 1+2\max_{x\in X_R} \mathit{md}(f(R),x,{\succsim_R})$. For this, we fix an arbitrary alternative $x^*\in X_R$; we will construct a family of profiles $R^\epsilon$ (where $\epsilon$ is a parameter in $(0,1)$) such that ${{\succsim_R}={\succsim_{R^\epsilon}}}$ for every $\epsilon\in (0,1)$ and $\lim_{\epsilon\rightarrow0}\mathit{dist}(f(R^\epsilon), R^{\epsilon})= 1+2\mathit{md}(f(R),x^*,{\succsim_R})$. To this end, let $D^k=\{{x\in X_R}\colon \mathit{md}(x,x^*,{\succsim_R})=k\}$ denote the set of alternatives that has a majority distance of $k$ to $x^*$. Moreover, we define $D^0=\{x^*\}$ and $D^m=\{y\in X_R\colon \mathit{md}(y,x^*,{\succsim_R})=\infty\}$ denotes the set of alternatives that have no path to $x^*$ in $\succsim_R$. We note that $x\succ_R y$ for all $x\in D^j$, $y\in D^{j'}$ such that $j+2\leq j'<m$ as otherwise, $y$ would have a path to $x^*$ of length $j+1< j'$ by going to $x$. Furthermore, $x\succ_R y$ for all $x\in X_R\setminus D^m$, $y\in D^m$ as there is a path from $x$ to $x^*$ in $\succsim_R$, but no such path exists for $y$. Based on this observation, we construct the following profile $R^\epsilon$ for $\epsilon \in (0,1)$, where $D^i\succ_v D^j$ denotes that voter $v$ prefers all alternatives in $D^i$ to all alternatives in $D^j$: 
	\begin{enumerate}
		\item There is a set of voters $I_1$ such that $|I_1|=\lceil\frac{1}{\epsilon}\rceil$ and $D^0\succ_v D^2\succ_v D^1\succ_v D^4\succ_v D^3\succ_v D^6\succ_v D^5\succ_v\dots \succ_v D^m$ for each $v\in I_1$. The alternatives within each set $D^i$ are ordered lexicographically.
		\item There is a set of voters $I_2$ such that $|I_2|=\lceil\frac{1}{\epsilon}\rceil$ and $D^1\succ_i D^0\succ_v D^3\succ_v D^2\succ_v D^5\succ_v D^4 \succ_v\dots\succ_v D^m$ for each $v\in I_2$. The alternatives within each set $D^i$ are ordered inverse lexicographically.
		\item For each pair of alternatives $x,y$ such that $x\succ_R y$ and $x\in D^j$, $y\in D^{j'}$ for $|j-j'|\leq 1$, we add two voters $v, v'$ with preferences $x\succ_v y\succ_v z_1\succ_v\dots\succ_v z_{m-2}$ and $z_{m-2}\succ_{v'}\dots\succ_{v'} z_1\succ_{v'} x\succ_{v'} y$. The set of these voters is called $I_3$ and we note that $|I_3|\leq m(m-1)$.
	\end{enumerate}
	
	We first note that the profile $R^{\epsilon}$ has indeed the same majority relation as $R$: the voters in $I_1$ and $I_2$ together enforce that a majority of voters prefers every alternative in $D^j$ to every alternative in $D^{j'}$ for all $j\in\mathbb{N}$, $j'\in\mathbb{N}\cup\{\infty\}$ with $j+2\leq j'$ and cancel each other out with respect to the majority comparison between every other pair of alternatives. Hence, the voters in $I^3$ set these majority comparisons in the same way as in $\succsim_R$, so ${\succsim_R}={\succsim_{R^\epsilon}}$. 
	
	Next, we define the following (partial) metric $d$ that is consistent with $R^{\epsilon}$: 
	\[d(v,x)=\begin{cases}
		2\lceil \frac{k}{2}\rceil &\text{if $v\in I_1$ and $x\in D^k$}\\
		1+2\lfloor \frac{k}{2}\rfloor\qquad &\text{if $v\in I_2$ and $x\in D^k$}\\
		m &\text{if $v\in I_3$}
	\end{cases}\]
	
	It can be checked that $d$ can be extended to a full metric on $V_{R^\epsilon}\cup X_{R^\epsilon}$. For instance, we may assume that the voters and alternatives are placed in a two-dimensional space such that every alternative $x\in D^k$ lies at $(-k,0)$ if $k$ is even and at $(k+1,0)$ if $k$ is odd. Moreover, the voters $i\in I_1$ all lie at $(0,0)$, the voters $i\in I_2$ lie at $(1,0)$, and the voters $i\in I_3$ lie at $(0,m)$. Then, $d$ corresponds to the $|\cdot|_\infty$ norm, which is known to be a metric.
	
	Finally, we can compute the social cost of our alternatives and the distortion of $f$. To this end, we note that $sc(y,d)=2\lceil{\frac{k}{2}}\rceil |I_1| + (1+2\lfloor{\frac{k}{2}}\rfloor) |I_2|+m|I_3|=(2k+1)\lceil\frac{1}{\epsilon}\rceil+m|I_3|$ for every alternative $y\in D^k$ and every $k$. In particular, this means that $sc(x^*,d)=\lceil\frac{1}{\epsilon}\rceil+m|I_3|$. Moreover, it holds that $f(R^\epsilon)=f(R)$ since ${\succsim_R}={\succsim_{R^\epsilon}}$ and $f$ is majoritarian. Next, because $\mathit{md}(f(R),x,{\succsim_R})<\infty$ for all $x\in X_R$, we can compute  for every $\epsilon\in(0,1)$ that 
	\begin{align*}
		\mathit{dist}_m(f)&\geq \mathit{dist}(f(R^\epsilon),R^\epsilon)\\
		&\geq \frac{\sum\limits_{y\in X_R} \!\!\!\!f(R,y) (1+2\mathit{md}(y,x,{\succsim_{R}}))\lceil{\frac{1}{\epsilon}}\rceil + m|I_3|}{\lceil{\frac{1}{\epsilon}}\rceil+m|I_3|}\\
		&=\frac{(1+2\mathit{md}(f(R), x, {\succsim_R}))\lceil{\frac{1}{\epsilon}}\rceil+m|I_3|}{\lceil{\frac{1}{\epsilon}}\rceil+m|I_3|}.
	\end{align*}
	
	It is easy to see that, when $\epsilon$ goes to $0$, the right side converges to $1+2\mathit{md}(f(R), x, {\succsim_R})$ as $m|I_3|$ is a constant. Finally, since $x$ is chosen arbitrarily, we thus infer that $\mathit{dist}_m\geq 1+2\max_{x\in X_R} \mathit{md}(f(R),x,{\succsim_R})$.\smallskip
	
	\emph{Case 2:} As the second case, we assume that $\max_{x\in X_R}\mathit{md}(f(R),x,{\succsim_R})=\infty$ and we will show that $\mathit{dist}_m(f)=\infty$, too. To this end, we let $x$ denote an alternative such that $\mathit{md}(f(R),x,{\succsim_R})=\infty$ and we define the sets $A=\{y\in X_R\colon \mathit{md}(y,x,{\succsim_R})< \infty\}$ and $B=\{y\in X_R\colon \mathit{md}(y,x,{\succsim_R}=\infty\}$. By the definition of the sets $A$ and $B$, it holds that $y\succ_R z$ for all $y\in A$ and $z\in B$. We will next use this observation to construct a profile $R'$ with ${\succsim_R}={\succsim_{R'}}$ such that $f$ has unbounded distortion in $R'$. To this end, we use a variant of McGarvey's construction \citet{McGa53a}: for all pairs of alternatives $y,z\in A$ or $y,z\in B$ with $y\succ_R z$, we add two voters who \emph{i)} both prefer all alternatives in $A$ to all alternatives in $B$, \emph{ii)} both prefer $y$ to $z$, and \emph{iii)} order all remaining pairs of alternatives exactly inverse. It can be checked that each pair of voters only ensures that $y\succ_R z$ for its respective pair of alternatives $y,z$, and that $x'\succ_R y'$ for all $x'\in A$, $y'\in B$. Hence, it is easy to see that ${\succsim_R}={\succsim_R'}$, which implies that $f(R')=f(R)$ as $f$ is majoritarian. Finally, consider the metric $d\in D(R')$ given by $d(v,x)=0$ and $d(v,y)=1$ for all $v\in V_{R'}$, $x\in A$, $y\in B$. It is easy to verify that every alternative $x\in A$ has a social cost $sc(x,d)=0$. By contrast, $sc(f(R'),d)=sc(f(R),d)>0$ as $f(R,y)>0$ for some alternative $y\in B$. Hence, $\mathit{dist}(f(R'),R')=\infty$, which proves this case. 
\end{proof}

Next, we turn to the proof of \Cref{thm:maj}. In particular, we note that Claim (1) of this theorem was proven in the main body, so we focus here only on Claim (2). 

\maj*
\begin{proof}
	To prove Claim (2) of this theorem, we will rely on Claim (2) of \Cref{prop:pathlength} and thus aim to construct a profile $R$ such that every lottery $p$ has a large expected majority distance $\mathit{md}(p,x,{\succsim_R})$ for some alternative $x$. To this end, we note that is suffices to construct a complete binary relation $\succsim$ on $X_m$ as we can find for every such relation a profile $R$ with ${\succsim_R}={\succsim}$ \citep{McGa53a}. 
	
	We first focus on the case that $m\geq 3$ is odd and consider in this case the ``cyclic'' majority relation defined by $x_i\succ x_{i+_m k}$ for all $i\in \{1,\dots, m\}$ and $k\in \{1,\dots, \frac{m-1}{2}\}$, where $i+_m k=i+k$ if $i+k\leq m$ and $i+_m k=i+k-m$ if $i+k>m$. Our goal is to show that $\max_{x\in A} \mathit{md}(p,x,{\succsim})\geq \frac{3}{2}-\frac{3}{2m}$ as Claim (2) in \Cref{prop:pathlength} then implies the theorem. We thus assume for contradiction that there is a lottery $p$ such that $\max_{x\in A}\mathit{md}(p,x,{\succsim})< \frac{3}{2}-\frac{3}{2m}$. Moreover, we define the lotteries $p^k$ by $p^k(x_i)=p(x_{i+_mk})$ for all $i,k\in \{1,\dots, m\}$ and first aim to show that $\max_{x\in A} \mathit{md}(p^k,x,{\succsim})<\frac{3}{2}-\frac{3}{2m}$, too. For this, we note that the symmetry of $\succsim$ implies that $\mathit{md}(x_i,x_j,{\succsim})=\mathit{md}(x_{i+_m k}, x_{j+_m k}, {\succsim})$ for all $i,j,k\in \{1,\dots, m\}$. Consequently, it holds that $\mathit{md}(p^k,x_i,{\succsim})=\mathit{md}(p,x_{i+_m k}, {\succsim})$ as $p^k(x_j)=p(x_{j +_m k})$ and $\mathit{md}(x_j, x_i, {\succsim})=\mathit{md}(x_{j+_m k}, x_{i+_m k}, {\succsim})$ for all $x_i, x_j\in X_m$. This implies that $\max_{x\in A} \mathit{md}(p^k,x,{\succsim})=\max_{x\in A} \mathit{md}(p,x,{\succsim})$. Finally, we consider the lottery $p^*$ defined by $p^*(x)=\frac{1}{m}\sum_{k\in \{1,\dots,m\}} p^k(x)$ for all $x\in X_m$ and observe that $\mathit{md}(p^*,x_i, {\succsim})=\frac{1}{m}\sum_{k\in \{1,\dots, m\}} \mathit{md}(p^k,x_i,{\succsim})<\frac{3}{2}-\frac{3}{2m}$ for all~$x_i$. However, $p^*(x_i)=\frac{1}{m}\sum_{k\in \{1,\dots, m\}} p^k(x_i)=\frac{1}{m}\sum_{k\in \{1,\dots, m\}} p(x_{i+_m k})=\frac{1}{m}$ for all $x_i$. Since $\mathit{md}(x_1,x_j,{\succsim})=2$ for all $j\in \{2,\dots, \frac{m+1}{2}\}$ and $\mathit{md}(x_1,x_j,{\succsim})=1$ for all $j\in \{\frac{m+3}{2},\dots, m\}$, we can thus compute that $\mathit{md}(p^*,x_1,{\succsim})= \frac{1}{m} \sum_{x_i\in X_R} \mathit{md}(x_i, x_1, {\succsim})=\frac{m-1}{2m}+\frac{2(m-1)}{2m}=\frac{3}{2}-\frac{3}{2m}$. This contradicts that $\mathit{md}(p^*,x_i, {\succsim})<\frac{3}{2}-\frac{3}{2m}$ for all $x_i$, so the initial assumption that there is a lottery $p$ with $\max_{x\in A}\mathit{md}(p,x,{\succsim})<\frac{3}{2}-\frac{3}{2m}$ is wrong. Hence, $\max_{x\in A} \mathit{md}(p,x,{\succsim})\geq \frac{3}{2}-\frac{3}{2m}$ for every lottery $p$ and \Cref{prop:pathlength} shows Claim (2) of this theorem for odd~$m\geq 3$.
	
	Finally, to extend the result also to even $m$, we can add an alternative $x^*$ that loses all majority comparisons. Based on Claim (2) in \Cref{prop:pathlength}, the metric distortion of a majoritarian RSCF is unbounded if it assigns positive probability to $x^*$. On the other side, we can apply the same analysis as for the case that $m$ is odd if $p(x^*)=0$ and hence infer our lower bound. 
\end{proof}

\section{Proof of Proposition 2}

Next, we will present the proof of \Cref{prop:LPcorrect}.

\LP*
\begin{proof}
	Let $R$ denote an arbitrary profile, $p$ a lottery, and $x^*$ denote an arbitrary alternative. We will prove the proposition in two steps: we first show that $\mathit{dist}(p,R,x^*)\geq o_{LP}$ for the objective value $o_{LP}$ of every feasible solution of \ref{LP} and then that $\mathit{dist}(p,R,x^*)\leq o_{LP}^*$ where $o_{LP^*}$ denotes the optimal objective value of \ref{LP} if this value is bounded and $o_{LP}^*=\infty$ otherwise. From the first insight, it follows immediately that $\mathit{dist}(p,R,x^*)=\infty$ if \ref{LP} is unbounded as we can find for every $x\in\mathbb{R}$ a feasible solution with higher objective value. On the other hand, combining the first and the second insight imply that $\mathit{dist}(p,R,x^*)= o_{LP}^*$ if the optimal objective value of \ref{LP} is bounded.\medskip
	
	\textbf{Claim (1): $\mathit{dist}(p,R,x^*)\geq o_{LP}$ for the objective value $o_{LP}$ of every feasible solution of \ref{LP}.}
	
	Let $d_{LP}$, $t_{LP}$ denote a feasible solution of \ref{LP} and let $o_{LP}$ denote its objective value. 
	To prove that $\mathit{dist}(p,R,x^*)\geq o_{LP}$, we will infer a metric $d\in D(R)$ from $d_{LP}$ such that $d(x^*,v)=d_{LP}(x^*,v)$ for all $v\in V_R$ and $d(x,v)\geq d_{LP}(x,v)$ for all $v\in v_R$ and $x\in X_R\setminus \{x^*\}$. 
	Since $\sum_{v\in V_R} d_{LP}(x^*,v)=1$, we can then infer that 
	\begin{align*}
		o_{LP}&=\sum_{x\in X_R} p(x)\sum_{v\in V_R} d_{LP}(x,v)
		\leq \frac{sc(p,d)}{sc(x^*,d)}\leq \max_{d\in D(R)} \frac{sc(p,d)}{sc(x^*,d)}
		=\mathit{dist}(p,R,x^*).
	\end{align*}
	
	Towards proving this claim, we will first construct another feasible solution $d_{LP}'$, $t_{LP}'$ with corresponding objective value $o'_{LP}$ that satisfies that $d_{LP}'(x,v)\geq d'_{LP}(x^*,v)$ for all $x\in X_R$, $v\in V_R$ and $o'_{LP}\geq o_{LP}$. Now, if $d_{LP}$ satisfies these conditions, we can simply set $d'_{LP}=d_{LP}$ and $t'_{LP}=t_{LP}$. 
	We thus assume that there is an alternative $x$ and a voter $v$ such that $d_{LP}(x,v)<d_{LP}(x^*,v)$. In this case, we consider the solution $\bar d_{LP}$ derived from $d_{LP}$ by setting $\bar d_{LP}(x,v)=d_{LP}(x^*,v)$. First, it is easy to verify that $\bar d_{LP}$ combined with the function $\bar t_{LP}=t_{LP}$ is still a feasible solution. Indeed, the only upper bounds on $\bar d_{LP}(x,v)$ are of the form $\bar d_{LP}(x,v)\leq \bar d_{LP}(x^*,v)+t(y)$, which are true since $\bar d_{LP}(x,v)= \bar d_{LP}(x^*,v)$ and $t(y)\geq 0$ for all $y\in X_R$. Moreover, it is straightforward that increasing the value of $d_{LP}(x,v)$ does not decrease the objective value. Hence, $\bar o_{LP}\geq o_{LP}$, and by repeating this step, we will arrive at a feasible solution $d_{LP}'$, $t_{LP}'$ such that $d_{LP}'(x,v) \geq d_{LP}'(x^*,v)$ for all alternatives $x\in X_R$ and voters $v\in V_R$. 
	
	As second step, we will again construct a feasible solution $d''_{LP}$, $t''_{LP}$ of \ref{LP} such that $o_{LP}''\geq o_{LP}$ and $d_{LP}''(x,v)\leq d_{LP}''(y,v)$ for all voters $v\in V_R$ and alternatives $x,y\in X_R$ with $x\succ_v y$. If $d'_{LP}$ satisfies this condition, we are immediately done and we hence suppose that there is a voter $v$ and two distinct alternatives $x$, $y$ such that $x\succ_j y$ and $d_{LP}'(x,v)> d_{LP}'(y,v)$. Note first that this is not possible if $y=x^*$ because the fourth condition of \ref{LP} ensures in this case that $d_{LP}'(x,v)\leq d'_{LP}(x^*,v)+t'_{LP}(x^*)=d'_{LP}(x^*,v)$. We hence assume from now on that $y\neq x^*$. In this case, we consider the solution $\bar d_{LP}$, $\bar t_{LP}$ derived from $d_{LP}'$, $t_{LP}'$ by setting $\bar d_{LP}(y,v)=d'_{LP}(x,v)$. First, we note that this solution is feasible as the only upper bounds on $\bar d_{LP}(y,v)$ are given by $\bar d_{LP}(y,v)\leq \bar d_{LP}(x^*,v)+\bar t(z)=d'_{LP}(x^*,v)+t'_{LP}(z)$ for $z\in X_R$ with $y\succeq_v z$. Moreover, it holds that $\bar d_{LP}(x,v)=d'_{LP}(x,v)\leq d_{LP}'(x_{i^*},v)+t'_{LP}(z)$ for all $z\in X_R$ with $x\succeq_v z$ since $d_{LP}'$, $t'_{LP}$ is a feasible solution of \ref{LP}. Finally, since $x\succeq_v y$, it therefore follows that $\bar d_{LP}$ is a feasible solution, too. Moreover, it is again straightforward that we did not decrease the objective value because we only increased the value of variables. Now, by repeating this step, it is easy to see that we will eventually arrive at a feasible solution $d''_{LP}$ and $t''_{LP}=t_{LP}$ such that $o''_{LP}\geq o'_{LP}$ and $d''_{LP}(x,v)\leq d''(y,v)$ for all $v\in V_R$ and $x, y\in X_R$ with $x\succ_v y$. Moreover, $d''_{LP}$ still satisfies that $d''_{LP}(x,v)\geq d''_{LP}(x^*,v)$ for all $v\in V_R$ and $x_i\in X_R$ as we only increase the distances for alternatives $x\neq x^*$. 
	
	Finally, based on the solution $d''_{LP}$, $t''_{LP}$, we will construct a metric $d$ that satisfies all our criteria. In particular, we define:
	\begin{enumerate}
		\item $d(x,v)=d(v,x)=d''_{LP}(x,v)$ for all $x\in X_R$ and $v\in V_R$.
		\item $d(x, x)=0$ for all $x\in X_R$ and $d(v,v)=0$ for all $v\in V_R$. 
		\item $d(x,y)=\min_{v\in V_R} d''_{LP}(x,v)+d''_{LP}(y,v)$ for all distinct $x,y\in X_R$.
		\item $d(v,w)=\min_{x\in X_R} d''_{LP}(x,v)+d''_{LP}(x,w)$ for all distinct $v,w\in V_R$.
	\end{enumerate}
	
	By its definition, it is straightforward that $d$ is symmetric and that $d(z,z)=0$ for all $z\in X_R\cup V_R$. Moreover, because $d''_{LP}$ is consistent with $R$, the same holds for $d$. Hence, we only need to verify the triangle inequality, for which we start by an auxiliary observation: we will show that $d(x,v)\leq d(x,w)+d(y,w)+d(y,v)$ for all $x,y\in X_R$, $v,w\in V_R$. By the definition of $d$, this is equivalent to proving the same for $d''_{LP}$. We thus observe that 
	\begin{align*}
		d''_{LP}(x,v)&\leq d''_{LP}(x^*,v)+t''_{LP}(x)\\
		&\leq d''_{LP}(x^*,v)+d''_{LP}(x^*,w)+d''_{LP}(x,w)\\
		&\leq d''_{LP}(y,v)+d''_{LP}(y,w)+d''_{LP}(x,w).
	\end{align*}
	The first and second inequality directly use the third and fifth constraint of our LP. The last inequality uses that, by construction of $d''_{LP}$, it holds that $d''_{LP}(x^*,v)\leq d''_{LP}(y,v)$ and $d''_{LP}(x^*,w)\leq d''_{LP}(y,w)$. 
	
	Finally, we are ready to show that $d$ satisfies the triangle inequality. To this end, consider three distinct elements $x,y,z\in X_R\cup V_R$. We will show that $d(x,z)\leq d(x,y)+d(y,z)$ by considering three cases:
	\begin{itemize}
		\item $x,y,z\in X_R$: Let $v,w\in V_R$ denote the voters that minimize $d(x,v)+d(v,y)$ and $d(y,w)+d(w,z)$, respectively. By our auxiliary claim, it holds that $d(x,z)=\min_{v'\in V_R} d(x,v')+d(v',z)\leq d(x,v)+d(z,v)\leq d(x,v)+d(z,w)+d(w,y)+d(y,v)=\min_{v'\in V_R} d(x,v')+d(v',y)+ \min_{v'\in V_R} d(y,v')+d(v',z)=d(x,y)+d(y,z)$. An analogous argument works if $x,y,z\in V_R$.
		\item $x,y\in X_R$, $z\in V_R$: Let $v$ denote the voter that minimizes $d(x,v)+d(v,y)$. By our auxiliary claim, it holds that $d(x,z)\leq d(x,v)+d(v,y)+d(y,z)=d(x,y)+d(y,z)$. The cases that $y,z\in X_R$, $x\in V_R$; $x,y\in V_R$, $z\in X_R$; and $y,z\in V_R$, $x\in X_R$ are symmetric. 
		\item $x,z\in X_R$, $y\in V_R$: It holds that $d(x,z)=\min_{v\in N} d(x,v)+d(v,y)\leq d(x,y)+d(y,z)$. The case that $x,z\in V_R$, $y\in X_R$ is symmetric. 
	\end{itemize}
	
	This proves that $d$ is indeed a metric that is consistent with $R$. We can therefore conclude that $\mathit{dist}(p,R,x^*)\geq \frac{sc(p,d)}{sc(x^*,d)}=o''_{LP}\geq o_{LP}$ holds for all feasible solutions $d_{LP}$, $t_{LP}$ with objective value $o_{LP}$.\medskip
	
	\textbf{Claim (2): $\mathit{dist}(p,R,x_{i^*})\leq o_{LP}^*$ where $o_{LP}^*$ is the optimal objective value off \ref{LP}.}
	
	We will next show that $\mathit{dist}(p,R,x_{i^*})\leq o_{LP}^*$. To this end, we note that this is trivial if $o_{LP}^*=\infty$, so we focus on the case that the optimal objective value of \ref{LP} is bounded. To this end, let $d\in D(R)$ denote a metric that maximizes $\frac{sc(p,d)}{sc(x^*, d)}$. We will next construct a biased metric $d^*\in D(R)$ that satisfies $\frac{sc(p,d^*)}{sc(x^*, d^*)}\geq \frac{sc(p,d)}{sc(x^*, d)}$. As second step, we will then derive a feasible solution $d_{LP}$, $t_{LP}$ of \ref{LP} with objective value $o_{LP}=\frac{sc(p,d^*)}{sc(x^*, d^*)}$. This clearly proves the claim. 
	
	Following the proof of \citet{ChRa22a}, we define the function $t(x)$ for all $X_R$ by $t(x)=d(x,x^*)$. The biased metric $d^*$ is then defined by 
	\begin{align*}
		d^*(x^*,v)&=\frac{1}{2}\max_{x,y\in X_R\colon x\succeq_v y} t(x)-t(y)\\
		d^*(x,v)&=d^*(x^*,v)+\min_{y\in X_R\colon x\succeq_v y} t(y).
	\end{align*}
	
	We first note that $d^*$ can be extended to a metric that is consistent with $R$ due to Proposition 5.1 of \citet{ChRa22a}. Hence, it only remains to show that $\frac{sc(p,d^*)}{sc(x^*, d^*)}\geq \frac{sc(p,d)}{sc(x^*, d)}$. To this end, we will show that $sc(x^*,d^*)\leq sc(x^*,d)$ and $sc(x,d^*)-sc(x^*,d^*)\geq sc(x,d)-sc(x^*,d)$. This shows $\frac{sc(p,d^*)}{sc(x_{i^*}, d^*)}\geq \frac{sc(p,d)}{sc(x^*, d)}$ as demonstrated by the following inequality. 
	
	\begin{align*}
		\frac{sc(p,d^*)}{sc(x^*, d^*)}-1&=\frac{\sum_{x\in X_R} p(x)(sc(x,d^*)-sc(x^*, d^*))}{sc(x^*, d^*)}
		 \geq \frac{\sum_{x\in X_R} p(x)(sc(x,d)-sc(x^*, d))}{sc(x^*, d)}
	=\frac{sc(p,d)}{sc(x^*, d)}-1
	\end{align*}
	
	We first show that $sc(x^*,d^*)\leq sc(x^*,d)$. To this end, we observe (analogous to \citet{ChRa22a} in Proposition 5.2) that $d(x,x^*)\leq d(x,v)+d(v,x^*)\leq d(y, v)+d(v, x^*)\leq d(y, x^*)+2d(v, x^*)$ for all voters $v$ and alternatives $x$, $y$ with $x\succeq_v y$. Hence, $d(v,x^*)\geq \frac{1}{2}\max_{x,y\in X_R\colon x\succeq_v y} t(x)-t(y)=d^*(v,x^*)$. Clearly, this implies that $sc(x^*, d^*)\leq sc(x^*,d)$, thus proving our claim. Secondly, we need to prove that $sc(x,d^*)-sc(x^*,d^*)\geq sc(x,d)-sc(x^*,d)$ for all $x\in X_R$. Since the inequality clearly holds for $x^*$, we assume that $x\neq x^*$. Following again the ideas of \citet{ChRa22a}, we observe that $d(x,v)\leq d(y,v)\leq d(y,x^*)+d(x^*,v)$ for all voters $v$ and alternatives $x, y$ with $x\succeq_v y$. Hence, $d(x,v)-d(x^*,v)\leq \min_{y\in X_R\colon x\succeq_v y} d(y, x^*)=\min_{y\in X_R\colon x\succeq_v y} t(y)=d^*(x,v)-d^*(x^*,v)$. We thus conclude that $sc(x,d^*)-sc(x^*,d^*)\geq sc(x,d)-sc(x^*,d)$. Therefore, it follows indeed that $\frac{sc(p,d^*)}{sc(x^*, d^*)}\geq \frac{sc(p,d)}{sc(x^*, d)}$.
	
	We next proceed with a case distinction with respect to whether $sc(x^*,d^*)=0$ or $sc(x^*,d^*)>0$. First, we consider the case that $sc(x^*,d^*)>0$. In this case, we aim to construct a feasible solution $d_{LP}$, $t_{LP}$ of \ref{LP} with objective value $o_{LP}=\frac{sc(p,d^*)}{sc(x^*,d^*)}$. Now, to derive this solution, we first note that every biased metric $d\in D(R)$ (together with its inducing function $t$) satisfies the first four constraints of \ref{LP} by definition. Moreover, $d$ also satisfies the fifth constraint since $d(x,v)+d(x^*,v)=2d(x^*,v)+\min_{y\in X_r\colon x\succeq y} t(y)\geq t(x)$ for all $x\in X_R$, $v\in V_R$. The last inequality follows as $2(d^*,v)=\max_{x,y\in X_R\colon x\succeq_v y} t(x)-t(y)\geq t(x)-\min_{y\in X_r\colon x\succeq y} t(y)$. Furthermore, we note that, for every biased metric $d\in D(R)$, and $\ell\in \mathbb{R}_{>0}$, the function $t^\ell$ defined by $t^\ell(x)=\ell t(x)$ induces a biased metric $d^\ell\in D(R)$ with $sc(x^*,d^\ell)=\ell sc(x^*,d)$ and $sc(p,d^\ell)=\ell sc(p,d)$. Because $sc(x^*,d)> 0$, it is thus easy to check that the biased metric $d^\ell$ together with its defining function $t^\ell$ for $\ell=\frac{1}{sc(x^*,d)}$ defines a feasible solution to \ref{LP} with $o_{LP}=\frac{sc(p,d^\ell)}{sc(x^*,d^\ell)}=\frac{sc(p,d^*)}{sc(x^*,d^*)}$. Hence, $\mathit{dist}(p,R,x^*)=\frac{sc(p,d^*)}{sc(x^*,d^*)}=o_{LP}\leq o^*_{LP}$, where $o^*_{LP}$ denotes the optimal objective value of \ref{LP}. 
	
	For the second case, we suppose that $sc(x^*,d^*)=0$. For this case, we make a further case distinction with respect to whether $sc(p,d^*)=0$ or $sc(p,d^*)>0$. First, suppose that $sc(p,d^*)=0$, which means that $\mathit{dist}(p,R,x^*)=\frac{sc(p,d^*)}{sc(x^*,d^*}=1$. To show that $\mathit{dist}(p,R,x^*)\leq o_{LP}^*$, it thus suffices to construct a feasible solution of \ref{LP} with objective value $1$. To this end, consider the following solution: $d_{LP}(x,v)=\frac{1}{n_R}$ for all $x\in X_R$, $v\in V_R$ and $t_{LP}(x)=0$ for all $x\in X_R$. It is easy to check that this is indeed a feasible solution and that $\sum_{x\in X_R} p(x) \sum_{v\in V_R} d(x,v)=\sum_{x\in X_R} p(x)=1$, thus verifying our claim. 
	
	As the last case, we assume that $sc(x^*,d^*)=0$ and $sc(p,d^*)>0$, which means that $\mathit{dist}(p,R,x^*)=\infty$. We need to show that the optimal objective value of \ref{LP} is unbounded. Towards this end, we note that $d^*(x^*,v)=0$ for all voters $v\in V_R$ since $sc(x^*,d^*)=0$. Next, we consider again the function $t^\ell(x)=\ell\cdot t(x)$ for all $x\in X_R$, $\ell\in \mathbb{R}_{>0}$ and let $d^\ell$ denote the corresponding biased metric. Finally, we define the solutions $d_{LP}^\ell$, $t_{LP}^\ell$ to \ref{LP} by \emph{i)} $d_{LP}^\ell(x^*,v)=\frac{1}{n_R}$ for all $v\in V_R$, \emph{ii)} $d_{LP}^\ell(x,v)= d^\ell(x,v)$ for all $x\in X_R\setminus \{x^*\}$, $v\in V_R$, and \emph{iii)} $t_{LP}^\ell=t^\ell$. It can be checked that $d_{LP}^\ell$, $t_{LP}^\ell$ is a feasible solution to \ref{LP}: to this end, we recall that every biased metric satisfies the first five constraints of our LP. Now, to infer $d_{LP}^\ell$ from $d^\ell$, we only increase the distance $d_{LP}^\ell(x^*,v)$ to $\frac{1}{n_R}$ for all $v\in N_R$. Since there is no upper bound on $d_{LP}^\ell(x^*,v)$, this does not violate any of the first five constraints and ensures that the last one is true. Finally, we note that there is an alternative $y$ such that $p(y)>0$ and $sc(y,d^*)> 0$ as $sc(p,d^*)>0$. Consequently, the objective value of the solutions $d^\ell_{LP}$, $t^\ell_{LP}$ is lower bounded by $\ell p(y) sc(y,d^*)$. Letting $\ell$ go to infinity thus shows that the objective value of \ref{LP} is not bounded in this case. Hence, it holds in all cases that $\mathit{dist}(p,R,x^*)\leq o_{LP}^*$, where $o_{LP}^*$ denotes the optimal objective value of \ref{LP} if it is bounded and $\infty$ otherwise.  
\end{proof}

\section{Proofs from Section 5}

As last part of this paper, we will formally prove the statements about the expected metric distortion of RSCFs given in \Cref{sec:IC}. To this end, we let $n_\succ(R)=|\{v\in V_R\colon {\succ_v}={\succ}\}|$ denote the number of voters that report the preference relation $\succ$ in the profile $R$. Moreover, we will subsequently show three auxiliary lemmas: first, we investigate the metric distortion of every lottery on profiles where all preference relations are reported by the same number of voters (cf. \Cref{lem:optmetric,lem:distuni}). Under the IC model, we can expect that the output profile is similar to such a profile if the number of voters is large enough. We hence prove in \Cref{lem:merge} that we can bound the metric distortion of such a profile $R$ based on the metric distortion of the chosen lottery for a large subprofile. Finally, we use these claims to prove \Cref{thm:ML} and \Cref{prop:expecteddistRD}.

In more detail, in our first lemma, we will identify a class of metrics $d\in D(R)$ that satisfy $\mathit{dist}(p,R,x^*)=\frac{sc(p,d)}{sc(x^*,d)}$ for all profiles $R$ in which all preference relations appear equally often, all lotteries $p$, and all alternatives $x^*\in X_R$. Surprisingly, we show that we can focus on a single type of metrics for this maximization problem: it always suffices to consider the biased metric $d\in D(R)$ given by the function $t$ with $t(x^*)=0$ and $t(x)=2$ for all $x\in X_R\setminus \{x^*\}$. We note that this gives further evidence for the conjecture by \citet{ChRa22a} that this type of metric is the worst-case for \emph{all} profiles. 

\begin{lemma}\label{lem:optmetric}
	Assume $m\geq 3$ and let $R\in\mathcal{R}^*_m$ denote a profile such that $n_{\succ}(R)=n_{\succ'}(R)>0$ for all preference relations ${\succ},{\succ'}\in \mathcal{R}(X_R)$. It holds for all lotteries $p\in \Delta(X_R)$ and alternatives $x^*\in X_R$ that $\mathit{dist}(p,R,x^*)=\frac{sc(p,d^*)}{sc(x^*,d^*)}$, where $d^*$ denotes the biased metric induced by the function $t$ with $t(x^*)=0$ and $t(x)=2$ for all $x\in X_R\setminus \{x^*\}$
\end{lemma}
\begin{proof}
	Let $R$ denote a profile as specified by the lemma and consider a lottery $p$ and an alternative $x^*$. If $p(x^*)=1$, then $\frac{sc(p,d)}{sc(x^*,d)}=1$ for every metric $d\in D(R)$, so we assume that $p(x^*)<1$. In this case, let $\hat d\in D(R)$ denote the biased metric given by the function $\hat t$ with $\hat t(x^*)=0$ and $\hat t(x)=\ell$ for all $x\in X_R\setminus \{x^*\}$, where $\ell$ is chosen such that $sc(x^*, \hat d)=1$. First, $\hat d$ is indeed a valid metric in $D(R)$ due to Proposition 5.1 of \citet{ChRa22a}. Next, we note that $\frac{sc(p,\hat d)}{sc(x^*,\hat d)}=\frac{sc(p,d^*)}{sc(x^*,d^*)}$ for the metric $d^*$ stated in the lemma as $\hat d=\alpha d^*$ for some $\alpha \in \mathbb{R}_{>0}$. Hence, we aim to show that $\mathit{dist}(p,R,x^*)=\frac{sc(p,\hat d)}{sc(x^*, \hat d)}$. For this, we will prove that $\mathit{dist}(p_x,R, x^*)=\frac{sc(p_x,\hat d)}{sc(x^*,\hat d)}$ for every alternative $x\in X_R$ and lottery $p_x$ with $p_x(x)=1$. This implies the lemma because 
	\begin{align*}
		\frac{sc(p,\hat d)}{sc(x^*,\hat d)}
		&\leq \mathit{dist}(p,R,x^*)\
		\leq \sum_{x\in X_R} p(x) \mathit{dist}(p_x,R, x^*)
		= \sum_{x\in X_R} p(x) \frac{sc(x,\hat d)}{sc(x^*, \hat d)}
		=\frac{sc(p,\hat d)}{sc(x^*, \hat d)}.
	\end{align*}
	
	Now, we first note that the claim trivially follows for the lottery $p_{x^*}$ as $\frac{sc(p_{x^*},d)}{sc(x^*,d)}=1$ for every metric $d\in D(R)$. We thus focus on an alternative $\hat x\in X_R\setminus \{x^*\}$. In more detail, we will show that $\hat d$, $\hat t$ correspond to an optimal solution of \ref{LP} for $\mathit{dist}(p_{\hat x},R, x^*)$ as \Cref{prop:LPcorrect} then implies that $\mathit{dist}(p_{\hat x}, R, x^*)=\frac{sc(p_{\hat x}, \hat d)}{sc(x^*,\hat d)}$. We therefore observe that it is easy to show that $\hat d$, $\hat t$ are a feasible solution for this linear program, so we will subsequently only prove that our solution is also optimal.\medskip
	
	\textbf{Step 1:} Since we want to reason about the optimal solutions of \ref{LP} (for $\mathit{dist}(p_{\hat x},R, x^*)$), we first prove that the optimal objective value of this linear program is bounded. To this end, let $d_{LP}$, $t_{LP}$ denote a feasible solution to \ref{LP}. We first note that $\sum_{v\in V_R} d_{LP}(x^*,v)=1$ and hence $d_{LP}(x^*,v)\leq 1$ for all $v\in V_R$. Moreover, since every preference relation appears at least once in $R$, there is a voter $v$ such that $\hat x\succ_v x^*$ and we can conclude by the first and third constraints that $1\geq d_{LP}(x^*,v)\geq \frac{1}{2} (t_{LP}(\hat x)-t_{LP}(x^*))=\frac{1}{2}t_{LP}(\hat x)$. Hence, it holds that $t_{LP}(\hat x)\leq 2$. By the fourth constraint, we can next conclude that $d_{LP}(\hat x,v)\leq d_{LP}(x^*,v)+t(\hat x)\leq 1+2=3$ for all $v\in V_R$. Finally, we can now compute that the objective value of any solution is at most $\sum_{x\in X_R}p_{\hat x}(x)\sum_{v\in V_R} d(x,v)=\sum_{v\in V_R} d(\hat x,v)\leq 3n_R$. Since this holds for every feasible solution of \ref{LP}, its optimal objective value is indeed bounded.\medskip
	
	\textbf{Step 2:} Let $d_{LP}^0$, $t_{LP}^0$ denote an optimal solution of \ref{LP} and let $o_{LP}^0$ denote its objective value. Our next goal is to construct an optimal solution $d_{LP}^1$, $t_{LP}^1$ of \ref{LP} such that $t_{LP}^1(x)=t_{LP}^1(y)$ for all $x,y\in X_R\setminus\{x^*, \hat x\}$. For this, we denote by $\Pi$ the set of permutations $\pi:X_R\rightarrow X_R$ such that $\pi(x^*)=x^*$ and $\pi(\hat x)=\hat x$. Moreover, given a permutation $\pi\in \Pi$, we let $R^\pi$ denote the profile defined by $x\succ^\pi_v y$ iff $\pi(x)\succ_v\pi(y)$ for all $x,y\in X_R$ and $v\in V_R$. Finally, we define $d^\pi(x,v)=d^0_{LP}(\pi(x),v)$ and $t^\pi(x)=t^0_{LP}(\pi(x))$ for all $x\in X_R$ and $v\in V_R$. Since $R^\pi$, $d^\pi$, and $t^\pi$ are all derived from $R$, $d^0_{LP}$, and $t^0_{LP}$ by renaming the alternatives according to $\pi$, it can be checked that $d^\pi$ and $t^\pi$ constitute a feasible solution of \ref{LP} for $\mathit{dist}(p_x, R^\pi, x^*)$ with objective value $o_{LP}^\pi=o_{LP}^0$. In particular, it is important here that $\pi(\hat x)=\hat x$ and $\pi(x^*)=x^*$ as these ensure that $t^{\pi}(x^*)=0$ and $d^{\pi}(\hat x,v)=d_{LP}^0(\hat x,v)$ for all $v\in V_R$. Next, since all preference relations appear equally often in the profile $R$, the profile $R^\pi$ equals $R$ up to renaming the voters. Hence, there is another bijection $\tau:V_{R}\rightarrow V_{R^\pi}$ such that ${\succ}_v={\succ_{\tau(v)}^\pi}$ for all voters $v\in V_R$. Based on this permutation, we define the functions $\bar d^\pi$ and $\bar t^\pi$ by $\bar d^\pi(x,v)=d^{\pi}(x,\tau(v))$ and $\bar t^\pi(x)=t^\pi(x)$ for all $x\in X_R$ and $v\in V_R$. Since we essentially only rename variables in this step, it follows that $\bar d^\pi$, $\bar t^\pi$ are a feasible solution to \ref{LP} for $\mathit{dist}(p_{\hat x}, R, x^*)$. Moreover, the objective value of this solution is $\bar o_{LP}^\pi=o_{LP}^\pi=o_{LP}^0$. 
	
	Next, we define the solution $d_{LP}^1$, $t_{LP}^1$ by $d_{LP}^1(x,v)=\frac{1}{(m-2)!}\sum_{\pi \in \Pi} \bar d^{\pi}(x,v)$ and $t_{LP}^1(x)=\frac{1}{(m-2)!}\sum_{\pi \in \Pi} \bar t^{\pi}(x)$ for all $x\in X_R$ and $v\in V_R$. Since $d_{LP}^1$, $t_{LP}^1$ is a convex combination of feasible solutions of \ref{LP}, it is itself again feasible. Furthermore, for every $\pi\in \Pi$, it holds that $\bar o_{LP}^\pi=o_{LP}^0$, so the objective value of our new solution is $o_{LP}^1=o_{LP}^0$. In particular, this means that $d_{LP}^1$, $t_{LP}^1$ is an optimal solution to \ref{LP} (for $\mathit{dist}(p_{\hat x},R, x^*)$). Finally, we note that $t^1_{LP}(x)=\frac{1}{(m-2)!} \sum_{\pi \in \Pi} \bar t^{\pi}(x)=\frac{1}{m-2} \sum_{z\in X_R\setminus \{\hat x, x^*\}} t^0_{LP}(z)=\frac{1}{(m-2)!} \sum_{\pi \in \Pi} \bar t^{\pi}(y)=t^1_{LP}(y)$ for all $x,y\in X_R\setminus \{x^*, \hat x\}$. Thus, our new solution satisfies all our requirements.\medskip
	
	\textbf{Step 3:} As third step, we will show that there is a  biased metric $\bar d$ defined by a function $\bar t$ with $\bar t(x)=\bar t(y)$ for all $x,y\in X_R\setminus \{\hat x, x^*\}$ that constitutes an optimal solution to \ref{LP}. For this, let $d^1_{LP}$, $t^1_{LP}$ denote the optimal solution constructed in the last step. First, we note that for all $x\in X_R\setminus \{x^*\}$, $v\in V_R$ with $d^1_{LP}(x,v)<d^1_{LP}(x^*,v)+\min_{y\in X_R\colon x\succeq_v y} t^1_{LP}(y)$, we can simply increase the value of $d^1_{LP}$ to $d^1_{LP}(x^*,v)+\min_{y\in X_R\colon x\succeq_v y} t(y)$ without violating any constraints. Moreover, increasing the value of $d^1_{LP}(x,v)$ does not reduce the objective value, so there is another optimal solution $d^2_{LP}$, $t^2_{LP}$ with $d^2_{LP}(x^*,v)=d^1_{LP}(x^*,v)$ for all $v\in V_R$, $t^2_{LP}(x)=t^1_{LP}(x)$ for all $x\in X_R$, and $d^2_{LP}(x,v)=d^2_{LP}(x^*,v)+\min_{y\in X_R\colon x\succeq_v y} t^2_{LP}(y)$ for all $x\in X_R\setminus \{x^*\}$, $v\in V_R$. 
	
	Next, we want to ensure that $d^2_{LP}(x^*,v)=\frac{1}{2}\max_{x,y\in X_R\colon x\succeq_v} t^2_{LP}(x)-t^2_{LP}(y)$. To this end, we assume that there is a voter $v^*$ such that $d^2_{LP}(x^*,v^*)>\frac{1}{2}\max_{x,y\in X_R\colon x\succeq_{v^*} y} t^2_{LP}(x)-t^2_{LP}(y)$. In this case, we define $\delta=d^2_{LP}(x^*,v^*)-\frac{1}{2}\max_{x,y\in X_R\colon x\succeq_{v^*} y} t^2_{LP}(x)-t^2_{LP}(y)$ and observe that $\delta<1$ as $d^2_{LP}(x^*,v)<1$ for all voters $v\in V_R$. In more detail, we first note here that the third constraint of \ref{LP} implies that $d^2_{LP}(x^*,v)\geq \frac{1}{2}(t^2_{LP}(x^*)-t^2_{LP}(x^*))=0$. So, if $d^2_{LP}(x^*,v)\geq 1$ for some voter $v\in V_R$, then $d^2_{LP}(x^*,v')=0$ for all $v'\in V_R\setminus \{v\}$. Since $m\geq 3$ and $X_R$ contains every preference relation equally often (and therefore at least once), there is for every alternative $x\in X_R\setminus \{x^*\}$ a voter $v'\in V_R\setminus \{v\}$ such that $x\succ_v x^*$. Since $d^2_{LP}(x^*,v')=t^2_{LP}(x^*)=0$, we can infer from the second and third conditions that $t^2_{LP}(x)=0$ for all $x\in X_R$. Moreover, the fourth condition then implies that $d^2_{LP}(x,v')\leq 0$ for all $x\in X_R$, $v'\in V_R\setminus \{v\}$ and that $d^2_{LP}(\hat x, v)\leq 1$, so the optimal objective value is at most $1$. However, the biased metric $\hat d$ corresponds to a feasible solution with a higher objective value, so $d^2_{LP}(x^*,v)<1$ for all $v\in V_R$. 
	
	Now, consider the solution $\tilde{d}$, $\tilde{t}$ derived from $d^2_{LP}$ and $t^2_{LP}$ by setting $\tilde{d}(x,v^*)=d^2_{LP}(x,v^*)-\delta$ for all $x\in X_R$. We first note that $\tilde d$, $\tilde t$ still satisfies the first four constraints of \ref{LP}. Moreover, it holds for all $x\in X_R$ that $\tilde{d}(x,v^*)=\tilde{d}(x^*,v^*)+\min_{y\in X_R\colon x\succeq_{v^*} y} \tilde{t}(y)$, so $\tilde{d}(x,v^*)+\tilde{d}(x^*,v^*)=2\tilde{d}(x^*,v^*)+\min_{y\in X_R\colon x\succeq_{v^*} y} \tilde{t}(y)\geq \tilde{t}(x)$ because $2\tilde{d}(x^*,v^*)\geq \tilde t(x)-\min_{y\in X_R\colon x\succeq_{v^*} y} \tilde{t}(y)$. Hence, our new solution only violates the normalization condition of \ref{LP}, and we can restore this by scaling all variables by the value $\frac{1}{1-\delta}$, i.e., $\tilde{d}'(x,v)=\frac{1}{1-\delta} \tilde{d}(x,v)$ and $\tilde{t}'(x)=\frac{1}{1-\delta} \tilde{t}(x)$ for all $x\in X_R$ and $v\in V_R$ while leaving the remaining conditions intact. Finally, we compute the objective value of our new solution $\tilde{d'}$, $\tilde{t}'$: 
	
	\begin{align*}
		\sum_{v\in V_R} \tilde{d'}(\hat x, v)
		&=\frac{1}{1-\delta}\sum_{v\in V_R} \tilde{d}(\hat x, v)\\
		&=\frac{1}{1-\delta}\sum_{v\in V_R} \left(d^2_{LP}(\hat x, v)\right) - \frac{\delta}{1-\delta}\\
		&=\frac{1}{1-\delta}\sum_{v\in V_R} \left(d^2_{LP}(x^*, v) + \!\!\!\min_{y\in X_R\colon \hat x\succeq_v y}\!\!t^2_{LP}(y)\right) - \frac{\delta}{1-\delta}\\
		&=\frac{1}{1-\delta} \sum_{v\in V_R} \left(\min_{y\in X_R\colon \hat x\succeq_v y} t^2_{LP}(y)\right) + \frac{1}{1-\delta}-\frac{\delta}{1-\delta}\\
		&\geq \sum_{v\in V_R} \left(d^2_{LP}(x^*,v)+\min_{y\in X_R\colon \hat x\succeq_v y} t^2_{LP}(y)\right)\\
		&=o_{LP}^2.
	\end{align*}
	
	Here, the first two inequalities use the definitions of $\tilde{d'}$ and $\tilde{d}$ respectively. Next, we apply that $d^2_{LP}(\hat x,v)=d^2_{LP}(x^*,v)+\min_{y\in X_R\colon \hat x\succeq_v y} t^2_{LP}(y)$ for all $v\in V_R$. In the third step, we then use that $\sum_{v\in V_R} d^2_{LP}(x^*,v)=1$. The remaining steps are simple arithmetic changes. This inequality proves that our new solution $\tilde{d'}$, $\tilde{t}'$ is an optimal solution to \ref{LP}. 
	
	Finally, we can repeat this step until we arrive at an optimal solution $d^3_{LP}$, $t^3_{LP}$ such that \emph{i)} $d^3_{LP}(x^*,v)=\frac{1}{2}\max_{x,y\in X_R\colon x\succeq_v y} t^3_{LP}(x)-t^3_{LP}(y)$ for all $v\in V_R$, \emph{ii)} $d^3_{LP}(x,v)=d^3_{LP}(x^*,v)+\min_{y\in X_R\colon x\succeq_v y} t^3_{LP}(y)$ for all $x\in X_R$, $v\in V_R$, \emph{iii)} $t^3_{LP}(x)=t^3_{LP}(y)$ for all $x,y\in X_R\setminus \{x^*,\hat x\}$, and \emph{iv)} $t^3_{LP}(x^*)=0$ and $t^3_{LP}(\hat x)\geq 0$. In particular, for the last points, we note that we only scale the values $t^1_{LP}$ by some constants during our constructions, so we directly inherit this insight from $d^1_{LP}$. Therefore, $d^3_{LP}$ is the biased metric $\bar d$ defined by $\bar t(x)=t^3_{LP}(x)$ for all $x\in X_R$.\medskip
	
	\textbf{Step 4:} As last step, we will show that $sc(\hat x, \hat d)\geq sc(\hat x, \bar d)$ for the metric $\bar d$ constructed during the last step. This completes the proof of this lemma since it means that $\hat d$, $\hat t$ are an optimal solution to \ref{LP}. To this end, we recall that the function $\hat t$ that defines $\hat d$ is specified by a single value $\ell\in \mathbb{R}_{>0}$: $\hat t(x^*)=0$ and $\hat t(x)=\ell$ for all $x\in X_R\setminus \{x^*\}$. Moreover, the function $\bar t$ that defines $\bar d$ is specified by two values $\ell_1$ and $\ell_2$: $\bar t(x^*)=0$, $\bar t(\hat x)=\ell_1$ and $\bar t(x)=\ell_2$ for all $x\in X_R\setminus \{\hat x, x^*\}$. If $\ell_1=\ell_2>0$, we are done and we thus suppose that $\ell_1\neq\ell_2$. 
	
	First suppose that $\ell_1\leq \ell$. In this case, we first note that $sc(x^*,\hat d)=sc(x^*,\bar d)=1$ by construction, so we will aim to show that $sc(\hat x,\hat d)-sc(x^*,\hat d)\geq sc(\hat x,\bar d)-sc(x^*,\bar d)$. To this end, we observe that 
	\begin{align*}
		sc(\hat x,\hat d)-sc(x^*,\hat d)=\sum_{v\in V_R} \min_{y\in X_R\colon x\succeq_v y} \hat t(y)=\frac{n_R}{2} \ell
	\end{align*}
	because half of the voters prefer $\hat x$ to $x^*$ (which means that $\min_{y\in X_R\colon \hat x\succeq_v y} \hat t(y)=0$) and the other half of the voters prefers $x^*$ to $\hat x$ (which means that $\min_{y\in X_R\colon \hat x\succeq_v y} \hat t(y)=\ell$). An analogous argument shows that $sc(\hat x,\bar d)-sc(x^*,\bar d)\leq \frac{n_R}{2} \ell_1$. Finally, since $\ell_1\leq \ell$, this implies that $sc(\hat x, \hat d)-sc(x^*, \hat d)\geq sc(\hat x, \bar d)-sc(x^*, \bar d)$, which shows that the lemma holds in this case.
	
	We thus suppose next that $\ell_1> \ell$, which implies that $\ell_2<\ell$. Indeed, if $\ell_1>\ell$ and $\ell_2\geq \ell$, then $\hat d(x^*,v)\leq \bar d(x^*,v)$ for all $v\in V_R$, and the inequality is strict for all voters that rank $x^*$ below $\hat x$. In more detail, it holds that $\hat d(x^*,v)=0\leq \bar d(x^*,v)$ for all voters $v$ that top-rank $x^*$ and $\hat d(x^*,v)=\frac{\ell}{2}\leq\frac{\min(\ell_1,\ell_2)}{2}\leq \bar d(x^*,v)$ for all other voters. Hence, $sc(x^*,\bar d)>sc(x^*,\hat d)=1$, which contradicts that $sc(x^*,\bar d)=1$. So, we derive indeed that $\ell_2<\ell$. 
	
	We thus suppose that $\ell_2<\ell<\ell_1$ and assume for contradiction that $sc(\hat x, \hat d)< sc(\hat x, \bar d)$. Since $sc(x^*, \hat d)=sc(x^*, \bar d)=1$, this assumption implies that $sc(\hat x, \hat d)-sc(x^*,\hat d)<sc(\hat x, \bar d)-sc(x^*,\bar d)$. We will thus compute the values of these differences and therefore recall that $sc(\hat x,\hat d)-sc(x^*,\hat d)=\frac{n_R}{2} \ell$. Moreover, $\bar d(\hat x, v)=\bar d(x^*,v)$ for all voters $v\in V_R$ with $\hat x\succ_v x^*$, $\bar d(\hat x, v)=\bar d(x^*,v)+\ell_1$ for all voters $v\in V_R$ that bottom-rank $\hat x$, and $\bar d(\hat x, v)=\bar d(x^*,v)+\ell_2$ for all remaining voters as these prefer $\hat x$ to some alternative $x\neq x^*$. Since $\frac{n_R}{2}$ voters prefer $\hat x$ to $x^*$, $\frac{n_R}{m}$ voters bottom-rank $\hat x$ in $R$, there are $n_R(\frac{1}{2}-\frac{1}{m_R})$ voters in the last case. Consequently, 
	\begin{align*}
		sc(\hat x, \bar d)-sc(x^*, \bar d)=\frac{n_R}{m}\ell_1+(\frac{n_R}{2}-\frac{n_R}{m})\ell_2.
	\end{align*}
	
	Because $sc(\hat x, \hat d)-sc(x^*,\hat d)<sc(\hat x, \bar d)-sc(x^*,\bar d)$, we conclude that 
	\begin{align}
		&\frac{n_R}{2}\ell<\frac{n_R}{m}\ell_1+(\frac{n_R}{2}-\frac{n_R}{m})\ell_2\nonumber\\
		\iff &(\frac{n_R}{2}-\frac{n_R}{m})(\ell-\ell_2)<\frac{n_R}{m}(\ell_1-\ell)\nonumber\\
		\iff &\frac{m-2}{2}(\ell-\ell_2)<\ell_1-\ell.\label{eq3}
	\end{align}
	
	To derive a contradiction, we next use that $sc(x^*,\hat d)=sc(x^*,\bar d)$. We hence observe that 
	\begin{align*}
		1=sc(x^*,\hat d)=n_R\cdot \frac{m-1}{m}\cdot \frac{\ell}{2}
	\end{align*}
	as the $\frac{n_R}{m}$ voters who top-rank $x^*$ satisfy $\hat d(x^*,v)=\max_{x,y\in X_R\colon x\succeq_v y} \hat t(x)-\hat t(y)=0$ and all other voters have $d(x^*,v)=\frac{\ell}{2}$. 
	
	Furthermore, to compute $sc(x^*,\bar d)$, we will determine (lower bounds on) $\bar d(x^*,v)$ for every voter $v\in V_R$. To verify the subsequent values, it suffices to identify the pair of alternatives $x,y\in X_R$ with $x\succeq_v y$ that maximizes $\frac{1}{2} (\bar t(x)-\bar t(y))$ due to the definition of $\bar d(x^*,v)$. 
	\begin{itemize}
		\item $\bar d(x^*,v)\geq 0$ for all voters $v$ top-rank $x^*$. There are $\frac{n_R}{m}$ such voters.
		\item $\bar d(x^*,v)=\frac{\ell_1}{2}$ for all voters $v$ that bottom-rank $x^*$. There are $\frac{n_R}{m}$ such voters.
		\item $\bar d(x^*,v)=\frac{\ell_1}{2}$ for all voters $v$ that neither top-rank nor bottom-rank $x^*$ and that prefer $\hat x$ to $x^*$. We note that there are $n_R\frac{m-2}{m}$ voters that neither top-rank nor bottom-rank $x^*$ and exactly half of them prefer $\hat x$ to $x^*$. Hence, there are $\frac{n_R(m-2)}{2m}$ such voters.
		\item $\bar d(x^*,v)\geq \frac{\ell_2}{2}$ for all voters $v$ that neither top-rank nor bottom-rank $x^*$ and that prefer $x^*$ to $\hat x$. The central observation for this is that these voters prefer an alternative $x$ with $\bar t(x)=\ell_2$ to $x^*$. Analogous to the last case, there are $\frac{n_R(m-2)}{2m}$ such voters.
	\end{itemize}
	
	Finally, we can now lower bound $sc(x^*, \bar d)$:
	\begin{align*}
		sc(x^*, \bar d)&=\sum_{v\in V_R} \bar d(x^*, v)
		\geq \frac{n_R}{2}\left(\frac{1}{m} 0 + \frac{1}{m} \ell_1 + \frac{m-2}{2m}\ell_1 + \frac{m-2}{2m} \ell_2\right)
		=\frac{n_R}{2m}\left(\frac{m}{2}\ell_1 + \frac{m-2}{2}\ell_2\right)
	\end{align*}
	
	On the other side, we have $sc(x^*,\bar d)=sc(x^*, \hat d)=1$. Since $sc(x^*,\hat d)=\frac{n_R(m-1)}{2m}\ell$, we derive that 
	\begin{align}
		\frac{n_R(m-1)}{2m}\ell&\geq \frac{n_R}{2m}\left(\frac{m}{2}\ell_1 + \frac{m-2}{2}\ell_2\right)\nonumber\\
		\iff \frac{m-2}{2}(\ell-\ell_2)&\geq \frac{m}{2}(\ell_1-\ell)\nonumber\\
		\iff\frac{m-2}{m}(\ell-\ell_2)&\geq \ell_1 -\ell.\label{eq4}
	\end{align}
	
	Finally, we get from Equations \ref{eq3} and \ref{eq4} that $\frac{m-2}{2}(\ell-\ell_2)<\ell_1-\ell\leq \frac{m-2}{m}(\ell-\ell_2)$. This is a contradiction as $m\geq 3$, so the initial assumption that $\frac{sc(\hat x, \hat d)}{sc(x^*, \hat d)}< \frac{sc(\hat x, \bar d)}{sc(x^*,\bar d)}$ must have been wrong. We have now exhausted all cases and thus conclude that $sc(\hat x, \hat d)\geq sc(\hat x, \bar d)$, which finally proves the lemma.
\end{proof}

Due to \Cref{lem:optmetric}, we can now compute the metric distortion of every lottery on a profile $R$ with $n_{\succ}(R)=n_{\succ'}(R)$ for all ${\succ},{\succ'}\in \mathcal{R}(X_R)$.

\begin{lemma}\label{lem:distuni}
	Assume $m\geq 3$ and let $R\in\mathcal{R}^*_m$ denote a profile such that $n_{\succ}(R)=n_{\succ'}(R)>0$ for all preference relations ${\succ},{\succ'}\in \mathcal{R}(X_R)$. It holds for every lottery $p\in \Delta(X_R)$ that $\mathit{dist}(p,R)=2+\frac{1}{m-1}-\frac{m}{m-1}\min_{x\in X_R} p(x)$. 
\end{lemma}
\begin{proof}
	Let $R$ denote a profile such that $n_{\succ}(R)=n_{\succ'}(R)>0$ for all preference relations ${\succ},{\succ'}\in \mathcal{R}(X_R)$ and consider an arbitrary lottery $p$. We will next compute $\mathit{dist}(p,R,x^*)$ for every alternative $x^*\in X_R$. To this end, we use that, by \Cref{lem:optmetric}, $\mathit{dist}(p,R,x^*)=\frac{sc(p,d)}{sc(x^*,d)}$ for the biased metric $d$ defined by the function $t$ with $t(x^*)=0$ and $t(x)=2$ for all $x\in X_R\setminus \{x^*\}$. Next, we observe that $\frac{sc(p,d)}{sc(x^*,d)}=\sum_{x\in X_R} p(x) \frac{sc(x,d)}{sc(x^*,d)}$. We will thus compute the social cost of every alternative. 
	
	For $x^*$, we first note that $d(x^*,v)=0$ for all voters that top-rank $x^*$ and $d(x^*,v)=1$ for all other voters. Hence, it is easy to infer that $sc(x^*,d)=\frac{n_R(m-1)}{m}$. By contrast, to compute the social cost of an alternative $x\in X_R\setminus \{x\}$, we need a more elaborate analysis of the distances $d(x,v)$: 
	\begin{itemize}
		\item $d(x,v)=2$ for all voters who top-rank $x^*$. There are $\frac{n_R}{m}$ such voters. 
		\item $d(x,v)=1$ for all voters who bottom-rank $x^*$. There are $\frac{n_R}{m}$ such voters. 
		\item $d(x,v)=1$ for all voters who do neither top-rank nor bottom-rank $x^*$ and prefer $x$ to $x^*$. There are $\frac{n_R (m-2)}{m}$ voters who do neither top-rank nor bottom-rank $x^*$ and precisely half of them prefer $x$ to $x^*$. Thus, there are $\frac{n_R (m-2)}{2m}$ such voters.
		\item $d(x,v)=3$ for all voters who do neither top-rank nor bottom-rank $x^*$ and prefer $x^*$ to $x$. There are again $\frac{n_R (m-2)}{2m}$ such voters.
	\end{itemize}
	
	We can hence compute that 
	\begin{align*}
		sc(x,d)& =\sum_{v\in V_R} d(x,v)
		=n_R\left(\frac{2}{m}+\frac{1}{m}+\frac{m-2}{2m}+\frac{3(m-2)}{2m}\right)
		=\frac{n_R}{m}(2m-1).
	\end{align*}
	
	It hence follows that $\frac{sc(x^*,d)}{sc(x^*,d)}=1$ and $\frac{sc(x,d)}{sc(x^*,d)}=\frac{2m-1}{m-1}=2+\frac{1}{m-1}$. Moreover, we can now compute that $\frac{sc(p,d)}{sc(x^*,d)}=(1-p(x^*))(2+\frac{1}{m-1})+p(x^*)=2+\frac{1}{m-1}-\frac{m}{m-1}p(x^*)$. Clearly, this function is decreasing in $p(x^*)$, so we derive that $\mathit{dist}(p,R)=2+\frac{1}{m-1}-\frac{m}{m-1}\min_{x\in X_R} p(x)$.
\end{proof}

We note that, by \Cref{lem:distuni}, the optimal lottery $p$ for a profile $R$ with $n_{\succ}(R)=n_{\succ'}(R)>0$ for all ${\succ},{\succ'}\in \mathcal{R}(X_R)$, assigns probability $p(x)=\frac{1}{m}$ to all $x\in X_R$. In particular, this lottery achieves a metric distortion of $2$ for $R$. By contrast, every lottery that assigns $0$ to some alternative has a metric distortion of $2+\frac{1}{m-1}$ in $R$. 

To be able to use \Cref{lem:distuni} in the analysis of the expected metric distortion of RSCFs, we observe that each preference relation will appear roughly equally often with high probability in a preference profile drawn from the IC distribution if the number of voters $n$ is sufficiently large. However, we cannot expect to get precisely a profile where every preference relation appears equally often, and we thus give next a lemma that allows to bound the metric distortion of a lottery $p$ in a profile $R$ based on a large subprofile of~$R$.

\begin{lemma}\label{lem:merge}
	Let $R$ be a profile and let $p\in \Delta(X_R)$ denote a lottery. Moreover, let $R'$ denote a profile derived from $R$ by choosing a subset of the voters $V_{R'}\subsetneq V_R$ and setting ${\succ_v'}={\succ_v}$ for all $v\in V_{R'}$, and define $\alpha=1-\frac{|V_{R'}|}{|V_R|}$. If $\mathit{dist}(p,R)<\infty$ and $\mathit{dist}(p,R')<\infty$, then $\mathit{dist}(p,R)\leq \mathit{dist}(p,R')+\alpha(\mathit{dist}(p,R)+1)$.
\end{lemma}
\begin{proof}
	Let $R$ and $R'$ denote two profiles as defined by the lemma and let $\alpha=1-\frac{|V_{R'}|}{|V_R|}$. Moreover, consider an arbitrary lottery $p$, let $d$ denote the metric $d\in D(R)$ that maximizes $\frac{sc(p,d)}{\min_{x\in X_R} sc(x,d)}$, and let $x^*$ denote an alternative with $sc(x^*,d)=\min_{x\in X_R} sc(x,d)$. Finally, we define the set $\bar V_{R}=V_R\setminus V_{R'}$ and note that $\alpha n_R=|\bar V_R|$. Our main goal is to bound $\sum_{v\in \bar V_R} d(x,v)$ for every alternative $x\in X_R$. To this end, we first note that $d(x,v)\leq d(x,x^*)+d(x^*,v)$ for every voter $v\in V_R$. Moreover, $d(x,x^*)\leq d(x,v)+d(v,x^*)$ for every voter $v\in V_R$, so $d(x, x^*)\leq \frac{1}{n_R}(sc(x,d)+sc(x^*,d))$. Combining these insights means that 
	\begin{align*}
		\sum_{v\in \bar V_R} d(v,x)&\leq |\bar V_R| d(x,x^*)+\sum_{v\in \bar V_R} d(x^*,v)
		\leq \alpha (sc(x,d)+sc(x^*,d))+ \sum_{v\in \bar V_R} d(x^*,v).
	\end{align*}
	
	Hence, we now compute that 
	\begin{align*}
		\mathit{dist}(p,R)&=\sum_{x\in X_R}p(x) \frac{ sc(x,d)}{sc(x^*,d)}\\
		&= \sum_{x\in X_R } p(x) \frac{\sum_{v\in V_{R'}} d(x,v)+\sum_{v\in \bar V_{R}} d(x,v)}{sc(x^*,d)}\\
		&\leq \sum_{x\in X_R } p(x) \frac{\sum_{v\in V_{R'}} d(x,v)+\sum_{v\in \bar V_{R}} d(x^*,v)}{sc(x^*,d)} + \sum_{x\in X_R } p(x)\frac{\alpha(sc(x,d)+sc(x^*,d))}{sc(x^*,d)}\\
		&=\sum_{x\in X_R } p(x) \frac{\sum_{v\in V_{R'}} d(x,v)}{\sum_{v\in V_{R'}} d(x^*,v)} \cdot \frac{\sum_{v\in V_{R'}} d(x^*,v)}{sc(x^*,d)} + \frac{sc(x^*,d)-\sum_{v\in V_{R'}} d(x^*,v)}{sc(x^*,d)} + \alpha \left(\frac{sc(p,d)}{sc(x^*,d)} +1\right)\\
		&\leq \mathit{dist}(p,R') \frac{\sum_{v\in V_{R'}} d(x^*,v)}{sc(x^*,d)} + \alpha(\mathit{dist}(p,R)+1)  + \mathit{dist}(p,R')\frac{sc(x^*,d)-\sum_{v\in V_{R'}} d(x^*,v)}{sc(x^*,d)} \\
		&= \mathit{dist}(p,R')+ \alpha (\mathit{dist}(p,R)+1).
	\end{align*}
	
	The first two equalities merely employ definitions. The next step uses our previously deduced upper bound for $\sum_{v\in\bar V_R} d(x,v)$. In the fifth step, we substitute $\frac{\sum_{v\in V_{R'}} d(x,v)+\sum_{v\in \bar V_{R}} d(x^*,v)}{sc(x^*,d)}$ by $\frac{\sum_{v\in V_{R'}} d(x,v)}{\sum_{v\in V_{R'}} d(x^*,v)}\cdot \frac{{\sum_{v\in V_{R'}} d(x^*,v)}}{sc(x^*,d)}+\frac{sc(x^*,d)-\sum_{v\in V_{R'}} d(x^*,v)}{sc(x^*,d)}$ and simplify the last summand. Next, we use that $1\leq \mathit{dist}(p,R')$ and $\frac{\sum_{v\in V_{R'}} d(x,v)}{\sum_{v\in V_{R'}} d(x^*,v)}\leq \mathit{dist}(p,R')$, and the last step is just a simple transformation. This completes the proof of this lemma. 
\end{proof}

Based on analysis so far, we now prove \Cref{thm:ML}. 

\ML*
\begin{proof}
	Fix some number of voters and alternatives $m$ and $n$ such that $n$ is significantly larger than $m!$ (i.e., such that all subsequent terms are well-defined) and consider an arbitrary RSCF $f$ with $\mathit{dist}_m(f)<\infty$. We will give lower and upper bounds on $\mathbb{E}_{R\sim IC(m,n)}[dist(f(R),R)]$ that converge to $2+\frac{z}{m-1}$ and $2+\frac{1}{m-1}$ respectively. To facilitate the proof, we let $R$ denote a random variable which is distributed according to $IC(m,n)$, and set $\alpha=\frac{1}{\sqrt[3]{n}}$ and $y=\mathit{dist}_m(f)$. Moreover, we define $T^\alpha$ as the set of profiles $R'$ such that $n_{\succ}(R')>(1-\alpha)\frac{n}{m!}$ for all ${\succ}\in \mathcal{R}(X_m)$, and $S$ as the set of profiles $R'$ such that $f(R',x)=0$ for some $x\in X_R$.
	By the law of total probability, we first observe that
	\begin{align*}
		\mathbb{E}[\mathit{dist}(f(R),R)|]
		&=\mathbb{P}[R\not\in T^\alpha]\cdot \mathbb{E}[\mathit{dist}(f(R),R)|R\not\in T^\alpha]\\
		& \quad+ \mathbb{P}[R\in T^\alpha]\cdot\mathbb{E}[\mathit{dist}(f(R),R)|R\in T^\alpha].
	\end{align*}
	
	Based on this insight, we will now compute our upper and lower bounds. 
	\medskip
	
	\textbf{Upper bound}: For our upper bound, we fist note that $\mathbb{E}[\mathit{dist}(f(R),R)|R\not\in T^\alpha]\leq y$ as $\mathit{dist}(f(R),R)\leq y$ for all profiles $R$ on $m$ alternatives. Moreover, it holds for every preference relation ${\succ_1}$ that 
	\begin{align*}
		\mathbb{P}[R\not\in T^\alpha]&=\mathbb{P}[\exists\, {\succ}\in \mathcal{R}(X_R)\colon n_{\succ}(R)\leq (1-\alpha)\frac{n}{m!}]\\
		&\leq m!\cdot \mathbb{P}[n_{\succ_1}(R)\leq (1-\alpha)\frac{n}{m!}]\\
		&\leq m!e^{-\frac{\alpha^2}{2}\cdot\frac{n}{m!}}\\
		&=m!e^{-\frac{\sqrt[3]{n}}{2m!}}.
	\end{align*}
	
	Here, the first inequality is the union bound and the second one a standard Chernoff bound, where we use that $\frac{n}{m!}$ is the expectation of $n_{\succ_1}(R)$. This means that 
	\begin{align*}
		&\mathbb{E}[\mathit{dist}(f(R),R)|]\leq ym!e^{-\frac{\sqrt[3]{n}}{2m!}} + (1-m!e^{-\frac{\sqrt[3]{n}}{2m!}}) \mathbb{E}[\mathit{dist}(f(R),R)|R\in T^\alpha]. 
	\end{align*}
	
	We next aim to bound $\mathbb{E}[\mathit{dist}(f(R),R)|R\in T^\alpha]$. Towards this, we note that every profile $R\in T^\alpha$ has a subprofile $R'$ such that $n_{\succ}(R')=\lceil (1-\alpha)\frac{n}{m!}\rceil$ for every ${\succ}\in\mathcal{R}(X_m)$. Now, by \Cref{lem:distuni}, it follows that $\mathit{dist}(f(R), R')\leq 2+\frac{1}{m-1}$. Applying \Cref{lem:merge} then shows that 
	\begin{align*}
		\mathit{dist}(f(R),R))&\leq 2+\frac{1}{m-1} + (1-\frac{|V_{R'}|}{n}) (y+1)
		\leq  2+\frac{1}{m-1}+\alpha (y+1).
	\end{align*}
	
	Hence, we can now conclude that 
	\begin{align*}
		&\mathbb{E}[\mathit{dist}(f(R),R)|]\leq ym!e^{-\frac{\sqrt[3]{n}}{2m!}} + (1-m!e^{-\frac{\sqrt[3]{n}}{2m!}})(2+\frac{1}{m-1}+\alpha(y+1)).
	\end{align*}

	Since $\alpha=\frac{1}{\sqrt[3]{n}}$, taking the limit shows that \[\limsup_{n\rightarrow\infty} \mathbb{E}_{R\sim IC(m,n)}[\mathit{dist}(f(R),R)]\leq 2+\frac{1}{m-1}.\]
	
	\textbf{Lower bound:} For the lower bound, we note that $\mathbb{E}[\mathit{dist}(f(R),R)|R\not\in T^\alpha]\geq 1$ because $\mathit{dist}(f(R),R)\geq 1$ for every profile $R$. Hence, we derive that
	\begin{align*}
		\mathbb{E}[\mathit{dist}(f(R),R)]
		&= \mathbb{P}[R\not\in T^\alpha] \cdot\mathbb{E}[\mathit{dist}(f(R),R)|R\not\in T^\alpha]\\
		&\quad + \mathbb{P}[R\in T^\alpha\setminus S]\cdot \mathbb{E}[\mathit{dist}(f(R),R)|R\in T^\alpha\setminus S]\\
		&\quad + \mathbb{P}[R\in T^\alpha\cap S] \cdot\mathbb{E}[\mathit{dist}(f(R),R)|R\in T^\alpha\cap S]\\
		&\geq \mathbb{P}[R\in T^\alpha\setminus S]\cdot \mathbb{E}[\mathit{dist}(f(R),R)|R\in T^\alpha\setminus S]\\
		&\quad + \mathbb{P}[R\in T^\alpha\cap S] \cdot\mathbb{E}[\mathit{dist}(f(R),R)|R\in T^\alpha\cap S].
	\end{align*}

	Next, it is simple to see that $\mathbb{P}[R\in T^\alpha\cap S]\geq 1- \mathbb{P}[R\not\in T^\alpha]-\mathbb{P}[R\not\in S]=\mathbb{P}[R\in S]-\mathbb{P}[R\not\in T^\alpha]$. Moreover, as shown for the upper bound, it holds that $\mathbb{P}[R\not\in T^\alpha]\leq m!e^{-\frac{\sqrt[3]{n}}{2m!}}$, so we have that 
	\begin{align*}
		&\mathbb{P}[R\in T^\alpha\cap S]\geq \mathbb{P}[R\in S]-m!e^{-\frac{\sqrt[3]{n}}{2m!}}.
	\end{align*}
	
	Subsequently, we will derive lower bounds on our expectations and first analyze $\mathbb{E}[\mathit{dist}(f(R),R)|R\in T^\alpha\cap S]$. To this end, we first fix a profile $R\in T^\alpha\cap S$ and investigate $\mathit{dist}(f(R),R)$. Moreover, let $x^*$ denote an alternative with $f(R,x^*)=0$ (which exists as $R\in S$) and consider the biased metric $d\in D(R)$ given by the function $t$ with $t(x^*)=0$ and $t(x)=2$ for all $x\in X_R\setminus \{x^*\}$. We observe again that $R$ has a subprofile $R'$ such that $n_\succ(R')=\lceil (1-\alpha)\frac{n}{m!}\rceil$ (because $R\in T^\alpha$). 
	Similar to the proof of \Cref{lem:optmetric}, it is easy to show for all $x\in X_R\setminus \{x^*\}$ that $\sum_{v\in V_{R'}} d(x,v)=\frac{n_{R'}(2m-1)}{m}$ and that $\sum_{v\in V_{R'}} d(x^*,v)=\frac{n_{R'}(m-1)}{m}$. Since $f(R,x^*)=0$, we can compute that 
	\begin{align*}
		sc(f(R),d)&=\sum_{x\in X_R}f(R,x)\sum_{v\in V_R} d(v,x)
		\geq \sum_{x\in X_R}f(R,x)\sum_{v\in V_{R'}} d(v,x)
		= \frac{n_{R'}(2m-1)}{m}
		\geq (1-\alpha)n \frac{2m-1}{m}.
	\end{align*}
	
	By contrast, we can infer that $sc(x^*,R)\leq \frac{n_{R'}(m-1)}{m}+(n-n_{R'})\leq (1-\alpha)n\frac{m-1}{m}+\alpha n$. In particular, we note for this inequality that $n_{R'}=m!\lceil(1-\alpha)\frac{n}{m}!\rceil\geq (1-\alpha)n$ and that $d(v,x^*)\leq 1$ for all $v\in V_R$. We can now derive that 
	\begin{align*}
		\frac{sc(f(R),d)}{sc(x^*,d)}&\geq \frac{(1-\alpha)n \frac{2m-1}{m}}{(1-\alpha)n\frac{m-1}{m}+\alpha n}
		=\frac{\frac{2m-1}{m-1}}{1+\frac{\alpha m}{(1-\alpha)(m-1)}}
		=(2+\frac{1}{m-1})\cdot \frac{1}{1+\frac{m}{m-1}\cdot\frac{1}{\sqrt[3]{n}-1}}.
	\end{align*}
	
	Finally, we conclude that $\mathit{dist}(f(R),R)\geq \frac{sc(f(R),R)}{sc(x^*,d)}\geq (2+\frac{1}{m-1})\cdot \frac{1}{1+\frac{m}{m-1}\cdot\frac{1}{\sqrt[3]{n}-1}}$ for all $R\in T^\alpha\cap S$. As a consequence, $\mathbb{E}[\mathit{dist}(f(R),R)|R\in T^\alpha\cap S]\geq (2+\frac{1}{m-1})\cdot \frac{1}{1+\frac{m}{m-1}\cdot\frac{1}{\sqrt[3]{n}-1}}$, too. 
	
	Next, we will bound $\mathbb{E}[\mathit{dist}(f(R),R)|R\in T^\alpha\setminus S]$. To this end, let $R\in T^\alpha\setminus S$, let $x^*\in X^R$ denote an alternative that minimizes $f(R,x^*)$, and let $d$ denote the same biased metric as before. Since $R\in T^\alpha$, there is a subprofile $R'$ that contains every ballot precisely $n_{\succ}(R')=\lceil(1-\alpha) \frac{n}{m!}\rceil$ times. Since $f(R,x^*)\leq\frac{1}{m}$, we can compute that 
	\begin{align*}
		sc(f(R), d)&\geq \sum_{v\in V_{R'}} \sum_{x\in X_R} f(R,x) d(x,v)\\
		&=(1-f(R,x^*)) \frac{n_{R'}(2m-1)}{m}+f(R,x^*)\frac{n_{R'}(m-1)}{m}\\
		&\geq 2n_{R'}\frac{m-1}{m}\\
		& \geq 2(1-\alpha) n\frac{m-1}{m}.
	\end{align*} 
	
	Moreover, by our previous analysis, $sc(x^*,d)\leq (1-\alpha) n\frac{m-1}{m}+\alpha n$. Hence, we derive that 
	\begin{align*}
		\mathit{dist}(f(R),R)&\geq \frac{sc(f(R),d)}{sc(x^*,d)}\\
		&\geq \frac{2(1-\alpha) n\frac{m-1}{m}}{(1-\alpha)\frac{m-1}{m} n+\alpha n}\\
		&\geq 2\frac{1}{1+\frac{m}{m-1}\cdot\frac{\alpha}{1-\alpha}}\\
		&\geq 2\frac{1}{1+\frac{m}{m-1}\cdot\frac{1}{\sqrt[3]{n}-1}}.
	\end{align*}
	
	Since this holds for every $R\in T^\alpha$, we infer that $\mathbb{E}[\mathit{dist}(f(R),R)|R\in T^\alpha\setminus S]\geq 2\frac{1}{1+\frac{1}{\sqrt[3]{n}-1}}$. Finally, we can now put everything together: 
	\begin{align*}
		\mathbb{E}[\mathit{dist}(f(R),R)] &\geq \mathbb{P}[R\in T^\alpha\setminus S]\cdot \mathbb{E}[\mathit{dist}(f(R),R)|R\in T^\alpha\setminus S]
		+ \mathbb{P}[R\in T^\alpha\cap S] \cdot\mathbb{E}[\mathit{dist}(f(R),R)|R\in T^\alpha\cap S]\\
		&\geq \mathbb{P}[R\in T^\alpha\setminus S]\cdot 2\cdot \frac{1}{1+\frac{m}{m-1}\cdot\frac{1}{\sqrt[3]{n}-1}}
		 + \mathbb{P}[R\in T^\alpha\cap S]\cdot(2+\frac{1}{m-1})\cdot \frac{1}{1+\frac{m}{m-1}\cdot\frac{1}{\sqrt[3]{n}-1}}\\
		&= \mathbb{P}[R\in T^\alpha]\cdot 2\cdot \frac{1}{1+\frac{m}{m-1}\cdot\frac{1}{\sqrt[3]{n}-1}}
		+\mathbb{P}[R\in T^\alpha\cap S]\cdot \frac{1}{m-1}\cdot \frac{1}{1+\frac{m}{m-1}\cdot\frac{1}{\sqrt[3]{n}-1}}\\
		&\geq \frac{1}{1+\frac{m}{m-1}\cdot\frac{1}{\sqrt[3]{n}-1}} \cdot 2\cdot (1-m!e^{-\frac{\sqrt[3]{n}}{2m!}})
		+ \frac{1}{1+\frac{m}{m-1}\cdot\frac{1}{\sqrt[3]{n}-1}}\cdot (\mathbb{P}[R\in S]- m!e^{-\frac{\sqrt[3]{n}}{2m!}})\cdot \frac{1}{m-1}.
	\end{align*}
	
	Now, it is easy to verify that 
	\begin{align*}
		&\lim_{n\rightarrow\infty} \frac{1}{1+\frac{m}{m-1}\cdot\frac{1}{\sqrt[3]{n}-1}}=1 \text{    and}\\
		&\lim_{n\rightarrow\infty} m!e^{-\frac{\sqrt[3]{n}}{2m!}}\cdot (2+\frac{1}{m-1})\cdot \frac{1}{1+\frac{m}{m-1}\cdot\frac{1}{\sqrt[3]{n}}}=0.
	\end{align*}
	
	Since $z=\liminf_{n\rightarrow \infty} \mathbb{P}[R\in S]$, it hence follows that $
	\liminf_{n\rightarrow \infty} \mathbb{E}_{R\sim IC(m,n)}[\mathit{dist}(f(R),R)]\geq 2 + \frac{z}{m-1}.
	$
\end{proof}

Finally, we present the proof of \Cref{prop:expecteddistRD}.

\expectedRD*
\begin{proof}
	We will prove each of our three claims separately.\medskip
	
	\textbf{Claim (1): Uniform Random Dictatorship.}\\
	We first will show that $\lim_{n\rightarrow\infty} \mathbb{E}_{R\sim IC(m,n,)}[\mathit{dist}(f_\mathit{RD}(R), R)]=2$. To this end, we note that $z=\liminf_{n\rightarrow n} \mathbb{P}[\exists x\in X_R\colon f(R,x)=0]\geq 0$ because we consider a probability here. Hence, we can simplify the lower bound in \Cref{thm:ML} to $\liminf_{n\rightarrow \infty} \mathbb{E}_{R\sim IC(m,n)}[\mathit{dist}(f_{RD}(R),R)]\geq 2$. We will next show a matching upper bound on this limit, i.e., that $\limsup_{n\rightarrow \infty} \mathbb{E}_{R\sim IC(m,n)}[\mathit{dist}(f_{RD}(R),R)]\leq 2$, which then implies our claim. In more detail, we will infer an upper bound on $\mathbb{E}_{R\sim IC(m,n)}[\mathit{dist}(f_{RD}(R),R)]$ for every fixed $n$ that will converge to $2$. For this, we denote by $R$ a random variable that is distributed according to $IC(m,n)$ and set $\alpha=\frac{1}{\sqrt[3]{n}}$. We furthermore define by $T^\alpha$ the set of profiles on $n$ voters and $m$ alternatives such that $n_{\succ}(R)>(1-\alpha) \frac{n}{m!}$ for all ${\succ}\in \mathcal{R}(X_R)$ and note that, by the law of total probability, it holds that
	\begin{align*}
		\mathbb{E}[\mathit{dist}(f_{RD}(R),R)|]&=\mathbb{P}[R\not\in T^\alpha]\cdot \mathbb{E}[\mathit{dist}(f_{RD}(R),R)|R\not\in T^\alpha]\\
		& + \mathbb{P}[R\in T^\alpha]\cdot \mathbb{E}[\mathit{dist}(f_{RD}(R),R)|R\in T^\alpha].
	\end{align*}
	
	For our upper bound, we note that $\mathbb{E}[\mathit{dist}(f_{RD}(R),R)|R\not\in T^\alpha]\leq 3$ as $\mathit{dist}(f_{RD}(R),R)\leq 3$ for all profiles $R$. Moreover, just in the proof of \Cref{thm:ML}, we can infer that 
	\begin{align*}
		\mathbb{P}[R\not\in T^\alpha]
		&\leq m!e^{-\frac{\sqrt[3]{n}}{2m!}}.
	\end{align*}
	
	We next aim to bound $\mathbb{E}[\mathit{dist}(f_{RD}(R),R)|R\in T^\alpha]$. For this, we observe that $f_\mathit{RD}(R,x)> (1-\alpha) \frac{1}{m}$ for all $x\in X_R$ and $R\in T^\alpha$. Next, let $R'$ denote the subprofile of $R$ such that each preference relation appears $\lceil(1-\alpha) \frac{n}{m!}\rceil$ times; such a subprofile exists as $R\in T^\alpha$. By \Cref{lem:distuni}, we hence have that $\mathit{dist}(f_{RD}(R), R')\leq 2+\frac{1}{m-1}-\frac{m}{m-1}\cdot (1-\alpha) \frac{1}{m}=2+\frac{\alpha}{m-1}$. By \Cref{lem:merge} and the fact that $\mathit{dist}(f_{RD}(R),R)\leq 3$ for all profiles $R$, we furthermore conclude for all $R\in T^\alpha$ that 
	\begin{align*}
		\mathit{dist}(f_{RD}(R),R)&\leq \mathit{dist}(f_{RD}(R), R')+\frac{n-|V_{R'}|}{n}(1+\mathit{dist}(f_{RD}(R), R))\\
		&\leq 2+\frac{\alpha}{m-1}+4\frac{n-m!\lceil (1-\alpha)\frac{n}{m!}\rceil}{n}\\
		&\leq 2+\frac{\alpha}{m-1} + 4\alpha\\
		&=2+\frac{1}{\sqrt[3]{n}}(4+\frac{1}{m-1}).
	\end{align*}
	
	Since $2+\frac{1}{\sqrt[3]{n}}(4+\frac{1}{m-1})\leq 3$ for sufficiently large $n$, we can now compute $\mathbb{E}[\mathit{dist}(f(R),R)]$: 
	\begin{align*}
		&\mathbb{E}[\mathit{dist}(f_{RD}(R),R)|]
		\leq 3m!e^{-\frac{\sqrt[3]{n}}{2m!}} + (1-m!e^{-\frac{\sqrt[3]{n}}{2m!}})(2+\frac{1}{\sqrt[3]{n}}(4+\frac{1}{m-1})). 
	\end{align*}
	
	Finally, it is easy to check that this bound indeed converges to $2$ as $n$ goes to infinity, thus demonstrating that $\limsup_{n\rightarrow\infty} \mathbb{E}_{R\sim IC(m,n)}[\mathit{dist}(f_{\mathit{RD}}(R),R)]\leq 2$. 
	\medskip
		
		\textbf{Claim (2): Randomized Plurality-Veto.}\\
		For the randomized Plurality-Veto rule, we first note that $\limsup_{n\rightarrow\infty}\mathbb{E}_{R\sim IC(m,n)}[\mathit{dist}(f_\mathit{PRV}(R),R)]\leq 2+\frac{1}{m-1}$ by \Cref{thm:ML}. We thus only need to improve the lower bound. To this end, we will prove that $\liminf_{n\rightarrow\infty} \mathbb{P}[\exists x\in X\colon f_\mathit{PRV}(R,x)=x]=1$ as the Claim (2) in \Cref{thm:ML} then shows that $\liminf_{n\rightarrow\infty}\mathbb{E}_{R\sim IC(m,n)}[\mathit{dist}(f_\mathit{PRV}(R),R)]\geq 2+\frac{1}{m-1}$.
		
		Since $f_\mathit{PRV}$ randomizes over the set of Plurality-Veto winners, we start our proof by reasoning about when every alternatives is in the set of Plurality-Veto winners $\mathit{PV}(R)$. To this end, we recall that $t_R(x)$ is the number of voters that top-rank alternative $x$ in the profile $R$, and we define $b_R(x)$ as the number of voters that bottom-rank alternative $x$ in $R$. We then claim that $\mathit{PV}(R)=X_R$ holds if and only if $b_R(x)=t_R(x)>0$ for all $x\in X$. We start by showing the direction from left to right and assume for contradiction that there is a profile $R$ such that $\mathit{PV}(R)=X_R$ even though there is an alternative $x$ such that $b_R(x)\neq t_R(x)$. (We note that, if $t_R(x)=0$, then it follows trivially that $x\not\in \mathit{PV}(R)$). By a simple counting argument, it follows that if $b_R(x)<t_R(x)$ for some $x\in X_R$, then there is also an alternative $y$ such that $b_R(y)>t_R(y)$. Because of this, we will assume that $b_R(x)>t_R(x)$. Now, in order for $x$ to be chosen by Plurality-Veto, the plurality scores $t_R(y)$ of every other alternative $y\neq x$ must be decreased to $0$. This requires at least $\sum_{y\in X_R\setminus \{x\}} t_R(y)=|V_R|-t_R(x)$ voters as every voter only decreases the score of an alternative by $1$. However, the voters who bottom-rank $x$ will not reduce the score of any other alternative if $x$ is the winner. Hence, there are only $|V_R|-b_R(x)<|V_R|-t_R(x)$ voters that can decrease the scores of the alternatives in $X_R\setminus \{x\}$. This contradicts that all these scores hit $0$, so the assumption that $x\in\mathit{PV}(R)$ is wrong. For the other direction, we assume that $t_R(x)=b_R(x)>0$ for all $x\in X_R$. Then, it is easy to see that $x\in \mathit{PV}(R)$ by ordering the voters by their bottom-ranked alternatives and ensuring that those who bottom-rank $x$ are last in our sequence. This completes the proof of our auxiliary claim  
		
		By this insight, it holds for the random variable $R\sim IC(m,n)$ (for some fixed $m$ and $n$), and an arbitrary alternative $x^*$ that 
		\begin{align*}
			\mathbb{P}[\exists x\in X\colon f_\mathit{PRV}(R,x)=0]&=1-\mathbb{P}[\forall x\in X\colon f_\mathit{PRV}(R,x)>0]\\
			&=1-\mathbb{P}[\forall x\in X_R\colon t_R(x)=b_R(x)>0]\\
			&\geq 1-\mathbb{P}[t_R(x^*)=b_R(x^*)].
		\end{align*}

		Now, define $Z(R)=t_R(x^*)-b_R(x^*)$ and define $r_i$ as the random variable such that $r_i(R)=1$ if $x^*$ is the favorite alternative of voter $i$ in $R$, $r_i(R)=-1$ if $x^*$ is the least preferred alternative of voter $i$ in $R$, and $r_i(R)=0$ otherwise. It is straightforward to verify that $Z(R)=\sum_{i\in V_R} r_i(R)$. Moreover, by the definition of the impartial culture distribution, all random variables $r_i(R)$ are independent of each other (because the preference relation of each voter is sampled independently from that of all other voters). We can now apply the central limit theorem to infer that $\mathbb{P}[Z(R)=x]$ converges to a normal distribution with mean $0$ (as the mean of $Z(R)$ is easily shown to be $0$). In principle, this is sufficient to infer that $\mathbb{P}[t_R(x^*)=b_R(x^*)]$ converges to $0$ as $n$ goes to $\infty$. We will, however, make this more precise by using the Berry-Esseen theorem \citep{Berr41a,Esse42a}, which bounds the error of the central limit theorem for a finite $n$. To apply this result, we first observe that $\mathbb{E}[r_i(R)]=0$, $\sigma=\sqrt{\mathbb{E}[r_i(R)^2]}=\sqrt{\frac{2}{m}}$, and $\varphi=\mathbb{E}[|t_i(R)|^3]\leq 1$. Now, let $\Pi(x)=\frac{1}{\sqrt{2\pi}}\int_{-\infty}^x e^{-t^2/2} dt$ denote the cumulative distribution  of the standard normal distribution. The Berry-Esseen theorem \citep{Berr41a,Esse42a} then states that there is a constant $C\leq 1$ such that for every $n$ and every $x$
		
		\begin{align*}
			|\mathbb{P}[\frac{Z(R)}{\sqrt{2n/m}}\leq x]-\Phi(x)|\leq \frac{C \varphi}{\sigma^3 \sqrt n}.
		\end{align*}
		
		We will use this bound for $x=0$, for which $\mathbb{P}[\frac{Z(R)}{\sqrt{2n/m}}\leq 0]=\mathbb{P}[Z(R)\leq 0]$ and $\Phi(0)=\frac{1}{2}$. Moreover, by defining $C'=\frac{C\varphi}{\sigma^3}$, we infer from this inequality that $\mathbb{P}[Z(R)\leq 0]\leq \frac{1}{2}+\frac{C'}{\sqrt n}$. Finally, we note that $\mathbb{P}[Z(R)\leq 0]=\mathbb{P}[Z(R)\geq 0]$ because if $Z(R)\leq 0$ for some profile $R$, then $Z(R')\geq 0$ for the profile $R'$ where we reversed the preferences of all voters. This means that 
		\begin{align*}
			\mathbb{P}[Z(R)=0]&=\mathbb{P}[Z(R)\leq 0]+\mathbb{P}[Z(R)\geq 0]-1
			=2\mathbb{P}[Z(R)\leq 0]-1
			\leq \frac{2C'}{\sqrt n}.
		\end{align*}
		
		Combining this with our first inequality shows that, for every $n\in\mathbb{N}$, that
		
		\begin{align*}
			\mathbb{P}[\exists x\in X\colon f_\mathit{PRV}(R,x)=0]\geq 1-\frac{2C'}{\sqrt n}.
		\end{align*}
		
		It is now straightforward to see that the right side converges to $1$, which finally proves our claim. 
	\end{proof}

\end{document}